\documentclass[11pt,a4paper]{article}
\usepackage{amssymb}
\usepackage{amsmath}
\usepackage{amsfonts}
\usepackage{bbm}
\usepackage{amsthm}
\usepackage{mathrsfs}
\usepackage{hyperref}
\usepackage{marvosym}
\usepackage{color}
\usepackage[margin=2.6cm]{geometry}
\usepackage[all,cmtip]{xy}
\usepackage[utf8]{inputenc}

\definecolor{darkred}{rgb}{0.8,0.1,0.1}
\hypersetup{
     colorlinks=false,         % false: boxed links; true: colored links,false is default
     linkcolor=darkred,
     citecolor=blue,
}

\theoremstyle{plain}
\newtheorem{theo}{Theorem}[section]
\newtheorem{lem}[theo]{Lemma}
\newtheorem{propo}[theo]{Proposition}
\newtheorem{cor}[theo]{Corollary}

\theoremstyle{definition}
\newtheorem{defi}[theo]{Definition}

\newtheorem{rem}[theo]{Remark}

\numberwithin{equation}{section}

\def\nn{\nonumber}

\def\bbK{\mathbb{K}}
\def\bbR{\mathbb{R}}
\def\bbC{\mathbb{C}}
\def\bbN{\mathbb{N}}
\def\bbZ{\mathbb{Z}}

\def\Ber{\mathrm{Ber}}

\def\OO{\mathscr{O}}
\def\eOO{\mathfrak{e}\mathscr{O}}

\def\bbM{\boldsymbol{M}}
\def\bbV{\boldsymbol{V}}

\def\LL{\mathfrak{L}}
\def\eLL{\mathfrak{eL}}
\def\AA{\mathfrak{A}}
\def\eAA{\mathfrak{eA}}

\def\QQ{\mathfrak{Q}}
\def\eQQ{\mathfrak{eQ}}

\def\ii{{\,{\rm i}\,}}

\def\op{\mathrm{op}}
\def\pt{\mathrm{pt}}
\def\res{\mathrm{res}}
\def\ext{\mathrm{ext}}

\def\Hom{\mathrm{Hom}}

\def\Der{\mathrm{Der}}

\def\id{\mathrm{id}}
\def\supp{\mathrm{supp}}

\def\dd{\mathrm{d}}
\def\vol{\mathrm{vol}}
\def\sc{\mathrm{sc}}
\def\cc{\mathrm{c}}

\def\dim{\mathrm{dim}}
\def\1{\mathbbm{1}}
\def\oone{\mathbf{1}}

\def\SCart{\mathsf{SCart}}
\def\tLor{\mathsf{otLor}}
\def\SMan{\mathsf{SMan}}

\def\ghSCart{\mathsf{ghSCart}}

\def\SLoc{\mathsf{SLoc}}
\def\eSLoc{\mathsf{eSLoc}}

\def\SAlg{\mathsf{S}^\ast\mathsf{Alg}}
\def\eSAlg{\mathsf{eS}^\ast\mathsf{Alg}}
\def\SVec{\mathsf{SVec}}
\def\eSVec{\mathsf{eSVec}}
\def\XX{\mathsf{X}}
\def\eXX{\mathsf{eX}}
\def\SSet{\mathsf{SSet}}
\def\Set{\mathsf{Set}}
\def\SPt{\mathsf{SPt}}

\def\Mod{\mathsf{SMod}}

\def\Gr{\mathsf{Gr}}

\newcommand{\ip}[2]{\left\langle #1,#2 \right\rangle}

\newcommand{\und}[1]{\widetilde{#1}}

\def\sk{\vspace{2mm}}

\title{%
Supergeometry in locally covariant quantum field theory
}

\author{%
Thomas-Paul Hack$^{1,2,a}$, Florian Hanisch$^{3,b}$ and Alexander Schenkel$^{4,c}$ \vspace{4mm}\\
{\small $^1$ Dipartimento di Matematica, Universit{\`a} degli Studi di Genova,}\\
{\small Via Dodecaneso 35, 16146 Genova, Italy.}\vspace{2mm}\\
{\small $^2$ Institut f\"ur Theoretische Physik, Universit\"at Leipzig,}\\
{\small Br\"uderstra\ss e 16, 04103 Leipzig, Germany.}\vspace{2mm}\\
{\small $^3$ Institut f\"ur Mathematik, Universit\"at Potsdam,}\\
{\small Karl-Liebknecht-Stra\ss e 24-25, 14476 Golm (Potsdam), Germany.}\vspace{2mm}\\
{\small $^4$ Department of Mathematics, Heriot-Watt University,}\\
{\small Colin Maclaurin Building, Riccarton, Edinburgh EH14 4AS, United Kingdom.}\vspace{4mm}\\
 {\footnotesize \texttt{Email:} $^a$ \texttt{hack@dima.unige.it} ~,~~ $^b$ \texttt{fhanisch@math.uni-potsdam.de} ~,~~$^c$  \texttt{as880@hw.ac.uk} }
 }

\date{September 2015}

\begin{document}

\maketitle

\begin{abstract}
In this paper we analyze supergeometric locally covariant quantum field theories. We develop suitable categories $\SLoc$ of super-Cartan supermanifolds, which generalize Lorentz manifolds in ordinary quantum field theory, and show that, starting from a few representation theoretic and geometric data, one can construct a functor $\AA : \SLoc\to \SAlg$ to the category of super-$\ast$-algebras which can be interpreted as a non-interacting super-quantum field theory. This construction turns out to disregard supersymmetry transformations as the morphism sets in the above categories are too small. We then solve this problem by using techniques from enriched category theory, which allows us to replace the morphism sets by suitable morphism supersets that contain supersymmetry transformations as their higher superpoints. We construct super-quantum field theories in terms of enriched functors $\eAA : \eSLoc\to \eSAlg$ between the enriched categories and show that supersymmetry transformations are appropriately described within the enriched framework. As examples we analyze the superparticle in $1\vert 1$-dimensions and the free Wess-Zumino model in $3\vert 2$-dimensions.
\end{abstract}

\paragraph*{Report no.:} EMPG-15-01

%%%%%%%%%%%%%%%%%%%%%%%%%%%%%%%%%%%%%%%%%%%%%%%%%%%%%%%
%%%%%%%%%%%%%%%%%%%%%%%%%%%%%%%%%%%%%%%%%%%%%%%%%%%%%%%
\newpage

{\baselineskip=12pt
\tableofcontents
}

\bigskip

\section{\label{sec:intro}Introduction and summary}
Over the past decades, supersymmetry and supergravity have been strongly 
vital research areas in theoretical and mathematical physics. On the one hand, supersymmetric
extensions of the standard model provide interesting perspectives on particle
physics and, on the other hand, supergravity arises as a low-energy limit of string theory
and it might have potential applications to e.g.\ early universe cosmology. 
Regarded from the perspective of a quantum field theorist, interest in supersymmetry
arises because of the well-known fact that certain supersymmetric quantum field theories
enjoy unexpected renormalization properties, which collectively go under the name
`non-renormalization theorems', see e.g.\ \cite{Grisaru:1979wc,Seiberg:1993vc}.
\sk

In contrast to the immense progress which theoretical physics has made during the
past decades, mathematically rigorous developments of supersymmetric 
quantum field theories are quite rare. There are however some notable exceptions:
In \cite{Buchholz:2006pq}, Buchholz and Grundling address the nontrivial problems of 
implementing supersymmetry transformations 
into the $C^\ast$-algebraic framework of algebraic quantum field theory
and constructing super-KMS states. The study of superconformal nets
in two spacetime dimensions has been initiated by Capri, Kawahigashi and 
Longo in \cite{Carpi:2007fj}. Since then superconformal nets have been intensively developed, also
with a focus on extended supersymmetry \cite{Carpi:2012va}.
Perturbative superconformal quantum field theories on a special class of (curved) spacetimes
have been discussed quite recently by de Medeiros and Hollands \cite{deMedeiros:2013mca},
where also a perturbative non-renormalization theorem is rigorously proven.
A formulation of (Euclidean) supersymmetric quantum field theories
within the Atiyah-Segal approach and their connection to elliptic cohomology
has been investigated by Stolz and Teichner, see e.g.\ the survey article
\cite{StolzTeichner}.
\sk
 
In our work we shall study supersymmetric quantum field theories from the perspective of
locally covariant quantum field theory \cite{Brunetti:2001dx}, which is a relatively modern 
extension of algebraic quantum field theory to curved spacetimes. 
In locally covariant quantum field theory, the focus is on the construction and analysis of functors from 
a category of spacetimes to a category of algebras, which are supposed to describe the assignment of
observable algebras to spacetimes. Besides establishing a mathematical foundation for 
quantum field theory on curved spacetimes, locally covariant quantum field theory is
essential for constructing perturbatively interacting models
\cite{Hollands:2001nf,Hollands:2001fb,Brunetti:2009qc}.
Our aim is to extend carefully the formalism of locally covariant 
quantum field theory to the realm of supergeometry,
focusing in the present work only on the case of non-interacting models.
On the one hand, a solid understanding of non-interacting super-quantum field theories
(super-QFTs) is a necessary prerequisite for constructing perturbative models
and analyzing their renormalization behavior, especially concerning potential non-renormalization theorems.
On the other hand, already simple examples of non-interacting super-QFTs indicate
that the basic framework of locally covariant quantum field theory has to be generalized
in order to be able to cope with the concept of supersymmetry transformations.
In more detail, as we will show in this work, the framework of ordinary category theory, on which locally covariant
quantum field theory is based, is insufficient to capture supersymmetry transformations
on the level of the super-QFT functor. In particular, we observe that both the fermionic and 
the bosonic component fields are locally covariant quantum fields in the sense of \cite{Brunetti:2001dx} (i.e.\ natural transformations
to the super-QFT functor), which indicates that supersymmetry transformations are not
appropriately described in this framework.
Using techniques from {\it enriched category theory}, we shall propose a generalization
of the framework in \cite{Brunetti:2001dx} which is general enough to capture supersymmetry transformations.
Loosely speaking, we  shall develop suitable categories of superspacetimes $\eSLoc$ and superalgebras $\eSAlg$
that are enriched over the monoidal category of supersets and consider super-QFTs as enriched 
functors $\eAA : \eSLoc\to\eSAlg$ between these enriched categories.
Supersymmetry transformations are captured in terms of the higher superpoints of the morphism
supersets in $\eSLoc$ and their action on the superalgebras of observables is dictated by the enriched functor
$\eAA : \eSLoc\to\eSAlg$. Our results will therefore clarify the structure of
supersymmetry transformations in locally covariant quantum field theory,
which will be essential for analyzing perturbative super-QFTs and their 
renormalization behavior in future works.
\sk

Let us outline the content of this paper: 
In Section \ref{sec:prelim} we shall give a self-contained and rather 
detailed introduction to those techniques of super-linear algebra and supergeometry
that will be used in our work. This should allow readers who do not have a solid 
background in those fields to follow our constructions in the main part of this paper. 
In Section \ref{sec:cartan} we introduce super-Cartan structures on supermanifolds
and study their properties. These structures have their origin in the superspace formulation
of supergravity \cite{Wess:1977fn} and they are used in our work in order to describe
`superspacetimes', which generalize Lorentz manifolds in ordinary quantum field theory.
The dimension and the `amount of supersymmetry' of a super-Cartan supermanifold
is captured in its local model space, which we describe by using representation theoretic data
corresponding to some spin group. We define a suitable category of super-Cartan supermanifolds
and show that to any super-Cartan supermanifold there is a functorially associated
oriented and time-oriented ordinary Lorentz manifold. This allows us to introduce
a natural notion of the chronological and causal future/past in a super-Cartan supermanifold
and therewith the concept of globally hyperbolic super-Cartan supermanifolds.
In Section \ref{sec:axioms} we formulate a set of axioms to describe non-interacting 
super-field theories in an abstract way. According to our Definition \ref{defi:sft}, a super-field theory
is specified by the following data: 1.)\ A choice of representation theoretic data that fixes
the local model space of the super-Cartan supermanifolds. 2.)\ A full subcategory $\SLoc$ 
of the category of globally hyperbolic super-Cartan supermanifolds, which allows us later
to implement constraints on the super-Cartan structures, e.g.\ the supergravity  
supertorsion constraints \cite{Wess:1977fn}.  3.)\ A suitable natural super-differential operator
which encodes the dynamics of the super-field theory.
We show in Section \ref{sec:ordinary} that given any super-field theory as described above,
one can construct a functor $\AA : \SLoc \to \SAlg$ to the category of super-$\ast$-algebras
which satisfies the axioms of locally covariant quantum field theory \cite{Brunetti:2001dx}
adapted to our supergeometric setting, cf.\ Theorem \ref{theo:LCQFT}. In other words, any super-field theory gives rise
to a super-QFT. As in the case of ordinary quantum field theory, we first construct
a functor $\LL : \SLoc \to \XX$ to the category of super-symplectic spaces
or the category of super-inner product spaces (depending on the representation theoretic data),
which is then quantized by a quantization functor $\QQ :\XX\to\SAlg$ that constructs
super-canonical (anti)commutation relation algebras. 
We analyze the functor $\AA :\SLoc\to\SAlg$ and show that, in addition
to the locally covariant quantum field which describes the linear superfield operators,
the bosonic and fermionic component fields are also natural transformations in this framework.
This is an undesirable feature which indicates that the framework developed in
the Sections \ref{sec:axioms} and \ref{sec:ordinary} does not capture supersymmetry transformations
as those would mix the bosonic and fermionic components. We then solve this problem
by making use of techniques from enriched category theory. In Section \ref{sec:axiomsenriched}
we provide a stronger axiomatic framework for super-field theories by generalizing
the category $\SLoc$ to a suitable category  $\eSLoc$ which is enriched over the monoidal  category 
of supersets $\SSet$. The category $\SSet$ is defined as the functor category $\mathrm{Fun}(\SPt^\op,\Set)$,
where $\SPt$ is the category of superpoints and $^{\op}$ denotes the opposite category.
Hence, a superset is a functor $\SPt^\op\to \Set$, which means that, in addition to its ordinary
points, a superset has further content that is captured by its `higher superpoints'.
Loosely speaking, enriching the morphism sets in $\SLoc$ to the morphism
supersets in $\eSLoc$ we obtain in addition to ordinary supermanifold morphisms
$M\to M^\prime$ also supermanifold morphisms between the `fattened' supermanifolds
$\pt_n\times M\to\pt_n\times M^\prime$, where $\pt_n$ is any superpoint, that are able
to capture supersymmetry transformations; indeed, the odd parameters which are used in
the physics literature in order to parametrize supersymmetry transformations
are elements in the structure sheaf $\Lambda_n = \bigwedge^\bullet \bbR^n$ (the Grassmann algebra) of $\pt_n$.
It is important to notice that in this functorial approach we  {\it do not} have to fix a superpoint (or equivalently a Grassmann algebra)
from the outside, as it is typically done in the physics literature, but we are working functorially over the category
of {\it all} superpoints. Similar techniques have been used before in order to describe super-mapping 
spaces between supermanifolds, see e.g.\ \cite{SachseDiss,SachseWockel,FH}.
In the enriched setting, the super-differential operators which govern the dynamics of the
super-field theory should form an enriched natural transformation. We explicitly characterize these enriched natural
transformations and show that they are in bijective correspondence to ordinary natural transformations
(as used in Section \ref{sec:axioms}) satisfying further conditions, which one may interpret as covariance
conditions under supersymmetry transformations. This allows us to give a simple
axiomatic characterization of enriched super-field theories in Definition \ref{defi:sftenriched}.
In Section \ref{sec:enriched} we show that any enriched super-field theory gives
rise to an enriched super-QFT that we describe by an enriched functor $\eAA : \eSLoc\to\eSAlg$
to a suitable enriched category of super-$\ast$-algebras.
We show that this enriched functor satisfies a generalization of the axioms of locally covariant 
quantum field theory, cf.\ Theorem \ref{theo:enLCQFT}. We further show that
the enriched super-QFT has an enriched locally covariant quantum field (given by an enriched natural transformation)
which describes the linear superfield operators. In contrast to the non-enriched theory studied in Section
\ref{sec:ordinary}, our enriched natural transformation does not decompose into the bosonic and fermionic 
component fields, which indicates that supersymmetry transformations are appropriately 
described within our enriched categorical framework. This is confirmed and illustrated in Section \ref{sec:examples} by
constructing and analyzing explicit examples of $1\vert 1$ and $3\vert 2$-dimensional enriched super-QFTs, together
with the structure of supersymmetry transformations. Our $1\vert 1$-dimensional example is the usual superparticle
and our $3\vert 2$-dimensional example is the free Wess-Zumino model on a class of curved super-Cartan supermanifolds. 
In Appendix \ref{app:enriched} we collect some elementary definitions from enriched category theory which are needed in
our work.

%%%%%%%%%%%%%%%%%%%%%%%%%%%%%%%%%%%%%%%%%%%%%%%%
%%%%%%%%%%%%%%%%%%%%%%%%%%%%%%%%%%%%%%%%%%%%%%%%

\section{\label{sec:prelim}Preliminaries on supergeometry}
We give a self-contained review of those aspects of super-linear algebra and supergeometry
which we shall need for our work. For more details see e.g.\ 
\cite{Carmeli} and \cite{DMSuper}. In the following the ground field $\bbK$ will be 
either $\bbR$ or $\bbC$ and we set $\bbZ_2 :=\{0,1\}$. Whenever there is no need to distinguish between
the real and complex case, we shall drop the field $\bbK$ from our notations.

\paragraph{Super-vector spaces:}
A {\em super-vector space} is a $\bbZ_2$-graded vector space 
$V = V_{0} \oplus V_{1}$.
We assign to the non-zero homogeneous elements $0\neq v\in V_i$ the $\bbZ_2$-parity $\vert v\vert :=i\in\bbZ_2$, for
$i=0,1$, and call elements in $V_0$ even and elements in $V_1$ odd.
The superdimension (or simply dimension) of a super-vector space $V$ 
is denoted by $\dim(V):= \dim(V_0) \vert \dim(V_1)$.
An example of an $n\vert m$-dimensional super-vector space 
 is $\bbK^{n\vert m} := \bbK^n \oplus \bbK^m$, with $n,m\in\bbN^0$. 
For simplicity, we shall denote $\bbK^{1\vert 0}$ simply by $\bbK$.
A {\em super-vector space morphism} $L : V\to V^\prime$ 
is a linear map which preserves the $\bbZ_2$-parity,
i.e.\ $L(V_i)\subseteq V_i^\prime$ for $i=0,1$. 
\sk

The category $\SVec$ of super-vector spaces
has as objects all super-vector spaces  and as morphisms
all super-vector space morphisms.
Recall that  $\SVec$ is a monoidal category
with tensor product functor $\otimes: \SVec \times \SVec \to \SVec$
and unit object $\bbK = \bbK^{1\vert 0}$.
Explicitly, the tensor product $V\otimes W$ 
of two super-vector spaces $V$ and $W$ is the ordinary tensor product $V\otimes W$  of vector spaces equipped
with the $\bbZ_2$-grading
\begin{subequations}
\begin{flalign}
(V\otimes W)_0 &:= (V_0\otimes W_0) \oplus (V_1\otimes W_1)~,\\
(V\otimes W)_1 &:= (V_0 \otimes  W_1) \oplus (V_1\otimes W_0)~.
\end{flalign}
\end{subequations}
The tensor product of two $\SVec$-morphisms is simply given by the tensor product of linear maps.
The monoidal category $\SVec$ is symmetric with respect to the commutativity constraints
\begin{flalign}
\sigma_{V,W} : V\otimes W \longrightarrow W\otimes V~,~~v\otimes w\longmapsto (-1)^{\vert v\vert\,\vert w\vert}\,w\otimes v~.
\end{flalign}
Moreover, it is closed with internal hom-objects given by
the vector space $\underline{\Hom}(V,W)$ of {\em all} linear maps $L : V\to W$
equipped with the obvious $\bbZ_2$-grading;  $L\in \underline{\Hom}(V,W)$ is even/odd if it preserves/reverses
the $\bbZ_2$-parity.
\sk

Given an object $V = V_0\oplus V_1$ in $\SVec$, a {\em super-vector subspace} is a vector subspace
$W\subseteq V$ together with a $\bbZ_2$-grading $W = W_0\oplus W_1$ such that
$W_i$ is a vector subspace of $V_i$, for $i=0,1$. We then may form the {\em quotient super-vector space}
$V/W := V_0/W_0\oplus V_1/W_1$,
which comes together with a canonical $\SVec$-morphism $V\to V/W$ assigning equivalence classes.

\paragraph{Superalgebras:}
A (unital and associative) {\em superalgebra} is an algebra object in $\SVec$. Explicitly, this means that 
a superalgebra is an object $A$ in $\SVec$ together with two $\SVec$-morphisms
$\mu_A : A\otimes A\to A$ (called product) and $\eta_{A} : \bbK \to A$ (called unit),
such that the diagrams
\begin{flalign}
\xymatrix{
\ar[d]_-{\id_A\otimes \mu_A}A\otimes A\otimes A \ar[rr]^-{\mu_A\otimes \id_A} && A\otimes A\ar[d]^-{\mu_A} &&\ar[drr]_-{\simeq}\bbK \otimes  A\ar[rr]^-{\eta_A\otimes \id_{A}} && \ar[d]_-{\mu_A}A\otimes  A && \ar[ll]_-{\id_A\otimes \eta_A}A\otimes \bbK \ar[dll]^-{\simeq}\\
A\otimes A \ar[rr]_-{\mu_A}&& A && &&A&&
}
\end{flalign}
in $\SVec$ commute. We shall often denote the products by juxtaposition, i.e.\ $\mu_A(a_1\otimes a_2) = a_1\,a_2$,
and the unit element by $ \eta_A(1) =\1$.
An example of a (real) superalgebra is the Grassmann algebra $\Lambda_n 
:= \bigwedge^\bullet\bbR^n$, for $n\in\bbN^0$, with product given by the wedge product
and unit element by $\1 = 1\in \bbR = \bigwedge^0\bbR^n\subseteq \bigwedge^\bullet\bbR^n$.
A {\em superalgebra morphism} $\kappa : A\to A^\prime$ 
is a $\SVec$-morphism which preserves products and units, i.e.\ $\mu_{A^\prime} \circ (\kappa \otimes \kappa) 
= \kappa\circ \mu_A$ and $\eta_{A^\prime} = \kappa \circ \eta_{A}$.
\sk

We denote the category of superalgebras  by $\mathsf{SAlg}$ and notice that
it is a monoidal category: The tensor product $A\otimes B$ of two superalgebras
is the super-vector space $A\otimes B$ equipped with the following product and unit
\begin{subequations}
\begin{flalign}
\mu_{A\otimes B} := (\mu_A\otimes \mu_B)\circ (\id_{A}\otimes \sigma_{B,A}\otimes \id_B) & : 
A\otimes B \otimes  A\otimes  B \longrightarrow A\otimes  B~,\\
\eta_{A\otimes B} := \eta_A\otimes \eta_B &: \bbK \otimes \bbK \simeq \bbK \longrightarrow A\otimes B~.
\end{flalign}
\end{subequations}
Explicitly, we have for the product $(a_1\otimes b_1)\,(a_2\otimes b_2) = 
(-1)^{\vert a_2\vert\,\vert b_1\vert}\,(a_1\,a_2)\otimes (b_1\,b_2)$, for all homogeneous
$a_1,a_2\in A$ and $b_1,b_2\in B$,
and for the unit $\1_{A\otimes B} = \1_A\otimes \1_B$.
The tensor product of two $\mathsf{SAlg}$-morphisms is simply given by the tensor product of linear maps.
\sk

We shall require some special classes of superalgebras.
A superalgebra $A$ is called {\em supercommutative} if the product is compatible with the commutativity constraint,
i.e.\ $\mu_A\circ \sigma_{A,A} = \mu_A$. Notice that supercommutative superalgebras form a monoidal subcategory 
of $\mathsf{SAlg}$, which is symmetric with respect to the commutativity constraints induced by $\SVec$.
Moreover, for a supercommutative superalgebra $A$ the product
$\mu_A : A\otimes  A\to A$ is a $\mathsf{SAlg}$-morphism with respect to the tensor product superalgebra structure
on $A\otimes A$.
\sk

Let us now consider superalgebras over $\bbC$. A {\em super-$\ast$-algebra}
is a superalgebra $A$ over $\bbC$ together with an even $\bbC$-antilinear 
map $\ast_A : A\to A $  (called superinvolution) which satisfies
$\ast_A\circ \eta_A = \eta_A$ and $\ast_A \circ \mu_A = \mu_A\circ \sigma_{A,A}\circ (\ast_A\otimes_{\bbC} \ast_A)$.
Explicitly, these conditions read as $\1^\ast = \1$ and
$(a_1\,a_2)^\ast = (-1)^{\vert a_1\vert\,\vert a_2\vert}\,a_2^\ast\,a_1^\ast$, for homogeneous elements
$a_1,a_2\in A$.
A {\em super-$\ast$-algebra morphism} $\kappa : A\to A^\prime$ is
a $\mathsf{SAlg}$-morphism satisfying $\kappa\circ\ast_A = \ast_{A^\prime} \circ \kappa$.
We denote the category of super-$\ast$-algebras  by $\SAlg$ and notice that
it is a monoidal category; the superinvolution on the tensor product $A\otimes_{\bbC} B$ of two 
super-$\ast$-algebras is defined component-wise, i.e.\ $(a\otimes_{\bbC} b)^\ast := a^\ast\otimes_{\bbC} b^\ast$.

\paragraph{Super-Lie algebras:}
A {\em super-Lie algebra} is a Lie algebra object in $\SVec$.
Explicitly, a super-Lie algebra  is an object $\mathfrak{g}$ in
$\SVec$ together with a $\SVec$-morphism
$[\,\cdot\,,\,\cdot\,]_{\mathfrak{g}} : \mathfrak{g}\otimes \mathfrak{g} \to\mathfrak{g}$ (called super-Lie bracket)
which satisfies the super-skew symmetry condition
\begin{subequations}
\begin{flalign}
[\,\cdot\,,\,\cdot\,]_{\mathfrak{g}} \circ \big(\id_{\mathfrak{g}\otimes \mathfrak{g}}+ \sigma_{\mathfrak{g},\mathfrak{g}} \big)= 0~
\end{flalign}
and the super-Jacobi identity
\begin{flalign}
[\,\cdot\,,[\,\cdot\,,\,\cdot\,]_{\mathfrak{g}}]_{\mathfrak{g}} \circ \big(\id_{\mathfrak{g}\otimes \mathfrak{g}\otimes \mathfrak{g}} 
+ \sigma_{\mathfrak{g},\mathfrak{g\otimes\mathfrak{g}}}
+\sigma_{\mathfrak{g}\otimes\mathfrak{g},\mathfrak{g}}\big)=0~.
\end{flalign}
\end{subequations}
A {\em super-Lie algebra morphism} $L : \mathfrak{g}\to\mathfrak{g}^\prime$ 
is a $\SVec$-morphism  which preserves the super-Lie brackets, i.e.\
$[\,\cdot\,,\,\cdot\,]_{\mathfrak{g}^\prime} \circ ( L\otimes L) = L\circ [\,\cdot\,,\,\cdot\,]_{\mathfrak{g}}$.

\paragraph{Supermodules and the Berezinian:}
Let $A$ be a superalgebra. A {\em left $A$-supermodule} is a left module object in $\SVec$. Explicitly,
a left $A$-supermodule is an object $V$ in $\SVec$ together with a $\SVec$-morphism
$l_{V} : A\otimes V\to V$ (called left $A$-action), such that the diagrams
\begin{flalign}
\xymatrix{
\ar[d]_-{\mu_A\otimes \id_V} A\otimes A\otimes V\ar[rr]^-{\id_{A}\otimes l_{V}} && A\otimes V\ar[d]^-{l_V} && \ar[drr]_-{\simeq}\bbK\otimes V \ar[rr]^-{\eta_A\otimes\id_V}&& A\otimes V\ar[d]^-{l_V}\\
A\otimes V \ar[rr]_-{l_{V}} && V && && V
}
\end{flalign}
in $\SVec$ commute. A {\em right $A$-supermodule} is defined similarly and an {\em $A$-bisupermodule}
is a left and right $A$-supermodule with commuting left and right $A$-actions.
If $A$ is a supercommutative superalgebra, then any left $A$-supermodule $V$ is also 
a right $A$-supermodule with right $A$-action $r_V := l_V\circ \sigma_{V,A} : V\otimes A\to V$
and, vice versa, any right $A$-supermodule is also a left $A$-supermodule with left $A$-action
$l_V := r_V\circ \sigma_{A,V}:A\otimes V\to V$. Notice that these left
and right $A$-actions are compatible, hence $V$ is an $A$-bisupermodule. 
We shall often denote the left and right $A$-actions simply by juxtaposition, i.e.\
$l_V(a\otimes v) = a\,v$ and $r_V(v\otimes a) = v\,a$.
A {\em left $A$-supermodule morphism} $L : V\to V^\prime$ is a $\SVec$-morphism
which preserves the left $A$-actions, i.e.\ $l_{V^\prime} \circ (\id_A\otimes L) = L\circ l_V$.
We denote the category of left $A$-supermodules by $A\text{-}\Mod$. In the case of $A$ being supercommutative, 
$A\text{-}\Mod$ is a monoidal category with tensor product functor 
$\otimes_A : A\text{-}\Mod\times A\text{-}\Mod\to A\text{-}\Mod$ (taking tensor products over $A$) 
and unit object $A$ (regarded as a left $A$-supermodule with left $A$-action given by the product $\mu_A$).
Again for $A$ being supercommutative, the monoidal category $A\text{-}\Mod$ is also symmetric
with commutativity constraints induced by those in $\SVec$ and
closed with internal hom-objects given by the left $A$-supermodules $\underline{\Hom}_{A}(V,W)$
of {\em all} right $A$-linear maps $L : V \to W$ equipped with the obvious $\bbZ_2$-grading;
$L\in \underline{\Hom}_A(V,W)$ is even/odd if it preserves/reverses the $\bbZ_2$-parity.
\sk

A {\em free left $A$-supermodule} of dimension $n\vert m$
is a left $A$-supermodule $V$ for which there exists a basis of $n\in\bbN^0$ even 
elements $\{e_1,\dots, e_n\}$ and $m\in\bbN^0$ odd
elements $\{\epsilon_1,\dots,\epsilon_m\}$, such that
\begin{subequations}
\begin{flalign}
V_0 &= \mathrm{span}_{A_0}\{e_1,\dots,e_n\} \oplus \mathrm{span}_{A_1}\{\epsilon_1,\dots,\epsilon_m\}~,\\
V_1 &= \mathrm{span}_{A_1}\{e_1,\dots,e_n\} \oplus \mathrm{span}_{A_0}\{\epsilon_1,\dots,\epsilon_m\}~.
\end{flalign}
\end{subequations}
The collection $\{e_1,\dots,e_{n+m}\} := \{e_1,\dots,e_n,\epsilon_1,\dots,\epsilon_m\}$ of elements in $V$ is called
an adapted basis for $V$. Notice that any free left $A$-supermodule of dimension $n\vert m$ 
is isomorphic (in the category $A\text{-}\Mod$) to the standard free 
left $A$-supermodule $A^{n\vert m} := A\otimes \bbK^{n\vert m}$ with the obvious left $A$-action.
The $A\text{-}\Mod$-morphisms between two free left $A$-supermodules can be represented in terms of 
matrices with entries in $A$. Explicitly, let $L : V\to V^\prime$ be any $A\text{-}\Mod$-morphism
between an $n\vert m$-dimensional free left $A$-supermodule $V$ and an $n^\prime\vert m^\prime$-dimensional 
free left $A$-supermodule $V^\prime$. Making use of any adapted bases for $V$ and $V^\prime$
we define the elements $\{L_i^j\in A : i=1,\dots,n+m \,,~j=1,\dots,n^\prime+m^\prime\}$
 via $L(e_i) = \sum_{j=1}^{n^\prime + m^\prime} L_i^j \,e^\prime_j$, which can be arranged
 in an $(n+m)\times (n^\prime + m^\prime)$-matrix  of the form
 \begin{flalign}\label{eqn:matrixbasis}
\underline{ \underline{L} } = \begin{pmatrix}
L_1 & L_2\\
L_3 & L_4
 \end{pmatrix}~,
 \end{flalign}
 where $L_1$ is an $n\times n^\prime$-matrix with entries in $A_0$,
 $L_2$ is an $n\times m^\prime$-matrix with entries in $A_1$,
 $L_3$ is an $m\times n^\prime$-matrix with entries in $A_1$ 
 and $L_4$ is an $m\times m^\prime$-matrix with entries in $A_0$.
\sk

Let now $A$ be a supercommutative superalgebra and $V$ any free left $A$-supermodule.
Denoting the group of $A\text{-}\Mod$-automorphisms  of $V$ by $\mathrm{GL}(V)$,
there exists a group homomorphism (called the {\em Berezinian}) to the group of invertible elements in $A$
\begin{flalign}\label{eqn:Berezinian}
\Ber : \mathrm{GL}(V) \longrightarrow A_0^\times ~,~~L \longmapsto \Ber(L) 
=\det(L_1 - L_2 L_4^{-1} L_3)\,\det(L_4)^{-1}~,
\end{flalign}
where we have made use of an arbitrary adapted basis for $V$
(the Berezinian does not depend on the choice of adapted basis). 
One can easily check that the Berezinian is multiplicative, i.e.\ 
$\Ber(L^\prime\,L) = \Ber(L^\prime)\,\Ber(L)$, for all $L,L^\prime\in \mathrm{GL}(V)$.
Moreover, we may assign to any free left $A$-supermodule $V$ of dimension $n\vert m$
its {\em Berezinian left $A$-supermodule} $\mathrm{Ber}(V)$ that is defined as follows:
$\mathrm{Ber}(V)$ is the free left $A$-supermodule that is generated by the elements
$[e_1,\dots,e_{n+m}]$, for all adapted bases $\{e_1,\dots,e_{n+m}\}$ for $V$,
subject to the relations
\begin{flalign}\label{eqn:Berezinianrelations}
[L(e_1),\dots,L(e_{n+m})] = \mathrm{Ber}(L)\,[e_1,\dots, e_{n+m}]~,
\end{flalign}
for all $L\in \mathrm{GL}(V)$. We declare the elements $[e_1,\dots,e_{n+m}]$ to be even if $m\in 2\bbN^0$
or to be odd if $m\in 2\bbN^0 +1$. Since any two adapted bases for $V$ can be related
by a $\mathrm{GL}(V)$-transformation, it is clear that $\mathrm{Ber}(V)$ is a free left $A$-supermodule
of dimension $1\vert 0$ if  $m\in 2\bbN^0$ or dimension $0\vert 1$ if $m\in 2\bbN^0+1$.
Notice that any choice of adapted basis $\{e_1,\dots,e_{n+m}\}$ for $V$ defines an adapted basis
$[e_1,\dots,e_{n+m}]$ for $\mathrm{Ber}(V)$. The assignment of the Berezinian left $A$-supermodules
is functorial: Given any $A\text{-}\Mod$-isomorphism $L : V\to V^\prime$
between two free left $A$-supermodules $V$ and $V^\prime$ 
we define an $A\text{-}\Mod$-isomorphism
$\mathrm{Ber}(L) : \mathrm{Ber}(V) \to\mathrm{Ber}(V^\prime)$ by setting for
any adapted basis $[e_1,\dots,e_{n+m}]$ for $\mathrm{Ber}(V)$
\begin{flalign}
\mathrm{Ber}(L)\big([e_1,\dots,e_{n+m}]\big):= [L(e_1),\dots,L(e_{n+m})]~,
\end{flalign}
and extending $\mathrm{Ber}(L)$ as an $A\text{-}\Mod$-morphism.

\paragraph{Supermanifolds:}
In the following let $\bbK= \bbR$.
A {\em superspace} is a pair $M := (\und{M},\OO_M)$ consisting of
a topological space $\und{M}$ (which we always assume to be second-countable and Hausdorff) 
and a sheaf of supercommutative superalgebras $\OO_M$ on $\und{M}$ (called the structure sheaf).
We denote the restriction $\mathsf{SAlg}$-morphisms in the structure sheaf $\OO_M$ by
$\res_{U,V} : \OO_{M}(U)\to \OO_M(V)$, for all open $V\subseteq U\subseteq\und{M}$,
and the global sections of the structure sheaf simply by $\OO(M) := \OO_{M}(\und{M})$.
A {\em superspace morphism} $\chi : M\to M^\prime$ is a pair $\chi := (\und{\chi} ,\chi^\ast)$
consisting of a continuous map $\und{\chi} : \und{M}\to\und{M^\prime}$ and a morphism
of sheaves of superalgebras $\chi^\ast : \OO_{M^\prime} \to \und{\chi}_\ast (\OO_M)$,
where $\und{\chi}_\ast (\OO_M)$ denotes the direct image sheaf. Explicitly,
$\chi^\ast : \OO_{M^\prime} \to \und{\chi}_\ast (\OO_M)$ consists of assigning
to any open $U\subseteq \und{M^\prime}$ a $\mathsf{SAlg}$-morphism
$\chi^\ast_U :\OO_{M^\prime} (U)\to \OO_{M}(\und{\chi}^{-1}(U))$, such that
the diagram
\begin{flalign}\label{eqn:sheafmorph}
\xymatrix{
\ar[d]_-{\res_{U,V}} \OO_{M^\prime} (U) \ar[rr]^-{\chi^\ast_U}  && \OO_{M}(\und{\chi}^{-1}(U))\ar[d]^-{\res_{\und{\chi}^{-1}(U), \und{\chi}^{-1}(V)} }\\
\OO_{M^\prime} (V) \ar[rr]_-{\chi^\ast_V}  && \OO_{M}(\und{\chi}^{-1}(V))\\
}
\end{flalign}
in $\mathsf{SAlg}$ commutes, for all open $V\subseteq U\subseteq \und{M^\prime}$.
For notational simplicity we shall denote the $\mathsf{SAlg}$-morphism of global sections
by  $\chi^\ast := \chi^\ast_{\und{M^\prime}} : \OO(M^\prime)\to \OO(M)$.
An example of a superspace is $\bbR^{n\vert m} :=(\bbR^n , C^\infty_{\bbR^n} \otimes \bigwedge^\bullet\bbR^m) $,
for $n,m\in\bbN^0$, which we denote with the usual abuse of notation by the same symbol 
as the standard super-vector space above.
\sk

A {\em supermanifold} of dimension $n\vert m$ is a superspace $M = (\und{M},\OO_M)$ that is locally isomorphic
to $\bbR^{n\vert m}$, with  $n,m\in\bbN^0$ fixed. In more detail,
given any $p\in\und{M}$, there exists an open neighborhood $V\subseteq \und{M}$ of $p$,
such that $V$ is homeomorphic to an open subset $U\subseteq \bbR^n$ and such that
the restricted sheaves of superalgebras $\OO_M\vert_{V}$ and
$C^\infty_{\bbR^n} \vert_{U}\otimes \bigwedge^\bullet\bbR^m$ are isomorphic.
Taking the standard coordinates $(x^1,\dots,x^n)$ of $\bbR^n$ and the standard
coordinates $(\theta^1,\dots,\theta^m)$ of $\bbR^m$, the sheaf isomorphism $\OO_M\vert_{V}\simeq 
C^\infty_{\bbR^n} \vert_{U}\otimes \bigwedge^\bullet\bbR^m$ induces local coordinate functions
of $M$ in $V$, which we denote by the same symbols $(x^1,\dots,x^{n+m}) := (x^1,\dots,x^n,\theta^1,\dots,\theta^m)$,
hence suppressing the isomorphism.
Notice that the first $n$ coordinate functions $x^i\in \OO_{M}(V)$, for $i=1,\dots,n$, are even and that the last $m$ coordinate
functions $x^{n+i} \in\OO_{M}(V)$, for $i=1,\dots,m$, are odd.
It is clear that the superspaces $\bbR^{n\vert m} = (\bbR^n,C^\infty_{\bbR^n}\otimes\bigwedge^\bullet\bbR^m)$
 are $n\vert m$-dimensional supermanifolds, for all $n,m\in\bbN^0$.
 Moreover, given any $n\vert m$-dimensional supermanifold $M = (\und{M},\OO_M)$ and any open $U\subseteq \und{M}$, then
$ M\vert_U := (U,\OO_M\vert_U)$
is an $n\vert m$-dimensional supermanifold, which we call the {\em open subsupermanifold determined by $U$}.
A {\em supermanifold morphism} is a superspace morphism between supermanifolds.
 It is well known that supermanifold morphisms $\chi : M\to M^\prime$ are 
 already uniquely specified by their associated  $\mathsf{SAlg}$-morphisms 
 $\chi^\ast :\OO(M^\prime)\to\OO(M)$ on \emph{global} sections of the structure sheaf,
see e.g.\ \cite[Proposition 4.6.1]{Carmeli}.  
\sk

We denote the category of supermanifolds by $\SMan$
and notice that it is a monoidal category with monoidal bifunctor
$\times : \SMan\times \SMan\to \SMan$ and unit object $\pt :=(\{\star\}, \bbR )$.
The monoidal bifunctor is defined as follows, cf.\ \cite[Chapter 4.5]{Carmeli}:
To any pair $(M,M^\prime)$ of objects in $\SMan$ it assigns the supermanifold
$M\times M^\prime := (\und{M}\times\und{M^\prime} , \OO_{M\times M^\prime})$,
where the sheaf of superalgebras $\OO_{M\times M^\prime}$
is defined on the rectangular open subsets $U\times V\subseteq \und{M}\times\und{M^\prime}$
by $\OO_{M\times M^\prime}(U\times V) := \OO_{M}(U)\,\widehat{\otimes}\,\OO_{M^\prime}(V) $
and extended to all open subsets of  $\und{M}\times\und{M^\prime}$ by a standard construction.
Here $\widehat{\otimes}$ denotes the completed tensor product in the projective tensor topology
corresponding to the standard Fr{\'e}chet topologies on $\OO_{M}(U)$ and $\OO_{M^\prime}(V)$.
To any pair $(\chi : M\to N,\chi^\prime  :M^\prime\to N^\prime)$ of $\SMan$-morphisms 
it assigns the $\SMan$-morphism $\chi\times\chi^\prime : M\times M^\prime\to N\times N^\prime$
given by $\und{\chi} \times\und{\chi^\prime} : \und{M}\times \und{M^\prime}\to \und{N}\times \und{N^\prime}$
and $\chi^\ast\,\widehat{\otimes}\,{\chi^\prime}^\ast : \OO_{N\times N^\prime} \to 
(\und{\chi} \times\und{\chi^\prime})_\ast(\OO_{M\times M^\prime})$,
which is uniquely determined by continuous extension of the usual tensor product morphism $\chi^\ast\otimes {\chi^\prime}^\ast$.
\sk

Let us recall that any $n\vert m$-dimensional supermanifold $M$ has an 
associated reduced $n$-dimensional ordinary manifold. Let us denote by $J_M$ the super-ideal sheaf of nilpotents in $\OO_M$, i.e.\
$J_M(U) := \{f\in\OO_M(U) : f\text{ is nilpotent}\}$ for all open $U\subseteq\und{M}$. The {\em reduced manifold}
is defined to be the $n\vert 0$-dimensional supermanifold
$(\und{M},\OO_{M}/J_M)$, which for notational convenience we shall simply denote by $\und{M}$.
For any $\SMan$-morphism $\chi : M\to M^\prime$
the underlying continuous map $\und{\chi} : \und{M}\to \und{M^\prime}$ 
is smooth with respect to the reduced manifold structures on $\und{M}$ and $\und{M^\prime}$.
Hence, the assignment $M\mapsto \und{M}$ and $(\chi : M\to M^\prime) \mapsto (\und{\chi} : \und{M}\to\und{M^\prime})$
is a functor from $\SMan$ to the category of manifolds.
Finally, recall that for any supermanifold $M$ there is a closed embedding
$\iota_{\und{M},M} : \und{M}\to M$ of the reduced manifold (regarded as an $n\vert 0$-dimensional supermanifold)
into the supermanifold $M$. Explicitly, we have that $\und{\iota}_{\und{M},M} = \id_{\und{M}}$ is the identity map
and that $\iota_{\und{M},M}^\ast : \OO_{M}\to \OO_{M}/J_M$ is the projection taking equivalence classes. 
Given any $\SMan$-morphism $\chi : M\to M^\prime$ the diagram
\begin{flalign}\label{eqn:reducedmanifolddiagram}
\xymatrix{
M\ar[rr]^-{\chi} && M^\prime\\
\ar[u]^-{\iota_{\und{M},M}}\und{M} \ar[rr]_-{\und{\chi}} && \und{M^\prime}\ar[u]_-{\iota_{\und{M^\prime},M^\prime}}
}
\end{flalign}
in $\SMan$ commutes.

\paragraph{Superderivations and super-differential operators:}
Let $M = (\und{M},\OO_M)$ be any object in $\SMan$.
For any open $U\subseteq \und{M}$, a homogeneous {\em superderivation} 
(or {\em super-vector field}) on $\OO_{M}(U)$ is a homogeneous 
$X\in\underline{\Hom}(\OO_M(U),\OO_M(U))$, such that
\begin{flalign}
X(f\,g) = X(f)\,g + (-1)^{\vert X\vert \,\vert f\vert}\,f\,X(g)~,
\end{flalign}
for all homogeneous $f,g\in \OO_{M}(U)$.
The collection of all superderivations on $\OO_{M}(U)$ forms a super-vector subspace $\Der_M(U)$
of $\underline{\Hom}(\OO_M(U),\OO_M(U))$, which is further a left $\OO_M(U)$-supermodule via
$(f\,X)(g) := f\,X(g)$,
for all $f,g\in \OO_M(U)$ and $X\in \Der_M(U)$. 
 As in the case of ordinary manifolds, the superderivations $\Der_M(U)$ also form a super-Lie algebra
with super-Lie bracket given by the supercommutator $[X,Y] := X\circ Y - (-1)^{\vert X\vert \,\vert Y\vert} Y\circ X$, 
for all  homogeneous $X,Y\in \Der_M(U)$.
The assignment
$U \mapsto\Der_M(U)$ (together with suitable restriction morphisms) 
is a sheaf of left $\OO_M$-supermodules on $\und{M}$,
 which is called the {\em superderivation sheaf} $\Der_M$.
\sk

Note that superderivations are local in the following sense:
Let $U\subseteq \und{M}$ be any open subset.  
The {\em support} of $f\in \OO_M(U)$ is the closed subset of $\und{M}$ that is defined by
\begin{flalign}
\supp(f) := U\setminus \bigcup\big\{ V : V\subseteq U \text{ is open and } \res_{U,V}(f)=0 \big\}~.
\end{flalign}
It then holds true that
$\supp(X(f))\subseteq \supp(f)$,
for all $X\in\Der_M(U)$ and $f\in\OO_M(U)$.
\sk

The superderivation sheaf $\Der_M$ of any object $M$ in $\SMan$ 
is a sheaf of locally free left $\OO_M$-supermodules
of the same dimension $n \vert m$ as the supermanifold $M$. In local coordinates $(x^1,\dots,x^{n+m})$ 
of $M$ in $V$ an adapted basis for $\Der_M(V)$ is given by the partial derivatives $\{\partial_1,\dots,\partial_{n+m}\}$, 
which are defined by $\partial_i x^j = \delta_i^j$, for all $i,j=1,\dots,n+m$. Notice
 that the first $n$ partial derivatives are even and the last $m$ partial derivatives are odd.
 \sk
 
We introduce {\em super-differential operators} via the usual recursive procedure:
For any open $U\subseteq \und{M}$, we have a $\SVec$-morphism $\OO_M(U)\to \underline{\Hom}(\OO_M(U),\OO_M(U))$
by assigning to any $g\in \OO_M(U)$ the linear map (denoted by the same symbol) $g: \OO_M(U)\to\OO_M(U)\,,~f\mapsto g\,f$.
The image of this $\SVec$-morphism is denoted by $\mathrm{DiffOp}_M^0(U)$
and it is called the super-differential operators of order zero. The super-differential operators of order $k\geq 1$ are then defined
recursively via
\begin{flalign}
\mathrm{DiffOp}_M^k(U) :=\big\{L\in \underline{\Hom}(\OO_M(U),\OO_M(U)) : [L,g]\in\mathrm{DiffOp}_M^{k-1}(U)~,~~\forall g\in \OO_M(U)\big\}~,
\end{flalign}
where the bracket denotes the supercommutator.
The super-vector space of super-differential operators on $\OO_M(U)$ is defined by
$\mathrm{DiffOp}_M(U) := \bigcup_{k\geq 0}\mathrm{DiffOp}_M^k(U)$.
Notice that $\mathrm{DiffOp}_M(U)$ carries the structure of a (not necessarily supercommutative) 
superalgebra with product given by the composition $\circ$ of linear maps 
$\underline{\Hom}(\OO_M(U),\OO_M(U))$.
The assignment $U\mapsto \mathrm{DiffOp}_M(U) $ (together with suitable restriction morphisms)
is a sheaf of superalgebras  on $\und{M}$, which is called the sheaf of super-differential operators $\mathrm{DiffOp}_M$.
Super-differential operators are local,
i.e.\ given any open $U\subseteq\und{M}$ we have that
$\supp(D(f))\subseteq\supp(f)$,
for any $D\in\mathrm{DiffOp}_M(U)$ and $f\in\OO_M(U)$.
Any superderivation is a super-differential operator of order $1$.

\paragraph{Super-differential forms:}
Let $M = (\und{M},\OO_{M})$ be any object in $\SMan$.
For any open $U\subseteq \und{M}$ we define the
left $\OO_{M}(U)$-supermodule of {\em super-one-forms} on $\OO_{M}(U)$
as the dual of the left $\OO_{M}(U)$-supermodule $\Der_{M}(U)$,
i.e.\ $\Omega^1_M(U) := \underline{\Hom}_{\OO_{M}(U)}(\Der_{M}(U),\OO_M(U))$.
The assignment $U \mapsto \Omega^1_M(U)$ (together with suitable restriction morphisms)
is a sheaf of left $\OO_M$-supermodules on $\und{M}$ that we shall denote by
$\Omega^1_M = \underline{\Hom}_{\OO_M}(\Der_M,\OO_M)$. For simplifying the notation we 
denote the left $\OO(M)$-supermodule of global sections of $\Omega^1_M$ 
by $\Omega^1(M)$. The canonical evaluation of elements in $\underline{\Hom}_{\OO_{M}(U)}(\Der_{M}(U),\OO_M(U))$
on superderivations $\Der_M(U)$ defines a sheaf morphism
$\mathrm{ev} : \Omega^1_M\otimes_{\OO_M}\Der_M \to \OO_{M}$, which due to the commutativity constraint $\sigma$ 
in $\OO_{M}\text{-}\Mod$ can also be understood as a pairing
\begin{flalign}
\ip{\,\cdot\,}{\,\cdot\,} := \mathrm{ev} \circ \sigma_{\Der_M,\Omega^1_M} : 
\Der_M\otimes_{\OO_M}\Omega^1_M \longrightarrow \OO_M~,
\end{flalign}
where the superderivations are on the left. (Moving $\Der_M$ to the left will lead to more convenient
sign conventions for the differential below.)
Explicitly, given any homogeneous $\omega\in\Omega^1_M(U)$ and $X\in\Der_M(U)$, the pairing reads as
$\ip{X}{\omega} = (-1)^{\vert\omega\vert\,\vert X\vert}\,\omega(X)$.
The {\em differential} $\dd : \OO_{M} \to\Omega^1_M$ is the sheaf morphism defined by the condition
$\ip{X}{\dd f} :=X(f)$,
for all $X\in\Der_M(U)$ and $f\in\OO_M(U)$, where $U\subseteq \und{M}$ is any open subset.
We notice that $\Omega^1_M$ is a sheaf of locally free left $\OO_M$-supermodules of the same dimension
$n\vert m$ as the supermanifold $M$. Explicitly, in local coordinates $(x^1,\dots,x^{n+m})$ of $M$
in $V$ an adapted basis for $\Omega^1_M(V)$ is given by the differentials
$\{\dd x^1,\dots , \dd x^{n+m}\}$, which can also be characterized by the duality relations
$\ip{\partial_i}{\dd x^j} = \delta^j_i$, for all $i,j=1,\dots,n+m$. Notice that the first $n$ differentials
 are even and that the last $m$ differentials are odd.
\sk

Super-one-forms can be pulled back along $\SMan$-morphisms $\chi : M\to M^\prime$.
In our work we shall only need the pull-back for the special case where $\und{\chi}(\und{M})\subseteq \und{M^\prime}$
is open and $\chi: M\to M^\prime\vert_{\und{\chi}(\und{M})}$ is a $\SMan$-isomorphism.
Let us first assume that $\chi: M\to M^\prime$ is a $\SMan$-isomorphism. 
Then $\chi_U^\ast : \OO_{M^\prime}(U)\to\OO_{M}(\und{\chi}^{-1}(U))$ is a $\mathsf{SAlg}$-isomorphism
and we can define a push-forward of superderivations by
\begin{flalign}
{\chi_{\ast}}_{U} : \Der_M(\und{\chi}^{-1}(U))\longrightarrow \Der_{M^\prime}(U) ~,~~
X\longmapsto (\chi_U^\ast)^{-1}\circ X\circ \chi_U^\ast~,
\end{flalign}
for all open $U\subseteq \und{M^\prime}$.
The {\em pull-back of super-one-forms} (denoted with a slight abuse of notation also
by $\chi_U^\ast$) $\chi_U^\ast : \Omega^1_{M^\prime}(U) \to \Omega^1_{M}(\und{\chi}^{-1}(U))$ 
is then defined by the duality relations
\begin{flalign}
\ip{X}{\chi_U^\ast(\omega)}:= \chi_U^\ast\left(\ip{{\chi_\ast}_{U}(X)}{\omega}\right)~,
\end{flalign}
for all $\omega\in\Omega_{M^\prime}^1(U)$ and $X\in\Der_M(\und{\chi}^{-1}(U))$.
In the case where we only have that  $\chi: M\to M^\prime\vert_{\und{\chi}(\und{M})}$
is a $\SMan$-isomorphism we define the pull-back of super-one-forms by
\begin{flalign}
\chi_{U}^\ast(\omega) := \chi^\ast_{U\cap\und{\chi}(\und{M})}\big(\res_{U,U\cap\und{\chi}(\und{M})}(\omega)\big)~,
\end{flalign}
for all $\omega\in\Omega_{M^\prime}^1(U)$ and all open $U\subseteq \und{M^\prime}$.
The following properties are immediate from this definition
\begin{flalign}
\chi_U^\ast(f\,\omega) = \chi_U^\ast(f)\,\chi_U^\ast(\omega)~~,\qquad \chi_U^\ast(\dd f) = \dd \chi_U^\ast(f)~,
\end{flalign}
for all $f\in\OO_{M^\prime}(U)$ and $\omega\in\Omega^1_{M^\prime}(U)$. We shall denote the
pull-back of global sections of $\Omega^1_{M^\prime}$ simply by $\chi^\ast := \chi^\ast_{\und{M^\prime}} : \Omega^1(M^\prime)
\to \Omega^1(M)$.
\sk

In complete analogy to the case of ordinary manifolds, one can define a sheaf
of differential graded superalgebras $\Omega^\bullet_M := \bigwedge^\bullet \Omega^1_M$,
for which we shall use the same sign conventions as in \cite[\S 3.3]{DMSuper}.
Sections of this sheaf are called {\em super-differential forms} and we recall that
for supermanifolds $M$ of dimension $n\vert m$ with $m\neq 0$ there are no top-degree forms.

\paragraph{Berezin integration:}
Given any object $M=(\und{M},\OO_M)$ in $\SMan$,
one can construct the Berezinian sheaf $\Ber(\Omega^1_M)$ 
of the super-one-form sheaf $\Omega^1_M$.
Notice that $\Ber(\Omega^1_M)$ is a sheaf of locally free left $\OO_M$-supermodules:
In local coordinates $(x^1,\dots,x^{n+m}) = (x^1,\dots,x^n,\theta^1,\dots,\theta^m)$ 
of $M$ in $V$ an adapted basis for $\Ber(\Omega^1_M)(V)$
is given by $[\dd x^1,\dots,\dd x^{n},\dd\theta^1,\dots,\dd\theta^m]$.
Given now any $f\in\OO_M(V) $ 
with compact support in $V\subseteq\und{M}$, we define the local {\em Berezin integral} over the open subsupermanifold
$M\vert_V$ by
\begin{flalign}\label{eqn:localintegral}
\int_{M\vert_V} [\dd x^1,\dots,\dd x^{n},\dd\theta^1,\dots,\dd\theta^m]\,f := \int_V \dd x^1\,\cdots \dd x^n~f_{(1,\dots, 1)}~,
\end{flalign}
where we have used the expansion $f = \sum_{\vert I\vert <m} \theta^I\,f_I + \theta^m\,\cdots\theta^1 f_{(1,\dots,1)}$
in terms of the odd coordinate functions. Here $I=(i_1,\dots,i_m)\in \{0,1\}^m$ is a multiindex and $\theta^I :=
(\theta^m)^{i_1}\cdots (\theta^1)^{i_m}$. Due to the change of variables formula, 
 the local Berezin integral can be globalized, see e.g.\ \cite[\S 3.10]{DMSuper}.
In the case of $M$ being oriented this yields a linear functional
\begin{flalign}
\int_M  : \Ber(\Omega^1_M)_{\cc}(\und{M}) \longrightarrow \bbR~,~~v\longmapsto \int_M v~
\end{flalign}
on the compactly supported sections of $\Ber(\Omega^1_M)$,
which for sections $v$ with support in $V$ is given by (\ref{eqn:localintegral}).
\sk

Given any $\SMan$-isomorphism $\chi : M \to M^\prime$,
there exists a pull-back $\chi^\ast : \Ber(\Omega^1_{M^\prime})\to \Ber(\Omega^1_{M})$.
In local coordinates $(x^{\prime 1},\dots,x^{\prime n+m})$ 
of $M^\prime$ in $V^\prime$ and $(x^1,\dots,x^{n+m})$ of $M$ in $V = \und{\chi}^{-1}(V^\prime)$,
the pull-back is given by
\begin{flalign}
\nn \chi^\ast_{V^\prime}\big([\dd x^{\prime 1},\dots,\dd x^{\prime n+m}]\big)
&= [\dd \chi^\ast_{V^\prime}(x^{\prime 1}),\dots,\dd \chi^\ast_{V^\prime}(x^{\prime n+m})]\\
&= \Ber(J(\chi))\,[\dd x^1,\dots,\dd x^{n+m}]~,
\end{flalign}
where $J(\chi)$ is the super-Jacobi matrix with entries defined 
by $\dd\chi^{\ast}_{V^\prime}(x^{\prime i}) =\sum_{j=1}^{n+m}  J(\chi)_j^i\,\dd x^j$.
Notice the property $\chi^\ast_U(v\,f) = \chi^{\ast}_U(v)\,\chi^\ast_U(f)$,
for all $v\in\Ber(\Omega^1_{M^\prime})(U)$, $f\in\OO_{M^\prime}(U)$ 
and all open $U\subseteq\und{M^\prime}$.
Given now two oriented supermanifolds $M$ and $M^\prime$
together with an orientation preserving $\SMan$-isomorphism
$\chi : M \to M^\prime$, the global Berezin integral transforms as
\begin{flalign}\label{eqn:trafoformel}
\int_{M} \chi^\ast(v) = \int_{M^\prime} v~,
\end{flalign}
for all $v\in \Ber(\Omega^1_{M^\prime})_{\cc}(\und{M^\prime})$. This follows in the usual 
way from the aforementioned local change of variables formula.

%%%%%%%%%%%%%%%%%%%%%%%%%%%%%%%%%%%%%%%%%%%%%%%%
%%%%%%%%%%%%%%%%%%%%%%%%%%%%%%%%%%%%%%%%%%%%%%%%

\section{\label{sec:cartan}Super-Cartan supermanifolds}
We introduce the concept of super-Cartan structures on supermanifolds.
These structures have their origin in the superspace formulation of supergravity, 
see e.g.\ \cite{Wess:1977fn}. In our work they are required for constructing natural super-differential
operators which govern the dynamics of super-field theories at the classical 
and quantum level. In this paper we shall only focus on the case where the super-principal bundle 
underlying the super-Cartan structure is {\em globally trivial} (and trivialized),
which considerably simplifies our discussion as we do not have to 
deal with super-principal bundles and their associated 
super-vector bundles. Notice that the latter concepts are
essentially well understood, see e.g.\ \cite{Bruzzo},
but they are cumbersome to work with.
To the best of our knowledge there is not yet a fully developed theory of (global)
super-Cartan supermanifolds along the lines of standard treatments in ordinary 
differential geometry \cite[Chapter 5]{Sharpe}. In particular, we are not aware of conditions 
(on the basis supermanifold) which ensure the existence of super-Cartan structures or 
conditions which ensure that the underlying super-principal bundle is generically trivial.
We hope to come back to this problem in a future work, which then would allow us
to study super-field theories on globally non-trivial super-Cartan supermanifolds.

\subsection{Representation theoretic data and super-Poincar{\'e} super-Lie algebras}
We shall briefly review super-Poincar{\'e} super-Lie algebras in various dimensions. See e.g.\
\cite{Deligne} and \cite[Lecture 3]{FreedLecture} for details.
\sk

Let $W$ be a finite-dimensional real vector space and $g : W\otimes W\to\bbR$ a Lorentz metric
of signature $(+,-,\cdots,-)$.
Let further $S$ be a real spin representation of the associated 
spin group $\mathrm{Spin}(W,g)$ and
\begin{flalign}
\Gamma : S \otimes S \longrightarrow W
\end{flalign}
a symmetric and $\mathrm{Spin}(W,g)$-equivariant pairing.
We denote the spin group actions on $W$ and $S$ by, respectively,
$\rho^W : \mathrm{Spin}(W,g)\times W \to W$ and $\rho^{S} : \mathrm{Spin}(W,g)\times S \to S$.
Moreover, we simply write $\mathfrak{spin}$ for the Lie algebra of $\mathrm{Spin}(W,g)$ and recall
that the spin group actions above induce Lie algebra actions,
which we denote by $\rho^W_\ast : \mathfrak{spin}\otimes W \to W$ and $\rho^{S}_\ast :
\mathfrak{spin}\otimes S \to S$.

\begin{defi}
 Let us fix any choice of the data $(W,g,S,\Gamma)$.
\begin{itemize}
\item[(i)] The super-Poincar{\'e} super-Lie algebra $\mathfrak{sp}$ (corresponding to this data)
is given by the super-vector space
\begin{subequations}
\begin{flalign}
\mathfrak{sp} := (\mathfrak{spin} \oplus W) \oplus S ~
\end{flalign}
(with $(\mathfrak{spin} \oplus W)$ even and $S$ odd), 
together with the super-Lie bracket defined by
\begin{multline}\label{eqn:superbracketsp}
[L_1\oplus w_1 \oplus s_1 , L_2 \oplus w_2\oplus s_2] =\\
 [L_1,L_2] \oplus \Big( \rho^W_\ast\big(L_1\otimes w_2 - L_2\otimes w_1\big)-2\,\Gamma(s_1,s_2) \Big) \oplus \rho_{\ast}^{S} \big(L_1\otimes s_2-L_2\otimes s_1\big)~,
\end{multline}
\end{subequations}
for all $L_1\oplus w_1\oplus s_1, L_2\oplus w_2\oplus s_2 \in \mathfrak{sp}$.

\item[(ii)] 
The  supertranslation super-Lie algebra $\mathfrak{t}$ (corresponding to this data) is given by the 
super-vector space
\begin{subequations}
\begin{flalign}
\mathfrak{t}:= W\oplus S
\end{flalign}
(with $W$ even and $S$ odd), together with the
super-Lie bracket defined by
\begin{flalign}
[w_1\oplus s_1 , w_2\oplus s_2] = -2\,\Gamma(s_1,s_2) \oplus 0~,
\end{flalign}
\end{subequations}
for all $w_1\oplus s_1,w_2\oplus s_2 \in \mathfrak{t}$.
\end{itemize}
\end{defi}
\begin{rem}
Notice that the bracket defined in (\ref{eqn:superbracketsp}) is indeed a super-Lie bracket. The super-skew symmetry
is evident from its definition and the super-Jacobi identity is a straightforward check using 
the $\mathrm{Spin}(W,g)$-equivariance of $\Gamma$, which implies that
\begin{flalign}
\rho^W_\ast\big(L \otimes \Gamma(s_1,s_2)\big) = \Gamma\big(\rho_\ast^{S}(L\otimes s_1),s_2\big) + 
\Gamma\big(s_1,\rho_\ast^{S}(L\otimes s_2)\big)~,
\end{flalign}
for all $L\in\mathfrak{spin}$ and $s_1,s_2\in S$.
\end{rem}

The following statement is easily shown. We therefore can omit the proof.
\begin{propo}
For any choice of the data $(W,g,S,\Gamma)$, the super-Poincar{\'e} super-Lie algebra $\mathfrak{sp}$
is a super-Lie algebra extension of the Lie algebra  $\mathfrak{spin}$ (regarded as a super-Lie algebra)
by the supertranslation super-Lie algebra $\mathfrak{t}$, i.e.\ 
\begin{flalign}
0\longrightarrow \mathfrak{t}\longrightarrow \mathfrak{sp}\longrightarrow \mathfrak{spin}\longrightarrow 0~~,
\end{flalign}
with the obvious definition of the arrows, is a short exact sequence of super-Lie algebras.
\end{propo}

The data $(W,g,S,\Gamma)$ we have introduced above is sufficient in order to construct the super-Poincar{\'e} and
supertranslation super-Lie algebras (corresponding to this choice of data).
For our applications to super-Cartan supermanifolds and super-field theories
we require some additional data. First, let us fix a positive cone $C\subset W$ of timelike vectors and assume that
$\Gamma : S \otimes S \to W$ is positive in the sense that $\Gamma(s,s)\in\overline{C}$, for all $s\in S$,
with $\Gamma(s,s)=0$ only for $s=0$. Here $\overline{C}$ denotes the closure of the cone 
$C\subset W$. The existence of such $\Gamma$ 
has been shown in \cite{Deligne}. The positive cone $C\subset W$ will later play the role of a time-orientation.
Next, we assume that we have given a $\mathrm{Spin}(W,g)$-invariant linear map
\begin{flalign}
\epsilon : S \otimes S \longrightarrow \bbR~,
\end{flalign} 
which is either a metric (of positive signature) or a symplectic structure. 
Such linear maps exist if $\dim(W)$  is not equal to $2$ or $6$ modulo $8$,
see  \cite{Supersolutions}. Finally, we take as part of the data 
a choice of orientations $o_W$ of $W$ and $o_{S}$ of $S$.
These orientations and also $\epsilon$ will be used define a canonical Berezinian density
on any super-Cartan supermanifold and therefore a notation of integration.
In summary, we will always assume as a starting point for our constructions
that the data $(W,g,S,\Gamma,C,\epsilon,o_W,o_{S})$ are given.

\subsection{Super-Cartan structures}
In order to simplify our studies on super-Cartan supermanifolds 
we shall restrict our attention to super-Cartan structures which are 
based on globally trivial (and also trivialized) super-principal $\mathrm{Spin}(W,g)$-bundles.
Let us fix any choice of the data $(W,g,S,\Gamma,C,\epsilon,o_W,o_{S})$.
\begin{defi}  
Let $M $ be a $\dim(W)\vert \dim(S)$-dimensional supermanifold. 
A (globally trivial) super-Cartan structure
on $M$ is a pair $(\Omega,E)$ consisting of
an even super-one-form $\Omega\in \Omega^1(M,\mathfrak{spin})$ (called the super-spin connection)
and an even and non-degenerate super-one-form $E \in \Omega^1(M,\mathfrak{t})$ (called the supervielbein).
The triple $\bbM:=(M,\Omega,E)$ is called a super-Cartan supermanifold.
\end{defi}
\begin{rem}
Notice that the requirement that $E$ is non-degenerate fixes the dimension of $M$ to be
the dimension $\dim(W)\vert \dim(S)$ of the supertranslation super-Lie algebra $\mathfrak{t}$.
\end{rem}

To any super-Cartan supermanifold $\bbM = (M,\Omega,E)$ we can assign its supercurvature and supertorsion,
which play an important role in supergravity. They are defined by
\begin{subequations}\label{eqn:supertorsion}
\begin{flalign}
R_{\bbM} &:= \dd \Omega + [\Omega,\Omega]\in\Omega^2(M,\mathfrak{spin})~,\\
T_{\bbM} &:= \dd_\Omega E := \dd E +[\Omega,E]\in\Omega^2(M,\mathfrak{t})~,
\end{flalign}
\end{subequations}
where the brackets are those induced by the super-Lie bracket in $\mathfrak{sp}$ via
\begin{flalign}
\nn [\,\cdot\,,\,\cdot\,] : \Omega^k(M,\mathfrak{sp}) \otimes \Omega^l(M,\mathfrak{sp})&\longrightarrow \Omega^{k+l}(M,\mathfrak{sp})~,\\
(\omega_1\otimes X_1)\otimes (\omega_2\otimes X_2)&\longmapsto (-1)^{\vert X_1\vert \,\vert\omega_2\vert}\,\omega_1\wedge\omega_2 \otimes [X_1,X_2]~.
\end{flalign}

We now shall study integration on super-Cartan supermanifolds. Let us recall that the super-vector space $\mathfrak{t}$ has 
an adapted basis $\{p_0,\dots,p_{\dim(M)-1}, q_{1},\dots,q_{\dim(S)}\}$, 
i.e.\ $p_\alpha\in W$ and $q_{a}\in S$, for all $\alpha =0,\dots,\dim(M)-1$ and $a\in 1,\dots,\dim(S)$.
Making use of the Lorentz metric $g$ on $W$ and the metric (or symplectic structure) $\epsilon$ on $S$,
we can demand that $\{p_{\alpha}\}$ is an orthonormal basis for $(W,g)$ 
and that $\{q_a\}$ is an orthonormal (or symplectic/Darboux) basis for $(S,\epsilon)$.
Making further use of the orientations $o_W$ of $W$ and $o_{S}$ of $S$ 
we demand that these bases are oriented and finally by using the positive cone $C\subset W$ we demand
that the basis for $W$ is time-oriented, i.e.\ the timelike basis vector $p_0$ lies in $C$.
We shall call any adapted basis for $\mathfrak{t}$ which is of this kind 
an orthonormal (or orthosymplectic) time-oriented and oriented  adapted basis for $\mathfrak{t}$.
Notice that any two orthonormal (or orthosymplectic) time-oriented and oriented adapted bases for $\mathfrak{t}$
are related by a $\SVec$-automorphism $L \in \mathrm{GL}(\mathfrak{t})$,
whose block-matrix components (cf.\ (\ref{eqn:matrixbasis})) are
$L_1\in \mathrm{SO}_0(1,\dim(W)-1)$, $L_2=L_3=0$ and $L_4\in \mathrm{SO}(\dim(S))$ (or $L_4\in \mathrm{Sp}(\dim(S),\bbR)$).
Because of $\det(L_1) = \det(L_4)=1$ we find that
$\mathrm{Ber}(L) = 1$, cf.\ (\ref{eqn:Berezinian}).
We now may expand the supervielbein $E\in\Omega^1(M,\mathfrak{t})$ in terms of any orthonormal (or orthosymplectic) 
time-oriented and oriented  adapted basis for $\mathfrak{t}$, which yields
\begin{flalign}\label{eqn:Eexpansion}
E = \sum_{\alpha =0}^{\dim(M)-1} e^\alpha\otimes p_{\alpha} + \sum_{a=1}^{\dim(S)} \xi^a\otimes q_a~,
\end{flalign}
where all $e^\alpha \in\Omega^1(M)$ are even and all $\xi^a\in\Omega^1(M)$ are odd.
Notice that the collection $\{e^0,\dots,e^{\dim(W)-1}, \xi^1,\dots,\xi^{\dim(S)}\}$ is an adapted basis for $\Omega^1(M)$
since $E$ was assumed to be non-degenerate. We hence can define an element of the Berezinian 
supermodule of $\Omega^1(M)$ by
\begin{flalign}\label{eqn:Berezindensity}
\mathrm{Ber}(E) := [e^0,\dots, e^{\dim(M)-1},\xi^1,\dots,\xi^{\dim(S)}] \in \mathrm{Ber}(\Omega^1(M))~.
\end{flalign}
Recalling (\ref{eqn:Berezinianrelations}), we find that this definition does not depend 
on the choice of the orthonormal (or orthosymplectic) 
time-oriented and oriented  adapted basis for $\mathfrak{t}$, since, as we have explained above,
any two such bases are related by an $L\in\mathrm{GL}(\mathfrak{t})$ with $\mathrm{Ber}(L)=1$.
\sk

Using  the Berezinian density (\ref{eqn:Berezindensity})  we can define a pairing on the
compactly supported sections of the structure sheaf of $M$ by
\begin{flalign}\label{eqn:pairing}
\ip{\,\cdot\,}{\,\cdot\,}_{\bbM} : \OO_{\cc}(M)\otimes \OO_{\cc}(M)\longrightarrow \bbR~,~~
 F_1\otimes F_2 \longmapsto \int_{M}\mathrm{Ber}(E)\, F_1\,F_2~.
\end{flalign}
Notice that the $\bbZ_2$-parity of the linear map $\ip{\,\cdot\,}{\,\cdot\,}_{\bbM}$
 is $\dim(S)\,\, \mathrm{mod}\,\,2$  and that $\ip{\,\cdot\,}{\,\cdot\,}_{\bbM}$ can be extended to all $F_1,F_2\in\OO(M)$
  with compactly overlapping support.
Notice further that
\begin{flalign}\label{eqn:pairingsymmetry}
\ip{F_1}{F_2}_{\bbM} = (-1)^{\vert F_1\vert \,\vert F_2\vert} \ip{F_2}{F_1}_{\bbM}~,
\end{flalign}
for all homogeneous $F_1,F_2\in\OO(M)$ with compactly overlapping support.
\sk

We finish this subsection by defining a suitable category of super-Cartan supermanifolds.
\begin{defi}\label{defi:SCart0}
The category $\SCart$ consists of the following objects and morphisms:
\begin{itemize}
\item The objects are all super-Cartan supermanifolds $\bbM = (M,\Omega,E)$.
\item The morphisms $\chi : \bbM\to\bbM^\prime$  are all $\SMan$-morphisms
(denoted by the same symbol) $\chi: M\to M^\prime$, such that 
\begin{enumerate}
\item $\und{\chi} : \und{M} \to\und{M^\prime}$ is an open embedding,
\item $\chi : M\to M^\prime\vert_{\und{\chi}(\und{M})}$ is a $\SMan$-isomorphism,
\item the super-Cartan structures are preserved, i.e.\ $\chi^\ast(\Omega^\prime) = \Omega$ and $\chi^\ast(E^\prime) = E$.
\end{enumerate}
\end{itemize}
\end{defi}
\begin{rem}
Since $\chi^\ast(E^\prime) = E$ it is clear that any $\SCart$-morphism $\chi : \bbM\to \bbM^\prime$ preserves
the Berezinian densities (\ref{eqn:Berezindensity}), i.e.\ 
\begin{flalign}\label{eqn:berezindensitymorph}
\chi_{\und{\chi}(\und{M})}^\ast \left(\res_{\und{M^\prime},\und{\chi}(\und{M})}\left(\mathrm{Ber}(E^\prime)\right)\right)
= \mathrm{Ber}(E)~.
\end{flalign}
\end{rem}

\subsection{Lorentz geometry on the reduced manifold}
Given any super-Cartan supermanifold $\bbM=(M,\Omega,E)$ we can equip the reduced manifold
$\und{M}$ with a Lorentz metric, an orientation and a time-orientation. 
Explicitly, the pull-back of the supervielbein $E\in\Omega^1(M,\mathfrak{t})$
along the embedding of the reduced manifold $\iota_{\und{M},M} : \und{M}\to M$
provides us with a non-degenerate one-form on $\und{M}$ with values in the even part 
$W$ of $\mathfrak{t}$ that we  shall denote by $\und{E} := \iota^\ast_{\und{M},M}(E)\in\Omega^1(\und{M},W)$.
We define a Lorentz metric $\und{g}_{\bbM}$ on $\und{M}$ by setting
\begin{flalign}
\und{g}_{\bbM} : \Gamma^\infty (T \und{M})\times \Gamma^\infty(T\und{M}) \longrightarrow C^\infty(\und{M})~,~~(X,Y)\longmapsto g\big(\langle X,\und{E}\rangle,\langle Y,\und{E}\rangle\big)~,
\end{flalign}
where $g:W\otimes W\to \bbR$ is the Lorentz metric on $W$ and $\ip{\,\cdot\,}{\,\cdot\,}$ is the duality 
pairing between vector fields and one-forms on $\und{M}$.
The orientation $o_{W}$ of $W$ induces an orientation $\und{o}_{\bbM}$ 
of $\und{M}$, which we may represent by the volume form $\mathrm{vol}(\widetilde{E}) 
:= \und{e}^0\wedge\dots\wedge \und{e}^{\dim(W)-1}$,
where we have expanded $\und{E} = \sum_{\alpha=0}^{\dim(W)-1} \und{e}^\alpha\otimes p_{\alpha}$ in terms of 
any orthonormal time-oriented and oriented basis $\{p_\alpha\}$ for $W$.
The positive cone $C\subset W$ of timelike vectors in $W$ induces a time-orientation $\und{t}_{\bbM}$ 
on the reduced oriented Lorentz manifold $(\und{M},\und{g}_{\bbM},\und{o}_{\bbM})$,
which we may represent by the vector field $\widetilde{X}_0$ on $\und{M}$ that is
defined via the duality relations $\langle \und{X}_\alpha, \und{e}^\beta \rangle = \delta_{\alpha}^\beta$.
\sk

We shall now show that 
the assignment of the reduced oriented and time-oriented Lorentz manifolds 
$\und{\bbM} :=(\und{M},\und{g}_{\bbM},\und{o}_{\bbM},\und{t}_{\bbM})$
to super-Cartan supermanifolds $\bbM$ is functorial.
We define the category of oriented and time-oriented Lorentz manifolds $\tLor$ as follows:
The objects are all oriented and time-oriented Lorentz manifolds and the morphisms are all
open isometric embeddings that preserve the orientations and time-orientations.
\begin{propo}
The following assignment defines a functor $\und{\,\cdot\,} : \SCart \to \tLor$:
To any object $\bbM$ in $\SCart$ we assign the oriented and time-oriented Lorentz manifold
$\und{\bbM} := (\und{M},\und{g}_{\bbM},\und{o}_{\bbM},\und{t}_{\bbM})$ that has been constructed above.
To any morphism $\chi : \bbM\to \bbM^\prime$ in $\SCart$ we assign the $\tLor$-morphism
$\und{\chi} : \und{\bbM}\to\und{{\bbM^\prime}}$ that is defined by the reduced
morphism $\und{\chi} : \und{M}\to\und{{M^\prime}}$.
\end{propo}
\begin{proof}
We have to prove that the reduced morphism $\und{\chi} : \und{M}\to\und{{M^\prime}}$ 
is isometric and that it preserves the orientations and time-orientations.
All these properties follow from the fact that $\und{\chi}^\ast(\und{E^\prime}) = \und{E}$,
which is shown by the calculation
\begin{flalign}
\nn \und{\chi}^\ast(\und{E^\prime}) &= \und{\chi}^\ast \circ \iota_{\und{M^\prime},M^\prime}^\ast (E^\prime)
=(\iota_{\und{M^\prime},M^\prime} \circ \widetilde{\chi})^\ast(E^\prime)
=(\chi\circ \iota_{\und{M},M})^\ast(E^\prime)\\
&= \iota_{\und{M},M}^\ast\circ \chi^\ast(E^\prime) =  \iota_{\und{M},M}^\ast(E) = \und{E}~,
\end{flalign}
where in the third equality we have used the commutative diagram (\ref{eqn:reducedmanifolddiagram}).
\end{proof}

Due to this proposition we can define the 
chronological and causal future/past 
of a subset $A\subseteq \und{M}$ in a super-Cartan supermanifold $\bbM$
in terms of the corresponding reduced oriented and time-oriented Lorentz manifold
$\und{\bbM}$. (See e.g.\ \cite{Bar:2007zz} for a definition of the chronological and causal future/past
of a subset of a time-oriented Lorentz manifold.)
\begin{defi}
Let $\bbM$ be a super-Cartan supermanifold and $A\subseteq \und{M}$ 
a subset of its reduced manifold.
The chronological future/past of $A$ in $\bbM$ is defined by
\begin{subequations}
\begin{flalign}
I_{\bbM}^\pm(A):= I_{\und{\bbM}}^\pm(A) \subseteq \und{M}
\end{flalign}
and the causal future/past by 
\begin{flalign}
J_{\bbM}^{\pm}(A) :=
 J_{\und{\bbM}}^\pm(A) \subseteq \und{M}~.
\end{flalign}
\end{subequations}
We further define $I_{\bbM}(A) := I_{\bbM}^+(A)\cup I_{\bbM}^-(A)$
and $J_{\bbM}(A) := J_{\bbM}^+(A)\cup J_{\bbM}^-(A)$.
\end{defi}

\subsection{The category of globally hyperbolic super-Cartan supermanifolds}
For our studies on super-field theories the
category of globally hyperbolic super-Cartan supermanifolds
will play a major role. It can be defined as a certain subcategory of $\SCart$.
\begin{defi}\label{defi:SCart}
The category $\ghSCart$ consists of the following objects and morphisms:
\begin{itemize}
\item The objects in $\ghSCart$  are all objects $\bbM$ in $\SCart$ such that 
the reduced oriented and time-oriented Lorentz manifold $\und{\bbM}$ is globally hyperbolic.

\item The morphisms $\chi : \bbM \to \bbM^\prime$ between two objects $\bbM$ and $\bbM^\prime$ in $\ghSCart$
are all $\SCart$-morphisms such that the image of the reduced morphism $\und{\chi} : \und{\bbM}\to\und{\bbM^\prime}$
is a causally compatible subset of $\und{\bbM^\prime}$.\footnote{
Recall that a subset $A\subseteq N$ of a time-oriented Lorentz manifold $N$ 
is called causally compatible provided that $J^\pm_{A}(x) = J^\pm_{N}(x) \cap A$, for all $x\in A$.
}
\end{itemize}
\end{defi}

%%%%%%%%%%%%%%%%%%%%%%%%%%%%%%%%%%%%%%%%%%%%%%%%
%%%%%%%%%%%%%%%%%%%%%%%%%%%%%%%%%%%%%%%%%%%%%%%%

\section{\label{sec:axioms}Axiomatic definition of super-field theories}
Motivated by the examples we will discuss in Section \ref{sec:examples},
we shall give an axiomatic characterization of super-field
theories by representation theoretic and geometric data. This is a reasonable and useful
approach since all of our statements concerning the construction of super-QFTs
in Section \ref{sec:ordinary} can be made 
at this abstract level, so there is no need to focus on explicit models at this point. Moreover, 
the problem of constructing models of super-QFTs is thereby reduced to
finding explicit realizations of the assumed representation theoretic and geometric data.
It will be instructive to first provide some motivations explaining
 our choice of data.

\paragraph{Representation theoretic data:}
Motivated by Section \ref{sec:cartan}, our first choice of data is given by an eight-tuple
$(W,g,S,\Gamma,C,\epsilon,o_{W},o_{S})$ consisting
 of a finite-dimensional real vector space $W$, a Lorentz metric $g:W\otimes W\to \bbR$, 
 a real spin representation $S$ of $\mathrm{Spin}(W,g)$, a symmetric and
 $\mathrm{Spin}(W,g)$-equivariant pairing $\Gamma : S\otimes S\to W$
which is positive with respect a choice of positive cone $C\subset W$ of timelike vectors,
 a $\mathrm{Spin}(W,g)$-invariant linear map $\epsilon : S \otimes S\to\bbR$ 
 which is either a metric (of positive signature) or a symplectic structure,
 and orientations $o_W$ on $W$ and $o_{S}$ on $S$.
It becomes evident from Section \ref{sec:cartan} that this data is required,
on the one hand, to specify a super-Poincar{\'e} and supertranslation super-Lie algebra and,
on the other hand, to describe super-Cartan structures on supermanifolds together with
time-orientation and integration. In other words, the representation theoretic
data fixes the local model space of the super-Cartan supermanifold
and in particular its dimension to $\dim(W)\vert\dim(S)$. Physically speaking, this means 
that the representation theoretic data fixes the dimension of the reduced spacetime and the amount of supersymmetry.

\paragraph{Admissible super-Cartan supermanifolds:} 
As we will show in Section \ref{sec:examples} by studying explicit examples,
one should not expect that the super-field theory will be defined on the whole category
$\ghSCart$ of globally hyperbolic super-Cartan supermanifolds, see Definition \ref{defi:SCart}.
A common feature in many super-field theories (especially in supergravity) is that
one has to impose constraints on the superfields in order to arrive at a reasonable theory.
Such constraints may in particular include the supergravity supertorsion constraints \cite{Wess:1977fn},
which restrict the class of super-Cartan structures, i.e.\ the objects in $\ghSCart$. 
As we would like to keep our axiomatic setting as flexible as possible,
we shall at this point not specify the explicit form of these constraints
but rather include a {\em choice} of full subcategory $\SLoc$ of $\ghSCart$ as a part of the data.
Objects in $\SLoc$ will be called admissible super-Cartan supermanifolds
and, at least for the examples studied in Section \ref{sec:examples}, a reasonable
choice of $\SLoc$ is given by the supergravity supertorsion constraints.

\paragraph{Super-differential operators:}
We shall only consider super-field theories whose configurations on any object $\bbM= (M,\Omega,E)$ in $\SLoc$ 
can be described in terms of the super-vector space  $\OO(\bbM):=\OO(M)$ of global sections
of the structure sheaf. Notice that $\OO : \SLoc^\op \to \SVec$ is a functor, namely the global section functor.
The dynamics of the super-field theory will be encoded in terms 
of super-differential operators $P_{\bbM} : \OO(\bbM) \to \OO(\bbM)$ of $\bbZ_2$-parity
$\dim(S)\,\,\mathrm{mod}\,\,2$, which are assumed to be natural 
in the sense that the diagram
\begin{flalign}\label{eqn:Pnatural}
\xymatrix{
\ar[d]_-{\OO(\chi^\op)=\chi^\ast} \OO(\bbM^\prime) \ar[rr]^{P_{\bbM^\prime}}&& \OO(\bbM^\prime)\ar[d]^-{\OO(\chi^\op) = \chi^\ast} \\
\OO(\bbM) \ar[rr]_-{P_{\bbM}} && \OO(\bbM)  
}
\end{flalign}
of linear maps commutes, for all $\SLoc^\op$-morphisms $\chi^\op : \bbM^\prime\to\bbM$
(i.e.\ all $\SLoc$-morphisms $\chi: \bbM\to\bbM^\prime$).
$P_{\bbM}$ can be regarded as the components of a natural transformation
$P : \OO \Rightarrow \OO$, provided that we enlarge the morphism sets in $\SVec$ by parity reversing linear maps.
More precisely, we shall replace the morphism sets $\Hom_{\SVec}(V,V^\prime)$ in $\SVec$ by the sets underlying the
 internal hom-objects $\underline{\Hom}(V,V^\prime)$. The corresponding category is then denoted by $\underline{\SVec}$
 and we have an obvious functor $\SVec\to\underline{\SVec}$. 
Hence, $\OO : \SLoc^\op\to \SVec$ defines  a functor (denoted by the same symbol)
 $\OO : \SLoc^\op\to \underline{\SVec}$.\footnote{\label{footnote:underline}
Notice that enlarging the category $\SVec$ to $\underline{\SVec}$ is required only if $\dim(S)$ is odd,
which is a peculiarity of the superparticle discussed in Section \ref{sec:examples}. In the more common situation
where $\dim(S)$ is even, all our constructions can be done within the subcategory $\SVec$ of $\underline{\SVec}$, 
so there is no need to introduce the category $\underline{\SVec}$ in this case.
We however decided to work with $\underline{\SVec}$ in order to develop a framework
that is general enough to include the superparticle, which is a valuable example
that can be analyzed in full detail, cf.\ Section \ref{sec:examples}.
 }
We additionally demand that, for any object $\bbM$ in $\SLoc$, the super-differential operator $P_{\bbM}$
is formally super-self adjoint with respect to the pairing (\ref{eqn:pairing}), i.e.\
\begin{flalign}\label{eqn:superself}
\ip{F_1}{P_{\bbM}(F_2)}_{\bbM} = (-1)^{\vert F_1\vert \,\vert P_{\bbM}\vert} \,\ip{P_{\bbM}(F_1)}{F_2}_{\bbM}~,
\end{flalign}
for all homogeneous $F_1,F_2\in\OO(\bbM)$ with compactly overlapping support.
For many of our constructions we also have to assign retarded/advanced super-Green's operators to
the super-differential operators $P_{\bbM}$, for all objects $\bbM$ in $\SLoc$. 
We therefore demand that, for any object $\bbM$ in $\SLoc$, the super-differential operator
$P_{\bbM}$ is super-Green's hyperbolic in the following sense:
\begin{defi}\label{defi:green}
Let $\bbM$ be any object in $\SLoc$. A homogeneous super-differential operator 
$P_{\bbM}: \OO(\bbM) \to \OO(\bbM)$ is called super-Green's hyperbolic
if there exists a retarded and advanced super-Green's operator, i.e.\ linear maps
$G_{\bbM}^\pm : \OO_{\cc}(\bbM)\to \OO(\bbM)$ of $\bbZ_2$-parity $\vert G_{\bbM}^\pm\vert = \vert P_{\bbM}\vert$
that satisfy 
\begin{itemize}
\item[(i)] $P_{\bbM}\circ G_{\bbM}^\pm = \id_{\OO_{\cc}(\bbM)}$,
\item[(ii)] $G_{\bbM}^\pm\circ P_{\bbM} \big\vert_{\OO_{\cc}(\bbM)} = \id_{\OO_{\cc}(\bbM)}$,
\item[(iii)] $\supp\big(G_{\bbM}^\pm(F)\big)\subseteq J_{\bbM}^\pm(\supp(F))$, for all $F\in\OO_{\cc}(\bbM)$.
\end{itemize}
\end{defi}
\sk

Motivated by the discussion above, we can now define abstractly a notion
of super-field theories. Our present axiomatic framework is supposed to
cover all super-field theories which in the physics literature would be called `real superfields'. The typical 
examples in $1\vert 1$ and $3\vert 2$-dimensions are discussed in Section \ref{sec:examples}, where it is shown that
they comply with our axioms. In contrast, `super-gauge theories' and `chiral superfields'
will require a more sophisticated set of axioms, which should include aspects of gauge invariance
and the chirality constraints (see \cite{Hack:2012dm} for an axiomatic approach to ordinary gauge theories). 
We shall leave these problems for future work
and consider in the present work the case of (real) super-field theories which we characterize
by the following axioms:
\begin{defi}\label{defi:sft}
A super-field theory is specified by the following data:
\begin{enumerate}
\item A choice of the representation theoretic data $(W,g,S,\Gamma,C,\epsilon,o_{W},o_{S})$.

\item A full subcategory $\SLoc$ of $\ghSCart$.

\item  A natural transformation $P : \OO \Rightarrow \OO$ of functors
 from $\SLoc^\op$ to $\underline{\SVec}$, such that $P_{\bbM}$
is a formally super-self adjoint and super-Green's hyperbolic super-differential operator
of $\bbZ_2$-parity $\dim(S)\,\,\mathrm{mod}\,\,2$, for any object $\bbM$ in $\SLoc$.
\end{enumerate}
\end{defi}

%%%%%%%%%%%%%%%%%%%%%%%%%%%%%%%%%%%%%%%%%%%%%%%%
%%%%%%%%%%%%%%%%%%%%%%%%%%%%%%%%%%%%%%%%%%%%%%%%

\section{\label{sec:ordinary}Construction of super-quantum field theories}
We show that given any super-field theory according to Definition \ref{defi:sft},
one can construct a functor $\AA : \SLoc \to \SAlg$ which satisfies
a supergeometric modification of the axioms of locally covariant quantum field
theory \cite{Brunetti:2001dx}. 
In other words, any super-field theory gives rise to a super-QFT.
We establish a connection between the super-field theory and its associated
super-QFT by showing that the latter has a locally covariant quantum field
which satisfies (in a weak sense) the equations of motion given by
the super-differential operators $P$. As usual our construction
will be done in two steps. First, we assign to a super-field theory 
a functor $\LL : \SLoc\to \XX$, where $\XX$ is the category of
super-symplectic spaces in the case of $\dim(S)$ even and the
category of super-inner product spaces in the case of $\dim(S)$ odd.
In the spirit of \cite{Brunetti:2001dx} this functor should be interpreted as a 
locally covariant classical field theory.
The locally covariant classical field theory is then quantized by 
a quantization functor $\QQ : \XX \to \SAlg$,
which implements super-canonical commutation relations (SCCR) in the case of $\dim(S)$ even and
super-canonical anticommutation relations (SCAR) in the case of $\dim(S)$ odd.

\subsection{The functor \texorpdfstring{$\OO_{\cc} : \SLoc \to \SVec$}{Oc : SLoc -> SVec}}
As a preparatory step, we shall show that
assigning to objects $\bbM$ in $\SLoc$ 
the super-vector spaces $\OO_{\cc}(\bbM) := \OO_{\cc}(M)$ 
of compactly supported global sections of the structure sheaf
can be described by a functor $\OO_{\cc} : \SLoc \to\SVec$.
Given any $\SLoc$-morphism $\chi : \bbM\to \bbM^\prime$,
we can define a $\SVec$-morphism $\chi_\ast := \OO_{\cc}(\chi) : \OO_{\cc}(\bbM)\to \OO_{\cc}(\bbM^\prime) $,
called the {\em push-forward of compactly supported sections}, by the following construction:
Recall that for any $\SLoc$-morphism $\chi : \bbM\to\bbM^\prime$
the corresponding $\SMan$-morphism $\chi : M\to M^\prime\vert_{\und{\chi}(\und{M})}$ is by assumption
a $\SMan$-isomorphism. In particular, we have a $\SVec$-isomorphism
\begin{flalign}
\chi_{\und{\chi}(\und{M})}^\ast : \OO_{M^\prime}(\und{\chi}(\und{M})) \longrightarrow \OO_{M}(\und{M})~,
\end{flalign}
which induces a $\SVec$-isomorphism on compactly supported sections
\begin{flalign}
(\chi_{\und{\chi}(\und{M})}^\ast)^{-1} : \OO_{M,\cc}(\und{M})\longrightarrow \OO_{M^\prime,\cc}(\und{\chi}(\und{M})) ~.
\end{flalign}
Making use of the sheaf properties of $\OO_{M^\prime}$,
we can define a $\SVec$-morphism 
\begin{flalign}
\ext_{\und{\chi}(\und{M}),\und{M^\prime}} : 
\OO_{M^\prime,\cc}(\und{\chi}(\und{M}))  \longrightarrow \OO_{M^\prime,\cc}(\und{M^\prime}) ~,
\end{flalign}
which extends compactly supported sections by zero. 
This $\SVec$-morphism is a monomorphism since
\begin{flalign}\label{eqn:resext}
\res_{\und{M^\prime},\und{\chi}(\und{M})}\circ \ext_{\und{\chi}(\und{M}),\und{M^\prime}} = \id_{\OO_{M^\prime,\cc}(\und{\chi}(\und{M}))}~.
\end{flalign}
We define the push-forward of compactly supported sections by the composition
\begin{flalign}\label{eqn:pushdefi}
\OO_{\cc}(\chi) := \chi_\ast := \ext_{\und{\chi}(\und{M}),\und{M^\prime}}\circ (\chi_{\und{\chi}(\und{M})}^\ast)^{-1} : \OO_{\cc} (\bbM) \longrightarrow \OO_{\cc}(\bbM^\prime)~
\end{flalign}
and notice that it is a $\SVec$-monomorphism.
The following lemma collects important 
properties of the push-forward of compactly supported sections, 
which shall be frequently used in our work.
\begin{lem}\label{lem:pushforward}
Let $\chi : \bbM\to\bbM^\prime$ be any $\SLoc$-morphism. Then  the following properties hold true:
\begin{itemize}
\item[(i)] $\chi^\ast \circ \chi_\ast = \id_{\OO_{\cc}(\bbM)}$ and $\chi_\ast\circ \chi^\ast(F) = F$, for 
all $F\in \OO_{\cc}(\bbM^\prime)$ such that $\supp(F)\subseteq \und{\chi}(\und{M})$.
\item[(ii)] $ \ip{F_1}{\chi_\ast(F_2)}_{\bbM^\prime}= \ip{\chi^\ast(F_1)}{F_2}_{\bbM}$, for all $F_1\in\OO(\bbM^\prime)$
and $F_2\in\OO_{\cc}(\bbM)$, where the pairing is defined in (\ref{eqn:pairing}).
\item[(iii)] ${\id_{\bbM}}_\ast = \id_{\OO_{\cc}(\bbM)}$ and $(\chi^\prime\circ\chi)_\ast = \chi^\prime_\ast \circ \chi_\ast$, for all
$\SLoc$-morphisms $\chi^\prime : \bbM^\prime\to\bbM^{\prime\prime}$.
\end{itemize}
\end{lem}
\begin{proof}
Item (i) is shown by two simple calculations. The first part follows from
\begin{flalign}
\chi^\ast \circ \chi_\ast  = \chi_{\und{\chi}(\und{M})}^\ast\circ \res_{\und{M^\prime},\und{\chi}(\und{M})}\circ \ext_{\und{\chi}(\und{M}),\und{M^\prime}}\circ (\chi_{\und{\chi}(\und{M})}^\ast)^{-1} = \id_{\id_{\OO_{\cc}(\bbM)}}~,
\end{flalign}
where in the first equality we have used the diagram (\ref{eqn:sheafmorph}) characterizing sheaf morphisms
 (applied to $U=\und{M^\prime}$ and $V=\und{\chi}(\und{M})$) and in the second equality we have used (\ref{eqn:resext}).
 To show the second part, notice that if $F\in \OO_{\cc}(\bbM^\prime)$ is such that $\supp(F)\subseteq \und{\chi}(\und{M})$,
 then $\chi^\ast(F)\in\OO_{\cc}(\bbM)$, so the composition $\chi_\ast\circ \chi^\ast(F) $ is well-defined.
By using the same argument as above we find that
 \begin{flalign}
\nn  \chi_\ast\circ \chi^\ast(F) &= \ext_{\und{\chi}(\und{M}),\und{M^\prime}}\circ (\chi_{\und{\chi}(\und{M})}^\ast)^{-1} \circ
 \chi_{\und{\chi}(\und{M})}^\ast\circ \res_{\und{M^\prime},\und{\chi}(\und{M})}(F) \\
 &=  \ext_{\und{\chi}(\und{M}),\und{M^\prime}}\circ  \res_{\und{M^\prime},\und{\chi}(\und{M})}(F) =F~,
 \end{flalign}
where the last equality holds true by direct inspection.
\sk

Item (ii) holds true as a consequence of the transformation formula for the Berezin integral  (\ref{eqn:trafoformel})
and the property (\ref{eqn:berezindensitymorph}); explicitly, we have 
\begin{flalign}
\nn \ip{F_1}{\chi_\ast(F_2)}_{\bbM^\prime} &= \int_{M^\prime}\mathrm{Ber}(E^\prime)\,F_1\, \chi_\ast(F_2) \\
\nn &=\int_{M^\prime\vert_{\und{\chi}(\und{M})}} \res_{\und{M^\prime},\und{\chi}(\und{M})}\left(\mathrm{Ber}(E^\prime)\right) ~
\res_{\und{M^\prime},\und{\chi}(\und{M})}(F_1)~\res_{\und{M^\prime},\und{\chi}(\und{M})}(\chi_\ast(F_2)) \\
\nn &=\int_{M} \mathrm{Ber}(E) ~\chi_{\und{\chi}(\und{M})}^\ast\big(\res_{\und{M^\prime},\und{\chi}(\und{M})}(F_1)\big)~\chi_{\und{\chi}(\und{M})}^\ast\big(\res_{\und{M^\prime},\und{\chi}(\und{M})}(\chi_\ast(F_2))\big) \\
&=\int_{M} \mathrm{Ber}(E) ~\chi^\ast(F_1) ~\chi^\ast(\chi_\ast(F_2)) = \ip{\chi^\ast(F_1)}{F_2}_{\bbM} ~,
\end{flalign}
for all $F_1\in\OO(\bbM^\prime)$ and $F_2\in\OO_{\cc}(\bbM)$.
In the second equality we have used that the support of $ \chi_\ast(F_2)$ is contained in $\und{\chi}(\und{M})$
and in the last equality item (i) of the present lemma.
\sk

The first part of item (iii) follows immediately from the definition (\ref{eqn:pushdefi})
and the second part from the following calculation
\begin{flalign}
\nn (\chi^\prime\circ\chi)_\ast  &= \ext_{\und{\chi^\prime}(\und{\chi}(\und{M})),\und{M^{\prime\prime}}}\circ
 ((\chi^\prime\circ\chi)_{\und{\chi^\prime}(\und{\chi}(\und{M}))}^\ast)^{-1} \\
\nn & = \ext_{\und{\chi^\prime}(\und{\chi}(\und{M})),\und{M^{\prime\prime}}}\circ \big(\chi^\ast_{\und{\chi}(\und{M})}\circ\chi_{\und{\chi^\prime}(\und{\chi}(\und{M}))}^{\prime\ast}\big)^{-1}\\
\nn &= \ext_{\und{\chi^\prime}(\und{\chi}(\und{M})),\und{M^{\prime\prime}}}\circ (\chi_{\und{\chi^\prime}(\und{\chi}(\und{M}))}^{\prime\ast})^{-1}\circ (\chi^\ast_{\und{\chi}(\und{M})})^{-1}\\
\nn &= \ext_{\und{\chi^\prime}(\und{\chi}(\und{M})),\und{M^{\prime\prime}}}\circ (\chi_{\und{\chi^\prime}(\und{\chi}(\und{M}))}^{\prime\ast})^{-1}\circ \res_{\und{M^\prime},\und{\chi}(\und{M})} \circ \ext_{\und{\chi}(\und{M}),\und{M^\prime}}\circ  (\chi^\ast_{\und{\chi}(\und{M})})^{-1}\\
\nn &= \ext_{\und{\chi^\prime}(\und{\chi}(\und{M})),\und{M^{\prime\prime}}}\circ \res_{\und{\chi^\prime}(\und{M^\prime}),\und{\chi^\prime}(\und{\chi}(\und{M}))} \circ (\chi_{\und{\chi^\prime}(\und{M^\prime})}^{\prime\ast})^{-1}\circ \ext_{\und{\chi}(\und{M}),\und{M^\prime}}\circ  (\chi^\ast_{\und{\chi}(\und{M})})^{-1}\\
\nn &= \ext_{\und{\chi^\prime}(\und{M^\prime}),\und{M^{\prime\prime}}} \circ (\chi_{\und{\chi^\prime}(\und{M^\prime})}^{\prime\ast})^{-1}\circ \ext_{\und{\chi}(\und{M}),\und{M^\prime}}\circ  (\chi^\ast_{\und{\chi}(\und{M})})^{-1}\\
&= \chi^\prime_\ast \circ \chi_\ast~,
\end{flalign}
where we have made frequent use of standard properties of sheaf morphisms.
\end{proof}

\begin{cor}\label{cor:cmpfunctor}
The following assignment is a functor $\OO_{\cc} : \SLoc \to\SVec$:  To any object
$\bbM$ in $\SLoc$ we assign the super-vector space 
$\OO_{\cc}(\bbM) := \OO_{\cc}(M)$ and to any $\SLoc$-morphism
$\chi :\bbM\to \bbM^\prime$ we assign the push-forward $\OO_{\cc}(\chi) := \chi_\ast :
\OO_{\cc}(\bbM)\to \OO_{\cc}(\bbM^\prime)$ defined in (\ref{eqn:pushdefi}).
\end{cor}
\begin{proof}
This is a direct consequence of Lemma \ref{lem:pushforward} (iii).
\end{proof}

\subsection{\label{subsec:SGOP}Properties of the super-Green's operators}
Let us fix any super-field theory according to Definition \ref{defi:sft}. By assumption,
there exists a  retarded and advanced super-Green's operator $G_{\bbM}^\pm : \OO_{\cc}(\bbM)\to \OO(\bbM)$
for the super-differential operator $P_{\bbM}: \OO(\bbM)\to\OO(\bbM)$, for all objects $\bbM$ in $\SLoc$.
We shall now derive important properties of the super-Green's operators.
\begin{lem}\label{lem:greenadjoint}
Let $\bbM$ be any object in $\SLoc$. Then
\begin{flalign}
\ip{F_1}{G_{\bbM}^\pm(F_2)}_{\bbM} = (-1)^{(\vert F_1\vert + \vert P_{\bbM}\vert)\, \vert P_{\bbM}\vert}\, \ip{G^{\mp}_{\bbM}(F_1)}{F_2}_{\bbM}~,
\end{flalign}
for all homogeneous $F_1,F_2\in\OO_{\cc}(\bbM)$.
\end{lem}
\begin{proof}
The proof follows from a short calculation
\begin{flalign}
\nn \ip{F_1}{G_{\bbM}^\pm(F_2)}_{\bbM}  &= \ip{P_{\bbM}\circ G_{\bbM}^\mp(F_1)}{G_{\bbM}^\pm(F_2)}_{\bbM}\\
 \nn &= (-1)^{(\vert F_1\vert +\vert G_{\bbM}^\mp\vert)\,\vert P_{\bbM}\vert }~\ip{G_{\bbM}^\mp(F_1)}{P_{\bbM}\circ G_{\bbM}^\pm(F_2)}_{\bbM}\\
  &= (-1)^{(\vert F_1\vert +\vert P_{\bbM}\vert)\, \vert P_{\bbM}\vert }~\ip{G_{\bbM}^\mp(F_1)}{F_2}_{\bbM}~.
\end{flalign}
The first equality holds because of property (i) of Definition \ref{defi:green}.
The integral on the right-hand side is well-defined because of property (iii) of the same Definition
and the fact that the reduced Lorentz manifold $\und{\bbM}$ is by assumption globally hyperbolic.
The second equality follows from formal super-self adjointness of $P_{\bbM}$, cf.\ (\ref{eqn:superself}).
The last equality is a consequence of property (i) of Definition \ref{defi:green}
and $\vert G_{\bbM}^\mp\vert = \vert P_{\bbM}\vert$. 
\end{proof}

We define the retarded-minus-advanced super-Green's operator
\begin{flalign}
G_{\bbM} := G_{\bbM}^+ - G_{\bbM}^- :  \OO_{\cc} (\bbM)\longrightarrow \OO_{\sc}(\bbM)\subseteq \OO(\bbM)~,
\end{flalign}
whose image lies in the super-vector space $\OO_{\sc}(\bbM)$
of spacelike compact sections (see \cite[Notation 3.4.5]{Bar:2007zz}) 
of the structure sheaf because of Definition \ref{defi:green} (iii).
This operator has the following properties.
\begin{theo}\label{theo:sequence}
Let $\bbM$ be any object in $\SLoc$. Then the sequence of linear maps
\begin{flalign}\label{eqn:sequence}
\xymatrix{
0 \ar[r] & \OO_{\cc}(\bbM) \ar[r]^-{P_{\bbM}} & \OO_{\cc}(\bbM) \ar[r]^-{G_{\bbM}} & \OO_{\sc}(\bbM)\ar[r]^-{P_{\bbM}} & \OO_{\sc}(\bbM)
}
\end{flalign}
is a complex which is exact everywhere.
\end{theo}
\begin{proof}
The proof follows easily by adapting the steps
in the proofs of \cite[Theorem 3.5]{Bar:2011iu} or \cite[Theorem 3.4.7]{Bar:2007zz} 
to our supergeometric setting. We therefore can omit the details.
\end{proof}

\begin{cor}\label{cor:greenunique}
Let $\bbM$ be any object in $\SLoc$. Then the retarded and advanced super-Green's operators
$G_{\bbM}^\pm$  for $P_{\bbM}$ are unique.
\end{cor}
\begin{proof}
Let us assume that there are two retarded/advanced super-Green's operators 
$G_{\bbM}^\pm$ and $\overline{G}_{\bbM}^\pm$
 for $P_{\bbM}$. Then, for any $F\in \OO_{\cc}(\bbM)$, 
we have that $\Phi := G_{\bbM}^\pm(F) - \overline{G}_{\bbM}^\pm(F)$
has  $\supp(\Phi)\subseteq J_{\bbM}^\pm(K)$, for some compact $K\subseteq\und{M}$,
and satisfies $P_{\bbM}(\Phi) =0$. We now show that $\Phi =0$ and hence
that $G_{\bbM}^\pm = \overline{G}_{\bbM}^\pm$ as $F\in\OO_{\cc}(\bbM)$ was arbitrary. Indeed,
we have that
\begin{flalign}
\ip{F^\prime}{\Phi}_{\bbM} = \ip{P_{\bbM}\circ G_{\bbM}^\mp (F^\prime)}{\Phi}_{\bbM} = 
(-1)^{(\vert F^\prime\vert + \vert G_{\bbM}^\mp\vert)\,\vert P_{\bbM}\vert}~\ip{G_{\bbM}^\mp (F^\prime)}{P_{\bbM}(\Phi)}_{\bbM} = 0~,
\end{flalign}
for all $F^\prime \in\OO_{\cc}(\bbM)$, which implies that $\Phi=0$.
\end{proof}

We shall now show that the retarded/advanced super-Green's operators are natural in the following sense.
\begin{lem}\label{lem:greennatural}
Let $\chi : \bbM\to\bbM^\prime$ be any morphism in $\SLoc$. Then
\begin{flalign}\label{eqn:greennatural}
G_{\bbM}^\pm = \chi^\ast \circ G_{\bbM^\prime}^{\pm} \circ \chi_\ast~
\end{flalign}
as linear maps from $\OO_{\cc}(\bbM)$ to $\OO(\bbM)$.
\end{lem}
\begin{proof}
Let us define $\overline{G}^\pm_{\bbM} := \chi^\ast\circ G_{\bbM^\prime}^{\pm} \circ \chi_\ast$.
We will show that $\overline{G}^\pm_{\bbM}$ is a retarded/advanced super-Green's operator
for $P_{\bbM}$, which due to the uniqueness result in Corollary \ref{cor:greenunique} implies
that $\overline{G}^\pm_{\bbM} =G^\pm_{\bbM}$.
\sk

We have to show that $\overline{G}^\pm_{\bbM}$ satisfies the three conditions of Definition \ref{defi:green}.
Item (i) is satisfied because of
\begin{flalign}
P_{\bbM}\circ \overline{G}^\pm_{\bbM} = P_{\bbM}\circ \chi^\ast\circ G_{\bbM^\prime}^{\pm} \circ \chi_\ast
= \chi^\ast\circ P_{\bbM^\prime} \circ G_{\bbM^\prime}^{\pm} \circ \chi_\ast = \chi^\ast\circ\chi_\ast = \id_{\OO_{\cc}(\bbM)}~,
\end{flalign}
where in the second equality we have used (\ref{eqn:Pnatural}) and in the last equality 
Lemma \ref{lem:pushforward} (i).
Item (ii) is satisfied because of
\begin{flalign}
\nn \overline{G}^\pm_{\bbM}\circ P_{\bbM}\big\vert_{\OO_{\cc}(\bbM)} &=
\chi^\ast\circ G_{\bbM^\prime}^{\pm} \circ \chi_\ast \circ P_{\bbM}\circ \chi^\ast\circ \chi_\ast =\chi^\ast\circ G_{\bbM^\prime}^{\pm} \circ \chi_\ast \circ  \chi^\ast\circ P_{\bbM^\prime}\circ \chi_\ast \\
&=\chi^\ast\circ G_{\bbM^\prime}^{\pm} \circ P_{\bbM^\prime}\circ \chi_\ast =\chi^\ast\circ \chi_\ast = \id_{\OO_{\cc}(\bbM)}~.
\end{flalign}
In the second equality we have used (\ref{eqn:Pnatural}) and in the first, third and last equality 
Lemma \ref{lem:pushforward} (i).
To show that item (iii) is satisfied we use  the same argument as in \cite[Lemma 3.2]{Bar:2011iu},
which is based on the causal compatibility of the image of the reduced 
morphism $\und{\chi} : \und{\bbM}\to\und{\bbM^\prime}$. Indeed,
\begin{flalign}
\nn \supp \big(\overline{G}^\pm_{\bbM}(F)\big) &= \supp\big(\chi^\ast\circ G_{\bbM^\prime}^{\pm} \circ \chi_\ast(F)\big) \subseteq \und{\chi}^{-1}\big(\supp\big(G_{\bbM^\prime}^{\pm} \circ \chi_\ast(F)\big)\big)\\
 &\subseteq \und{\chi}^{-1}\big( J_{\bbM^\prime}^\pm\big( \und{\chi} \big(\supp(F)\big)\big)\big)= J_{\bbM}^\pm(\supp(F))~,
\end{flalign}
for all $F\in\OO_{\cc}(\bbM)$.
\end{proof}

\subsection{The functor \texorpdfstring{$\LL : \SLoc \to\XX$}{L : SLoc -> X}}
Let us fix any super-field theory according to Definition \ref{defi:sft}. 
For any object $\bbM$ in $\SLoc$ we can define a linear map
\begin{flalign}\label{eqn:taupre}
\tau_{\bbM} : \OO_{\cc}(\bbM) \otimes \OO_{\cc}(\bbM)\longrightarrow \bbR~,~~F_1\otimes F_2 \longmapsto\tau_{\bbM}(F_1,F_2)
=\ip{G_{\bbM}(F_1)}{F_2}_{\bbM}~,
\end{flalign}
where $G_{\bbM} := G_{\bbM}^+ - G_{\bbM}^-$ is the retarded-minus-advanced
super-Green's operator and $\ip{\,\cdot\,}{\,\cdot\,}_{\bbM}$ is the pairing (\ref{eqn:pairing}).
Since, by definition, the $\bbZ_2$-parity of $G_{\bbM}$ agrees with  that of the pairing,
the linear map $\tau_{\bbM}$ is even and hence a $\SVec$-morphism.
As a consequence of (\ref{eqn:pairingsymmetry}) and Lemma \ref{lem:greenadjoint}
we find that
\begin{flalign}
\nn \tau_{\bbM}(F_1,F_2) &=(-1)^{\vert P_{\bbM}\vert +1} \, (-1)^{\vert F_1\vert\,\vert F_2\vert} \,\tau_{\bbM}(F_2,F_1)\\
&=\begin{cases}
- \,(-1)^{\vert F_1\vert\,\vert F_2\vert}\,\tau_{\bbM}(F_2,F_1) & ~,~~\text{ for $\dim(S)$ even}~,\\
(-1)^{\vert F_1\vert\,\vert F_2\vert}\,\tau_{\bbM}(F_2,F_1) & ~,~~\text{ for $\dim(S)$ odd}~,
\end{cases}
\end{flalign}
for all homogeneous $F_1,F_2\in\OO_{\cc}(\bbM)$.
Hence, $\tau_{\bbM}$ is super-skew symmetric if $\dim(S)$ is 
even and super-symmetric if $\dim(S)$ is odd.
\sk

Let us recall that by Theorem \ref{theo:sequence} the kernel
of the linear map $G_{\bbM} : \OO_{\cc}(\bbM)\to \OO_{\sc}(\bbM)$
coincides with the image of $P_{\bbM} : \OO_{\cc}(\bbM)\to\OO_{\cc}(\bbM)$.
As a consequence, the $\SVec$-morphism $\tau_{\bbM}$ defined in (\ref{eqn:taupre}) descends
to the $\SVec$-morphism (denoted with a slight abuse of notation by the same symbol)
\begin{flalign}
\tau_{\bbM} : \frac{\OO_{\cc}(\bbM)}{P_{\bbM}(\OO_{\cc}(\bbM))}\otimes \frac{\OO_{\cc}(\bbM)}{P_{\bbM}(\OO_{\cc}(\bbM))} \longrightarrow \bbR~,~~[F_1]\otimes[F_2]\longmapsto \ip{G_{\bbM}(F_1)}{F_2}_{\bbM}~,
\end{flalign}
which is weakly non-degenerate, i.e.\ $\tau_{\bbM}([F_1],[F_2]) =0$ for all
$[F_1]\in \OO_{\cc}(\bbM)/P_{\bbM}(\OO_{\cc}(\bbM))$ implies that $[F_2] =0$.
The pair
\begin{flalign}\label{eqn:SSspace}
\LL(\bbM) := \left(\frac{\OO_{\cc}(\bbM)}{P_{\bbM}(\OO_{\cc}(\bbM))},\tau_{\bbM}\right) 
\end{flalign}
is therefore a super-symplectic space if $\dim(S)$ is even and a super-inner product space if
$\dim(S)$ is odd.
\sk

We shall now show that the assignment (\ref{eqn:SSspace}) is functorial.
For this we introduce the following category which depends on the choice of
super-field theory via $\dim(S)~\mathrm{mod}~2$.
\begin{defi} \label{defi:XX}
The category $\XX$  consists of the following objects and morphisms: 
\begin{itemize}
\item The objects are all pairs $\bbV:=(V,\tau)$ consisting of a real super-vector space $V$ 
and a weakly non-degenerate $\SVec$-morphism
$\tau : V\otimes V\to \bbR$, which is super-skew symmetric if $\dim(S)$ is even and super-symmetric
if $\dim(S)$ is odd, i.e.\
\begin{flalign}
\tau(v_1,v_2) = (-1)^{\dim(S)+1}\,(-1)^{\vert v_1\vert\,\vert v_2\vert}\,\tau(v_2,v_1)~,
\end{flalign}
for all homogeneous $v_1,v_2\in V$.

\item The morphisms $L : \bbV\to\bbV^\prime$ are all
$\SVec$-morphisms (denoted by the same symbol) $L : V\to V^\prime$ satisfying
$\tau^\prime \circ (L\otimes L) = \tau$.
\end{itemize}
\end{defi}

\begin{propo}\label{propo:LLfunctor}
The following assignment is a functor $\LL : \SLoc \to \XX$: 
To any object $\bbM$ in $\SLoc$ we assign the object
$\LL(\bbM)$ in $\XX$ given by (\ref{eqn:SSspace}) and to any $\SLoc$-morphism
$\chi : \bbM\to\bbM^\prime$ we assign the $\XX$-morphism
\begin{flalign}\label{eqn:SSmap}
\LL(\chi) : \LL(\bbM)\longrightarrow \LL(\bbM^\prime)~,~~[F]\longmapsto [\chi_\ast(F)]~.
\end{flalign}
\end{propo}
\begin{proof}
We have already seen above that $\LL(\bbM)$ is an object in $\XX$.
It remains to show that (\ref{eqn:SSmap}) is well-defined and an $\XX$-morphism.
It is well-defined since 
\begin{flalign}
\nn \chi_\ast\circ P_{\bbM}(F) &= \chi_\ast\circ P_{\bbM}\circ\chi^\ast\circ\chi_\ast(F) 
= \chi_\ast\circ\chi^\ast\circ P_{\bbM^\prime}\circ\chi_\ast(F) \\
&= P_{\bbM^\prime}\circ\chi_\ast(F) \in P_{\bbM^\prime}(\OO_{\cc}(\bbM^\prime))~,
\end{flalign}
for all $F\in\OO_{\cc}(\bbM)$.
Moreover it is an $\XX$-morphism since 
\begin{flalign}
\nn \tau_{\bbM^\prime}\big([\chi_\ast(F_1)],[\chi_\ast(F_2)]\big) &= 
\ip{ G_{\bbM^\prime} \circ\chi_\ast(F_1)}{ \chi_\ast(F_2)}_{\bbM^\prime}=
\ip{ \chi^\ast\circ G_{\bbM^\prime} \circ \chi_\ast(F_1)}{F_2}_{\bbM}\\
&= \ip{ G_{\bbM} (F_1)}{F_2}_{\bbM} = \tau_{\bbM}\big([F_1],[F_2]\big) ~,
\end{flalign}
for all $[F_1],[F_2]\in\OO_{\cc}(\bbM)/P_{\bbM}(\OO_{\cc}(\bbM))$.
In the second equality we have used Lemma \ref{lem:pushforward} (ii)
and in the third equality we have used Lemma \ref{lem:greennatural}.
The functoriality of $\LL$ is induced by the functoriality of $\OO_{\cc} : \SLoc\to\SVec$
which has been established in Corollary \ref{cor:cmpfunctor}.
\end{proof}

We finish this subsection by proving some properties of the functor
$\LL : \SLoc\to \XX$, which are the axioms of locally covariant
quantum field theory \cite{Brunetti:2001dx} applied to classical theories.
\begin{theo}\label{theo:LCCFT}
For any super-field theory according to Definition \ref{defi:sft} the
associated functor $\LL : \SLoc\to\XX$ satisfies the following properties:
\begin{itemize}
\item Locality: For any $\SLoc$-morphism $\chi :\bbM\to\bbM^\prime$, the $\XX$-morphism 
$\LL(\chi) : \LL(\bbM)\to\LL(\bbM^\prime)$ is monic.

\item Super-causality: Given two $\SLoc$-morphisms 
$\bbM_1 \stackrel{\chi_1}{\longrightarrow}\bbM \stackrel{\chi_2}{\longleftarrow}\bbM_2$ such that
the images of the reduced $\tLor$-morphisms $\und{\bbM_1} \stackrel{\und{\chi_1}}{\longrightarrow}\und{\bbM} 
\stackrel{\und{\chi_2}}{\longleftarrow}\und{\bbM_2}$ are causally disjoint, 
then 
\begin{flalign}
\tau_{\bbM}\left(\LL(\chi_1)\big(\LL(\bbM_1)\big) , \LL(\chi_2)\big(\LL(\bbM_2)\big) \right) =\{0\}~.
\end{flalign}

\item Time-slice axiom: Given any Cauchy $\SLoc$-morphism\footnote{
A Cauchy $\SLoc$-morphism is a $\SLoc$-morphism $\chi: \bbM\to\bbM^\prime$
such that its reduced $\tLor$-morphism $\und{\chi}: \und{\bbM}\to\und{\bbM^\prime}$ is Cauchy,
i.e.\ the image of $\und{\chi}$ contains a Cauchy surface in $\und{\bbM^\prime}$.
} $\chi: \bbM\to\bbM^\prime$, then $\LL(\chi) : \LL(\bbM)\to\LL(\bbM^\prime)$ is an isomorphism.

\end{itemize}
\end{theo}
\begin{proof}
The locality property is as usual a consequence of $\tau_{\bbM}$ being weakly non-degenerate; indeed,
assuming that $\LL(\chi)([F]) =0$, for some $[F]\in\LL(\bbM)$, we find that
\begin{flalign}
\tau_{\bbM^\prime}\big(\LL(\chi)([H]) ,\LL(\chi)([F]) \big) = \tau_{\bbM}\big([H],[F]\big) =0~,
\end{flalign}
for all $[H]\in \LL(\bbM)$, and hence $[F]=0$.
The super-causality property is a consequence of the support properties of the super-Green's operators, 
cf.\ Definition \ref{defi:green} (iii).
\sk

To show the time-slice axiom, it remains to prove that $\LL(\chi)$ is surjective for any Cauchy
$\SLoc$-morphism, which is equivalent to proving that any class $[F]\in \LL(\bbM^\prime)$
has a representative $F^\prime\in\OO_{\cc}(\bbM^\prime)$ with support contained
in $\und{\chi}(\und{\bbM})\subseteq \und{\bbM^\prime}$. This follows from a standard argument
which we shall now generalize to the case of supermanifolds.
Let us take any two non-intersecting 
Cauchy surfaces $\Sigma^\pm \subset \und{\chi}(\und{\bbM})\subseteq\und{\bbM^\prime}$ of $\und{\bbM^\prime}$,
 such that $\Sigma^+$ lies in the chronological future of $\Sigma^-$. Take the open cover
 $U^\pm := I_{\bbM^\prime}^\pm (\Sigma^\mp)$ of $\und{\bbM^\prime}$
 and choose some partition of unity $\rho^\pm \in\OO(\bbM^\prime)$ subordinated to this cover (see
 \cite[Proposition 4.2.7]{Carmeli} for a proof of existence of partitions of unity on supermanifolds).
We choose any representative $F\in\OO_{\cc}(\bbM^\prime)$ of the class $[F]\in\LL(\bbM^\prime)$
and define
\begin{flalign}
H:= \rho^-\,G_{\bbM^\prime}^+(F) + \rho^+\,G_{\bbM^\prime}^-(F)\in\OO_{\cc}(\bbM^\prime)~.
\end{flalign}
Then $F^\prime := F-P_{\bbM^\prime}(H)\in\OO_{\cc}(\bbM^\prime)$ is a representative of the class $[F]$
with support contained in $\und{\chi}(\und{\bbM})\subseteq\und{\bbM^\prime}$.
\end{proof}

\subsection{The quantization functor \texorpdfstring{$\QQ : \XX \to \SAlg $}{Q : X -> S*Alg}}
The quantization is performed by assigning to objects $\bbV$
in $\XX$ SCCR superalgebras in the
case of $\dim(S)$ even and SCAR superalgebras
in the case of $\dim(S)$ odd. This reflects the fact that the objects in $\XX$ are
super-symplectic spaces if $\dim(S)$ is even and super-inner product spaces
if $\dim(S)$ is odd. We can perform this construction in one step by
using suitable sign and imaginary unit $\ii\in\bbC$ factors (depending on $\dim(S)\,\mathrm{mod}\,2$) in the definitions below.
\sk

Let $\bbV = (V,\tau)$ be any object in $\XX$. We consider the complexified tensor superalgebra
\begin{flalign}
\mathcal{T}_\bbC(\bbV) := \bigoplus_{n\geq 0} \mathcal{T}_{\bbC}^n(\bbV) := \bigoplus_{n\geq 0} V^{\otimes n}\otimes \bbC
\end{flalign}
and denote its product simply by juxtaposition. Notice that 
$\mathcal{T}_\bbC(\bbV)$ is generated (over $\bbC$) by the unit $\1 :=1\in\mathcal{T}_\bbC^0(\bbV)\simeq \bbC$
and the elements $v \cong v\otimes 1 \in \mathcal{T}_\bbC^1(\bbV) = V\otimes \bbC$, for all $v\in V$.
We equip $\mathcal{T}_\bbC(\bbV)$ with the superinvolution which is defined on the generators
by $\1^\ast =\1$ and $v^\ast = v$, for all $v\in V$,
and extended to all of $\mathcal{T}_\bbC(\bbV)$ by $\bbC$-antilinearity and
\begin{flalign}
(a_1\,a_2)^\ast =(-1)^{\vert a_1\vert\,\vert a_2\vert}\,a_2^\ast\,a_1^\ast ~,
\end{flalign}
for all $\bbZ_2$-parity homogeneous $a_1,a_2\in \mathcal{T}_\bbC(\bbV)$. Using the $\SVec$-morphism
$\tau : V\otimes V\to\bbR$, we define $\mathcal{I}(\bbV)$ to be the two-sided
super-$\ast$-ideal in $\mathcal{T}_\bbC(\bbV)$ that is generated by the elements
\begin{flalign}\label{eqn:comrel}
v_1\,v_2  + (-1)^{\dim(S)+1} \,(-1)^{\vert v_1\vert\,\vert v_2\vert}\, v_2\,v_1 - \beta \,\tau(v_1,v_2)\,\1~,
\end{flalign}
for all homogeneous generators $v_1,v_2\in V$,
where we have used
\begin{flalign}
\beta:=\begin{cases}
\ii &~,~~\text{ for $\dim(S)$ even}~,\\
1 &~,~~\text{ for $\dim(S)$ odd}~.
\end{cases}
\end{flalign}
Notice that (\ref{eqn:comrel}) describes SCCR
if $\dim(S)$ is even and SCAR
if $\dim(S)$ is odd.
We define the object $\QQ(\bbV)$ in $\SAlg$ by the quotient
\begin{flalign}\label{eqn:Qalg}
\QQ(\bbV) := \frac{\mathcal{T}_{\bbC}(\bbV)}{\mathcal{I}(\bbV)}~.
\end{flalign}
\begin{propo}
The following assignment is a functor $\QQ : \XX \to \SAlg$: To any object $\bbV$ in $\XX$
we assign the object $\QQ(\bbV)$ in $\SAlg$ given by (\ref{eqn:Qalg}) and to any 
$\XX$-morphism $L : \bbV \to\bbV^\prime$ we assign the $\SAlg$-morphism
$\QQ(L) : \QQ(\bbV) \to \QQ(\bbV^\prime)$ that is specified by
defining on the generators $\QQ(L)(v) := L(v)$, for all $v\in V$.
\end{propo}
\begin{proof}
By our constructions above, we already know that $\QQ(\bbV)$ is an object in $\SAlg$, for all
objects $\bbV$ in $\XX$. It remains to show that $\QQ(L) : \QQ(\bbV) \to \QQ(\bbV^\prime)$
specified above is well-defined, i.e.\ that it preserves the two-sided super-$\ast$-ideals.
This is a consequence of the explicit form of the generators of these ideals (\ref{eqn:comrel})
and the fact that $L : \bbV\to\bbV^\prime$ satisfies $\tau^\prime\circ (L\otimes L) = \tau$.
\end{proof}

\subsection{\label{subsec:lcqft}The locally covariant quantum field theory \texorpdfstring{$\AA : \SLoc \to \SAlg $}{A : SLoc -> S*Alg}}
We compose the functors $\LL : \SLoc\to \XX$ and $\QQ : \XX \to \SAlg$ in order to define
the functor
\begin{flalign}
\AA := \QQ \circ \LL : \SLoc \longrightarrow \SAlg~.
\end{flalign}
This functor satisfies a supergeometric modification of the axioms of locally covariant quantum field theory \cite{Brunetti:2001dx}.
\begin{theo}\label{theo:LCQFT}
For any super-field theory according to Definition \ref{defi:sft} the
associated functor $\AA : \SLoc\to\SAlg$ satisfies the following properties:
\begin{itemize}
\item Locality: For any $\SLoc$-morphism $\chi : \bbM \to \bbM^\prime$, the $\SAlg$-morphism
$\AA(\chi) : \AA(\bbM)\to\AA(\bbM^\prime)$ is monic.

\item Super-causality: Given two $\SLoc$-morphisms
 $\bbM_1\stackrel{\chi_1}{\longrightarrow}\bbM \stackrel{\chi_2}{\longleftarrow} \bbM_2$
 such that the images of the reduced $\tLor$-morphisms 
 $\und{\bbM_1}\stackrel{\und{\chi_1}}{\longrightarrow}\und{\bbM} \stackrel{\und{\chi_2}}{\longleftarrow} \und{\bbM_2}$
 are causally disjoint, then
 \begin{flalign}\label{eqn:supercausaltmp}
 a_1\,a_2 + (-1)^{\dim(S) +1} \,(-1)^{\vert a_1\vert\,\vert a_{2}\vert} \,a_2\,a_1 =0~,
 \end{flalign}
 for all homogeneous $a_1\in\AA(\chi_1)(\AA(\bbM_1))$ and $a_2\in \AA(\chi_2)(\AA(\bbM_2))$.
 
 \item Time-slice axiom: Given any Cauchy $\SLoc$-morphism $\chi: \bbM\to\bbM^\prime$,
  then $\AA(\chi) : \AA(\bbM)\to\AA(\bbM^\prime)$ is a $\SAlg$-isomorphism.
\end{itemize}
\end{theo}
\begin{proof}
All properties listed above follow by standard arguments from the corresponding properties of the classical theory given in
Theorem \ref{theo:LCCFT}. Let us briefly give a sketch or reference:
The locality property follows by using the techniques summarized in \cite[Appendix A]{Fewster:2011pn}.
The super-causality property for $a_1$ and $a_2$ being two generators 
follows explicitly from the form of the two-sided super-$\ast$-ideal (\ref{eqn:comrel})
and for generic $a_1$ and $a_2$ by expressing these elements in terms of generators and using
iteratively the super-causality property for the generators. 
The time-slice axiom for $\AA  = \QQ\circ \LL$ follows since functors (here $\QQ$) preserve isomorphisms.
\end{proof}
\begin{rem}\label{rem:evensubtheory}
Notice that the super-causality property (\ref{eqn:supercausaltmp}) is similar to the graded-causality property
encountered in fermionic quantum field theories, see e.g.\ \cite{Bar:2011iu}. In particular, if
$\dim(S)$ is even then (\ref{eqn:supercausaltmp}) implies that two even elements 
commute and two odd elements anti-commute whenever they are spacelike separated.
In an ordinary (i.e.\ non-supergeometric) quantum field theory one usually postulates 
that only the even elements of the algebras are true physical observables,
which includes in particular bilinear terms in fermionic quantum fields such as the stress-energy tensor.
A similar construction is also possible in our present description of super-QFTs:
We may define a new functor $\AA_0 : \SLoc \to {}^\ast\mathsf{Alg}$ (the even part of the super-QFT 
$\AA : \SLoc\to\SAlg$) to the category of ordinary $\ast$-algebras by assigning to any object $\bbM$ in $\SLoc$
the even sub-$\ast$-algebra $\AA_0(\bbM)$ of $\AA(\bbM)$ and to any $\SLoc$-morphism
$\chi : \bbM\to \bbM^{\prime}$ the restriction of $\AA(\chi) : \AA(\bbM)\to \AA(\bbM^{\prime}) $
to $\AA_0(\bbM)$. (Because $\AA(\chi)$ preserves $\bbZ_2$-parity we have that 
$\AA_0(\chi) :=\AA(\chi)\vert_{\AA_0(\bbM)}  : \AA_0(\bbM)\to \AA_0(\bbM^{\prime}) $ maps to $\AA_0(\bbM^{\prime})$.)
The functor $\AA_0 : \SLoc \to {}^\ast\mathsf{Alg}$  then satisfies the ordinary locality, causality and time-slice axiom
of locally covariant quantum field theory \cite{Brunetti:2001dx}.
\end{rem}

We conclude this section by constructing a locally covariant quantum field 
for the functor $\AA : \SLoc\to \SAlg$, which establishes a connection to the data specifying
a super-field theory in Definition \ref{defi:sft}. Let us consider the functor $\OO_{\cc} :\SLoc\to \SVec$
and regard $\AA$ also as a functor to $\SVec$ (denoted by the same symbol) by composing
with the forgetful functor. There is a natural transformation $\Phi : \OO_{\cc}\Rightarrow \AA$ of functors
from $\SLoc$ to $\SVec$ with components given by the $\SVec$-morphisms
\begin{flalign}\label{eqn:quantumfield}
\Phi_{\bbM} : \OO_{\cc}(\bbM)\longrightarrow \AA(\bbM)~,~~F\longmapsto [F]~.
\end{flalign}
In analogy to \cite{Brunetti:2001dx} we shall interpret $\Phi_{\bbM}(F)\in\AA(\bbM)$, for $F\in\OO_{\cc}(\bbM)$,
as a smeared linear hermitian superfield operator on $\bbM$. The connection to the data in Definition \ref{defi:sft}
is established by noticing that the superfield operators satisfy the equations of motion (in weak form)
\begin{flalign}
 \Phi_{\bbM}(P_{\bbM}(F))=0~,
\end{flalign}
for all $F\in\OO_{\cc}(\bbM)$ and all objects $\bbM$ in $\SLoc$. Moreover, they satisfy
the super-canonical (anti)commutation relations
\begin{flalign}
\Phi_{\bbM}(F_1) \,\Phi_{\bbM}(F_2) + (-1)^{\dim(S) +1} \,(-1)^{\vert F_1\vert \,\vert F_2\vert}\,
\Phi_{\bbM}(F_2) \,\Phi_{\bbM}(F_1) = \beta \,\tau_{\bbM}(F_1,F_2)~,
\end{flalign}
for all homogeneous $F_1,F_2\in\OO_{\cc}(\bbM)$ and all objects $\bbM$ in $\SLoc$.
We recall that $\beta =\ii$ if $\dim(S)$ is even and $\beta =1$ if $\dim(S)$ is odd.

\begin{rem}\label{rem:componentfields}
For any object $\bbM$ in $\SLoc$ we have a
decomposition $\OO_{\cc}(\bbM) = \OO_{\cc}(\bbM)_0 \oplus \OO_{\cc}(\bbM)_1$
into the even and odd part. We can define new super-vector spaces
$\OO_{\cc}^\mathrm{even} (\bbM) := \OO_{\cc}(\bbM)_0 \oplus 0 $
and $\OO_{\cc}^\mathrm{odd}(\bbM) := 0\oplus\OO_{\cc}(\bbM)_1$
and notice that $\OO_{\cc}^\mathrm{even} : \SLoc \to \SVec$ and $\OO_{\cc}^\mathrm{odd} : \SLoc \to \SVec$
are subfunctors of $\OO_{\cc} :\SLoc\to \SVec$. (The latter statement is due to the fact that the push-forwards
$\OO_{\cc}(\chi) = \chi_\ast$ preserve the $\bbZ_2$-parity.) Consequently, 
our locally covariant quantum field $\Phi : \OO_{\cc}\Rightarrow \AA$
 decomposes into two natural transformation $\Phi^\mathrm{even} : \OO_{\cc}^\mathrm{even}\Rightarrow \AA$ 
 and $\Phi^\mathrm{odd} : \OO_{\cc}^\mathrm{odd}\Rightarrow \AA$,
 which describe within our physical interpretation the even and odd component quantum fields
 of the superfield $\Phi$. The appearance of the even and odd quantum fields
 is an undesirable feature, which indicates that our formulation does not appropriately
capture supersymmetry transformations. In fact, supersymmetry transformations
are supposed to mix the even and odd component fields, hence allowing neither of them to be
 a natural transformation, i.e.\ a locally covariant quantum field.
It is the goal of the next section
to `enrich' (in a mathematically precise way) the categories and functors appearing in
our construction in order to capture also supersymmetry transformations.
\end{rem}

%%%%%%%%%%%%%%%%%%%%%%%%%%%%%%%%%%%%%%%%%%%%%%%%
%%%%%%%%%%%%%%%%%%%%%%%%%%%%%%%%%%%%%%%%%%%%%%%%

\section{\label{sec:axiomsenriched}Axiomatic definition of enriched super-field theories}
Motivated by the shortcomings of our present theory, which have been
summarized in Remark \ref{rem:componentfields}, we shall now systematically `enrich' all categories,
functors and natural transformations appearing in the Definition \ref{defi:sft} of super-field theories. 
A suitable mathematical framework is that of {\em enriched category theory}, see e.g.\ \cite{Enriched,Enriched2}
and also Appendix \ref{app:enriched} for a brief introduction to the basic concepts.
Loosely speaking, in an ordinary category the morphisms between two objects have to form a set
and in an enriched category the morphisms between two objects are allowed
to be an object in another (monoidal) category. Enriched functors and natural transformations
are then defined by a suitable generalization of the standard concepts of functors and natural transformations
in ordinary category theory. In our definition of enriched super-field theories, as well as in the construction
of the corresponding enriched super-QFTs in Section \ref{sec:enriched}, we shall consider enriched categories
over the monoidal category $\SSet$ of supersets, which we also call $\SSet$-categories. 
Again loosely speaking, while an ordinary set
is determined by its points, a superset is determined by its superpoints. To make precise the notion of supersets,
we shall use the category theoretical approach to supergeometry proposed by Schwarz \cite{Schwarz} and 
Molotkov \cite{Molotkov}, and developed in detail by Sachse \cite{Sachse}, see also \cite{SachseWockel,AL}.

\subsection{The monoidal category \texorpdfstring{$\SSet$}{SSet} of supersets}
For better understanding the concept of supersets, it will be helpful to view ordinary sets 
from a categorical perspective. Let $A$ be any set. Then $A$ is determined by its
points, which can be described by maps $x : \mathrm{pt} \to A$ from a (once and for all fixed) 
singleton set $\mathrm{pt}:=\{\star\}$ to the set $A$. In other words,
the points of $A$ are described by the morphism set $\Hom_{\Set}(\mathrm{pt},A)$.
Using the usual composition of maps, any map between two sets $f: A\to B$ 
induces a map between the morphism sets
\begin{flalign}\label{eqn:setfunctormap}
\Hom_{\Set}(\mathrm{pt},A) \longrightarrow \Hom_{\Set}(\mathrm{pt},B)~,~~x\longmapsto f \circ x~.
\end{flalign}
Let $\mathsf{Pt}$ be the category consisting of the single object $\mathrm{pt}$ and the single morphism 
$\id_{\mathrm{pt}}$. Then the morphism set above can be regarded as  a functor
$\Hom_{\Set}(\,\cdot\,,A) : \mathsf{Pt}^\op \to \Set$ (in foresight we use here the opposite category $\mathsf{Pt}^\op$)
and the map (\ref{eqn:setfunctormap}) as a natural transformation 
$\Hom_{\Set}(\,\cdot\,,A)\Rightarrow \Hom_{\Set}(\,\cdot\,,B)$ between functors from $\mathsf{Pt}^\op$ to $\Set$.
What this means is that we have constructed a functor
\begin{flalign}\label{eqn:setfunctor}
\Set \longrightarrow \mathrm{Fun}(\mathsf{Pt}^\op , \Set)
\end{flalign}
from the category of sets to the category of functors from $\mathsf{Pt}^\op$ to $\Set$.
Notice that the functor (\ref{eqn:setfunctor}) is an 
equivalence between the categories $\Set$ and $\mathrm{Fun}(\mathsf{Pt}^{\op},\Set)$.
In other words, we can choose freely if we want to work with the usual category $\Set$ of sets 
or with the functor category $\mathrm{Fun}(\mathsf{Pt}^{\op},\Set)$.
\sk

Motivated by this functorial point of view, we shall define the category of supersets
as the functor category $\mathrm{Fun}(\SPt^{\op},\Set)$, where $\SPt$ is the following category of superpoints:
\begin{defi}
The category $\SPt$ consists of the following objects and morphisms:
\begin{itemize}
\item The objects are given by the supermanifolds 
$\pt_n := (\mathrm{pt},\Lambda_n)$, where $\Lambda_n :=\bigwedge^\bullet \bbR^n$ 
is the real Grassmann algebra over $\bbR^n$ and $n\in \bbN^0$.
\item The morphisms $\lambda : \pt_n \to \pt_m$ are
all supermanifold morphisms.
\end{itemize}
The category $\SSet$ of supersets is defined as the functor category
\begin{flalign} \label{eqn:DefSSet}
\SSet := \mathrm{Fun}(\SPt^\op , \Set) ~.
\end{flalign}
\end{defi}
\begin{rem}
In \cite{Sachse}, the category of superpoints is defined
as the full subcategory of $\SMan$ with objects given by all 
supermanifolds whose underlying topological space is a singleton.
This category is equivalent to our category $\SPt$ and moreover we have
that $\SPt^\op$ is equivalent to the category of finite-dimensional real Grassmann 
algebras $\Gr$. Hence, our category of supersets (\ref{eqn:DefSSet}) is equivalent
to the functor category $\mathrm{Fun}(\Gr,\Set)$, which is used for example in \cite{Sachse}.
\end{rem}
\sk

Recall that the category $\Set$ of ordinary sets is a monoidal category with bifunctor
$\times : \Set\times\Set\to\Set$ given by the Cartesian product and unit object
given by the singleton set $\pt$. By a general construction, the monoidal structure
on $\Set$ induces a monoidal structure on the functor category $\mathrm{Fun}(\SPt^\op , \Set)$
and hence on the category $\SSet$ of supersets. Let us briefly recall this construction and give explicit formulas.
We define a bifunctor (denoted with a slight abuse of notation also by $\times$)
\begin{flalign}
\times :\SSet\times\SSet \longrightarrow \SSet
\end{flalign}
by assigning to any object $(\mathfrak{F} : \SPt^\op \to\Set, \mathfrak{F}^\prime : \SPt^\op \to \Set )$ 
in $\SSet\times \SSet$ the object $\mathfrak{F}\times \mathfrak{F}^\prime : \SPt^\op\to \Set$ 
in $\SSet$, which is the functor specified on objects $\pt_n$ in $\SPt^\op$ by 
\begin{subequations}\label{eqn:monoidalSSets}
\begin{flalign}
(\mathfrak{F}\times \mathfrak{F}^\prime)(\pt_n) := \mathfrak{F}(\pt_n)\times \mathfrak{F}^\prime(\pt_n)
\end{flalign}
and on $\SPt^\op$-morphisms $\lambda^\op : \pt_n\to \pt_m$  by
\begin{flalign}
(\mathfrak{F}\times \mathfrak{F}^\prime)(\lambda^\op) := \mathfrak{F}(\lambda^\op)\times \mathfrak{F}^\prime(\lambda^\op)  : (\mathfrak{F}\times \mathfrak{F}^\prime)(\pt_n) \longrightarrow (\mathfrak{F}\times \mathfrak{F}^\prime)(\pt_m)~.
\end{flalign}
To any morphism $(\eta : \mathfrak{F}\Rightarrow \mathfrak{G},\eta^\prime : \mathfrak{F}^\prime \Rightarrow \mathfrak{G}^\prime)$
in $\SSet\times \SSet$ we assign the morphism 
$\eta\times\eta^\prime : \mathfrak{F}\times\mathfrak{F}^\prime \Rightarrow
 \mathfrak{G}\times\mathfrak{G}^\prime$ in $\SSet$ which is given by the natural transformation with components
\begin{flalign}
(\eta\times\eta^\prime)_{\pt_n} := \eta_{\pt_n} \times \eta^\prime_{\pt_n} :  (\mathfrak{F}\times \mathfrak{F}^\prime)(\pt_n) \longrightarrow  (\mathfrak{G}\times \mathfrak{G}^\prime)(\pt_n) ~,
\end{flalign}
\end{subequations}
for all objects $\pt_n$ in $\SPt^\op$. The unit object in $\SSet$ is the functor
$\mathfrak{I} : \SPt^\op\to \Set$ specified on objects $\pt_n$  and morphisms $\lambda^\op : \pt_n\to \pt_m$ in $\SPt^\op$ by
\begin{flalign}\label{eqn:unitSSets}
\mathfrak{I}(\pt_n) :=\pt
~~,\quad \mathfrak{I}(\lambda^\op) := \id_{\pt}~.
\end{flalign}
In summary, we have
\begin{propo}
The category $\SSet$ of supersets is a monoidal category with bifunctor
$\times : \SSet\times\SSet \to\SSet$ defined by (\ref{eqn:monoidalSSets}) and unit object
$\mathfrak{I}$ defined by (\ref{eqn:unitSSets}).
\end{propo}

\subsection{The \texorpdfstring{$\SSet$}{SSet}-category \texorpdfstring{$\eSLoc$}{eSLoc}}
Let us choose as in Definition \ref{defi:sft} any full subcategory $\SLoc$ of $\ghSCart$.
The goal of this subsection is to define a $\SSet$-category
$\eSLoc$, such that the objects in $\eSLoc$ coincide with those in $\SLoc$
and that the morphism supersets in $\eSLoc$ enrich (in a suitable way)
the ordinary morphism sets in $\SLoc$. The main feature of this enrichment will be that
supersymmetry transformations appear as superpoints of the morphism supersets, see Section \ref{sec:examples}
for explicit examples.
\sk

Before we can define the $\SSet$-category $\eSLoc$ we need some preparations. A 
supermanifold $M$  can be described in the framework 
of supersets (\ref{eqn:DefSSet}) by the functor $\Hom_{\SMan}(\,\cdot\,,M): \SPt^\op \to \Set$. 
We will not describe the details of this approach (see \cite{Schwarz,Molotkov,Sachse,SachseDiss}), 
but we make use of an equivalent picture: $\Hom_\SMan(\pt_n,M)$ clearly coincides with the set sections 
of the trivial super-fibre bundle $\pt_n \times M \to \pt_n$ and natural transformations $\Hom_\SMan(\,\cdot\,,M) 
\to \Hom_\SMan(\,\cdot\,,M^\prime)$ correspond to super-fibre bundle morphisms. We will discuss the basic properties of this 
``family point of view'' and refer to the literature \cite[Chapter 3.3]{SachseDiss} and \cite[\S 2.8 and \S 2.9]{Deligne} 
for more details.
\sk

Given any object $M$ in $\SMan$ and any object $\pt_n$ in $\SPt^\op$, we can consider 
the product supermanifold $\pt_n\times M = (\und{M}, \Lambda_n\otimes \OO_M)$ together with the projection
$\SMan$-morphism $\mathrm{pr}_{\pt_n\times M,\pt_n} :\pt_n\times M\to\pt_n$ onto the first factor.
The pair $(\pt_n\times M, \mathrm{pr}_{\pt_n\times M,\pt_n})$ is typically 
called {\em a $\pt_n$-relative supermanifold} and denoted by the compact notation $M/\pt_n$.
A {\em a $\pt_n$-relative $\SMan$-morphism} (in short:  $\SMan/\pt_n$-morphism)
$\chi : M/\pt_n \to M^\prime/\pt_n$ between two $\pt_n$-relative supermanifolds
is a $\SMan$-morphism $\chi : \pt_n\times M\to \pt_n\times M^\prime$ between the product supermanifolds
which preserves the projections, i.e.\ the diagram
\begin{flalign}\label{diag:relmorph}
\xymatrix{
\ar[dr]_-{\mathrm{pr}_{\pt_n\times M,\pt_n}~~~} \pt_n\times M \ar[rr]^-{\chi} && \pt_n\times M^\prime \ar[dl]^-{~~~\mathrm{pr}_{\pt_n\times M^\prime,\pt_n}}\\
&\pt_n&
}
\end{flalign}
in $\SMan$ commutes.
Explicitly, a $\SMan$-morphism $\chi : \pt_n\times M \to \pt_n\times M^\prime$ is a $\SMan/\pt_n$-morphism
if and only if $\chi^\ast (\zeta\otimes \1) = \zeta\otimes \1$, for all $\zeta\in\Lambda_n$.
Notice that the identity $\id_{\pt_n\times M} : M/\pt_n\to M/\pt_n$ is a $\SMan/\pt_n$-morphism
and that two $\SMan/\pt_n$-morphisms  $\chi : M/\pt_n\to M^\prime/\pt_n$
and $\chi^\prime : M^{\prime}/\pt_n\to M^{\prime\prime}/\pt_n$ can be composed,
i.e.\ $\chi^\prime \circ \chi : M/\pt_n\to M^{\prime\prime}/\pt_n$ is a $\SMan/\pt_n$-morphism. 
Using the defining property \eqref{diag:relmorph}, the set of all $\SMan/\pt_n$-morphisms
 $\chi : M/\pt_n\to M^\prime/\pt_n$ can be easily characterized.
\begin{lem}\label{lem:SManptcharacterization}
Let $M$ and $M^\prime$ be any two objects in $\SMan$  and $\pt_n$ any object in $\SPt^\op$. Then 
the map 
\begin{flalign}
\nn \alpha_{\pt_n} : \Hom_{\SMan/\pt_n} (M/\pt_n,M^\prime/\pt_n) &\longrightarrow  \Hom_{\SMan}(\pt_n\times M,M^\prime)~,\\
\big(\chi : \pt_n\times M\to\pt_n\times M^\prime\big ) &\longmapsto  \big(\mathrm{pr}_{\pt_n\times M^\prime,M^\prime}\circ \chi: \pt_n\times M \to  M^\prime\big)~,
\end{flalign}
is a bijection of sets, where $\mathrm{pr}_{\pt_n\times M^\prime,M^\prime} : \pt_n\times M^\prime\to M^\prime$
denotes the projection $\SMan$-morphism on the second factor. In fact, its inverse is given by
\begin{flalign}
\nn \alpha_{\pt_n}^{-1} :  \Hom_{\SMan}(\pt_n\times M,M^\prime)&\longrightarrow \Hom_{\SMan/\pt_n} (M/\pt_n,M^\prime/\pt_n) ~,\\
 \big(\psi : \pt_n\times M \to  M^\prime\big)  &\longmapsto 
 \big((\id_{\pt_n},\psi) : \pt_n\times M\to\pt_n\times M^\prime\big )~.
\end{flalign}
\end{lem}
\sk

We next show that the assignment $\pt_n \mapsto \Hom_{\SMan/\pt_n}(M/\pt_n,M'/\pt_n)$ 
defines a functor $\SPt^{\op} \to \Set$, which is basically the functor 
used in \cite{SachseDiss,FH} to define super-mapping spaces. 
Given any two objects $M$ and $M^\prime$ in $\SMan$ and any
$\SPt^\op$-morphism $\lambda^\op : \pt_n\to \pt_m$ (i.e.\ a $\SPt$-morphism
$\lambda : \pt_m\to \pt_n$) we can define a map of sets
\begin{flalign}
\nn \Hom_{\SMan}(\pt_n\times M,M^\prime)&\longrightarrow \Hom_{\SMan}(\pt_m\times M,M^\prime)~,\\
\big(\psi : \pt_n\times M \to M^\prime\big) &\longmapsto \big(\psi\circ (\lambda\times\id_{M}) : \pt_m\times M\to M^\prime\big)~.
\end{flalign}
Using also Lemma \ref{lem:SManptcharacterization} we obtain a map of sets
\begin{flalign}
\nn \lambda^\op_\ast : \Hom_{\SMan/\pt_n} (M/\pt_n,M^\prime/\pt_n)  &\longrightarrow \Hom_{\SMan/\pt_m} (M/\pt_m,M^\prime/\pt_m) ~,\\
\chi &\longmapsto  \alpha_{\pt_m}^{-1}\left(\alpha_{\pt_n}(\chi)\circ (\lambda\times\id_{M})\right)~,\label{eqn:lambdaopast}
\end{flalign}
which describes how relative $\SMan$-morphisms behave under the exchange of superpoints.
The following  properties can be easily derived from (\ref{eqn:lambdaopast}). 
We therefore can omit the proof.
\begin{lem}\label{lem:lambdaopast}
\begin{itemize}
\item[(i)] For any identity $\SPt^\op$-morphism $\lambda^\op = \id_{\pt_n} : \pt_n\to\pt_n$
the map $\lambda^\op_\ast$ is the identity. For any two $\SPt^\op$-morphisms $\lambda^\op : \pt_n\to\pt_m$ and
 $\lambda^{\prime\op} : \pt_m\to\pt_l$ 
we have that $(\lambda^{\prime\op}\circ^\op \lambda^\op)_\ast =\lambda^{\prime\op}_\ast\circ \lambda^{\op}_\ast $.
\item[(ii)] $\lambda^\op_\ast$ preserves identities and compositions, i.e.\
\begin{flalign}
\lambda^\op_\ast(\id_{\pt_n\times M}) = \id_{\pt_m\times M}~,\qquad \lambda^\op_\ast(\chi^\prime\circ \chi) = \lambda^\op_\ast(\chi^\prime)\circ \lambda^\op_\ast(\chi) ~,
\end{flalign}
for all objects $M$ in $\SMan$ and 
all $\SMan/\pt_n$-morphisms $\chi : M/\pt_n\to M^\prime/\pt_n$ and
$\chi^\prime : M^\prime/\pt_n \to M^{\prime\prime}/\pt_n$.
\item[(iii)] $\lambda^\op_\ast$ preserves isomorphisms, i.e.\ 
$\chi : M/\pt_n\to M^\prime/\pt_n$ is a $\SMan/\pt_n$-isomorphism if and only if 
$\lambda^\op_\ast(\chi) : M/\pt_m\to M^\prime/\pt_m$ is a
$\SMan/\pt_m$-isomorphism.
\end{itemize}
\end{lem}
\sk

We shall also require a relative notion of differential geometry on $\pt_n$-relative 
supermanifolds $M/\pt_n$. Let $\Der_{\pt_n\times M}$ be the superderivation sheaf of the product supermanifold $\pt_n\times M$
and $U\subseteq \und{M}$ be any open subset. A superderivation 
$X\in \Der_{\pt_n\times M}(U)$ is called a {\em $\pt_n$-relative superderivation} 
provided that $X(\zeta\otimes \1) =0$, for all $\zeta \in \Lambda_n$.
The $\pt_n$-relative superderivations form a subsheaf $\Der_{M/\pt_n}$ of left
$\OO_{\pt_n\times M}$-supermodules of $\Der_{\pt_n\times M}$,
which is isomorphic to $\Lambda_n\otimes \Der_{M}$. 
The dual $\Omega^1_{M/\pt_n}:= \underline{\Hom}_{\OO_{\pt_n\times M}}(\Der_{M/\pt_n},\OO_{\pt_n\times M})$ 
of the $\pt_n$-relative superderivation sheaf $\Der_{M/\pt_n}$
is called the {\em $\pt_n$-relative super-one-form sheaf} and it is isomorphic
to $\Lambda_n\otimes \Omega^1_M$. The $\pt_n$-relative differential
$\dd_{M/\pt_n} : \OO_{\pt_n\times M}  \to \Omega^1_{M/\pt_n} $ is defined as in the non-relative case
 and it can be identified with $\id_{\Lambda_n}\otimes \dd : \Lambda_n\otimes \OO_{M} \to \Lambda_n\otimes \Omega^1_M$,
where $\dd :  \OO_{M}\to \Omega^1_M$ is the usual differential. Loosely speaking, we obtain 
the $\pt_n$-relative geometric objects on $M/\pt_n$ by $\Lambda_n$-superlinear extension of the ones on $M$. 
As a consequence, given any object $\bbM = (M,\Omega,E)$ in  $\SLoc$
and any object $\pt_n$ in $\SPt^\op$,
we can assign a $\pt_n$-relative supermanifold $M/\pt_n$ together with
$\pt_n$-relative super-one-forms $\1\otimes \Omega \in  \Lambda_n\otimes \Omega^1(M,\mathfrak{spin})
\simeq \Omega^1(M/\pt_n,\mathfrak{spin})$ and
$\1\otimes E \in  \Lambda_n\otimes \Omega^1(M,\mathfrak{t})
\simeq \Omega^1(M/\pt_n,\mathfrak{t})$.
\sk

With these preparations we can now define the $\SSet$-category
$\eSLoc$.
\begin{defi}\label{defi:eSLoc}
The $\SSet$-category $\eSLoc$ is given by the following data:
\begin{itemize}
\item The objects are all objects $\bbM = (M,\Omega,E)$ in $\SLoc$.

\item For any two objects $\bbM$ and $\bbM^\prime$ in $\eSLoc$, 
the object of morphisms from $\bbM$ to $\bbM^\prime$ is given by the following functor
$\eSLoc(\bbM,\bbM^\prime) : \SPt^\op \to \Set$:
For any object $\pt_n$ in $\SPt^\op$  we define $\eSLoc(\bbM,\bbM^\prime) (\pt_n)$ to be the
set of all $\SMan/\pt_n$-morphisms $\chi : M/\pt_n\to M^\prime/\pt_n$, such that
\begin{enumerate}
\item $\und{\chi}: \und{M}\to \und{M^\prime} $ is an open embedding with 
causally compatible image in the reduced oriented and time-oriented Lorentz manifold $\und{\bbM^\prime}$,
\item $\chi :  M/\pt_n\to M^\prime\vert_{\und{\chi}(\und{M})}/\pt_n$ is a $\SMan/\pt_n$-isomorphism,
\item the $\pt_n$-relative super-Cartan structures are preserved, i.e.\ $\chi^\ast(\1\otimes \Omega^\prime) = \1\otimes \Omega$
and $\chi^\ast(\1\otimes E^\prime) = 1\otimes E$.
\end{enumerate}
For any $\SPt^\op$- morphism $\lambda^\op : \pt_n\to \pt_m$ 
we define the map of sets 
\begin{flalign}
\eSLoc(\bbM,\bbM^\prime) (\lambda^\op)  :=\lambda^\op_\ast : 
\eSLoc(\bbM,\bbM^\prime) (\pt_n)\longrightarrow \eSLoc(\bbM,\bbM^\prime) (\pt_m)~,
\end{flalign}
where $ \lambda^\op_\ast$ is given in (\ref{eqn:lambdaopast}).

\item For any three objects $\bbM$, $\bbM^\prime$ and $\bbM^{\prime\prime}$ 
in $\eSLoc$, we define the composition morphism $\bullet : \eSLoc(\bbM^\prime,\bbM^{\prime\prime})
\times \eSLoc(\bbM,\bbM^{\prime})\to \eSLoc(\bbM,\bbM^{\prime\prime})$ to be the natural transformation
with components 
\begin{flalign}
\nn\bullet_{\pt_n} :=\circ : \eSLoc(\bbM^\prime,\bbM^{\prime\prime})(\pt_n) \times \eSLoc(\bbM,\bbM^{\prime})(\pt_n) 
&\longrightarrow \eSLoc(\bbM,\bbM^{\prime\prime})(\pt_n)~,
\end{flalign}
where $\circ$ is the composition of $\SMan/\pt_n$-morphisms.

\item For any object $\bbM$ in $\eSLoc$, we define the identity on $\bbM$ morphism
$\oone : \mathfrak{I} \to \eSLoc(\bbM,\bbM)$ to be the natural transformation with components
\begin{flalign}
\oone_{\pt_n} : \pt = \mathfrak{I}(\pt_n)  \longrightarrow \eSLoc(\bbM,\bbM)(\pt_n)~,~~\star\longmapsto  \id_{\pt_n\times M}~,
\end{flalign}
where $\id_{\pt_n\times M}$ is the identity  $\SMan/\pt_n$-morphism.
\end{itemize}
\end{defi}

\begin{rem}
Using Lemma \ref{lem:lambdaopast} one can easily see that $\eSLoc$ 
is a $\SSet$-category according to Definition \ref{defi:enrichedcategory}:
Lemma \ref{lem:lambdaopast} (iii) implies that the map of sets $\eSLoc(\bbM,\bbM^\prime)(\lambda^\op)$
is well-defined, i.e.\ that it has the claimed codomain, and Lemma \ref{lem:lambdaopast} (i) implies that 
$\eSLoc(\bbM,\bbM^\prime) : \SPt^\op \to \Set$ is a functor. The composition $\bullet$ and identity
$\oone$ are natural transformations because of Lemma \ref{lem:lambdaopast} (ii). Finally,
the diagrams in Definition \ref{defi:enrichedcategory} commute because of the associativity and identity property
of the composition $\circ$ of $\SMan/\pt_n$-morphisms.
\end{rem}

\subsection{\label{ssec:eOO}The \texorpdfstring{$\SSet$}{SSet}-functor \texorpdfstring{$\eOO : \eSLoc^\op\to\eSVec$}{eOO : eSLocop -> eSVec}}
Our next goal is to show that the ordinary global section functor $\OO : \SLoc^\op \to\SVec$
can be promoted to a $\SSet$-functor $\eOO : \eSLoc^\op\to \eSVec$ with values in the $\SSet$-category
$\eSVec$ of super-vector spaces. For defining the latter $\SSet$-category we will first 
discuss how to promote super-vector spaces to (left) supermodules over Grassmann algebras. 
This procedure is  known as ``extension of ring of scalars'' and discussed in detail in \cite[Chapter II.5]{BourbakiA1}.
\sk

Given any object $V$ in $\SVec$ and any object $\pt_n$ in $\SPt^\op$,
we can consider the left $\Lambda_n$-supermodule $\Lambda_n\otimes V$.
Given any two objects $V$ and $V^\prime$ in $\SVec$, the set of $\Lambda_n\text{-}\Mod$-morphisms 
$L : \Lambda_n\otimes V\to \Lambda_n\otimes V^\prime$ can be easily characterized.
\begin{lem}\label{lem:SVecptcharacterization}
Let $V$ and $V^\prime$ be any two objects in $\SVec$  and $\pt_n$ any object in $\SPt^\op$. 
Then the map 
\begin{flalign}
\nn \beta_{\pt_n} : \Hom_{\Lambda_n\text{-}\Mod} (\Lambda_n\otimes V,\Lambda_n\otimes V^\prime) &\longrightarrow  \Hom_{\SVec}(V,\Lambda_n\otimes V^\prime)~,\\
\big( L : \Lambda_n\otimes V\to \Lambda_n\otimes V^\prime \big) 
&\longmapsto  \big( L \circ (\eta_n\otimes \id_{V}) : V\simeq \bbR\otimes V\to \Lambda_n\otimes V^\prime \big)~,
\end{flalign}
is a bijection of sets, where $\eta_n : \bbR\to \Lambda_n$ denotes the unit in $\Lambda_n$.
\end{lem}
\begin{proof}
The map $\beta_{\pt_n}$ is invertible by assigning to any $\SVec$-morphism $K : V\to \Lambda_n\otimes V^\prime$
the $\Lambda_n\text{-}\Mod$-morphism $\beta_{\pt_n}^{-1} (K) : \Lambda_n\otimes V\to \Lambda_n\otimes V^\prime$
that is specified by $\beta_{\pt_n}^{-1} (K)(\zeta\otimes v) := (\zeta\otimes \1)\,K(v)$, for all $\zeta\otimes v\in\Lambda_n\otimes V$,
and $\bbR$-linear extension.
\end{proof}
\sk

Given any two objects $V$ and $V^\prime$ in $\SVec$ and any $\SPt^\op$-morphism $\lambda^\op : \pt_n\to\pt_m$ 
(i.e.\ a $\SPt$-morphism $\lambda : \pt_m\to\pt_n$)
we can define a map of sets
\begin{flalign}
\Hom_{\SVec}(V,\Lambda_n\otimes V^\prime) \longrightarrow \Hom_{\SVec}(V,\Lambda_m\otimes V^\prime)~,~~ K\longmapsto (\lambda^\ast\otimes \id_{V^\prime})\circ K~.
\end{flalign}
Using also Lemma \ref{lem:SVecptcharacterization} we obtain a map of sets
\begin{flalign}
\nn \lambda^\op_\ast : \Hom_{\Lambda_n\text{-}\Mod} (\Lambda_n\otimes V,\Lambda_n\otimes V^\prime) &\longrightarrow
\Hom_{\Lambda_m\text{-}\Mod} (\Lambda_m\otimes V,\Lambda_m\otimes V^\prime) ~,\\
L&\longmapsto \beta_{\pt_m}^{-1}\left((\lambda^\ast\otimes \id_{V^\prime}) \circ \beta_{\pt_n}(L)\right)~.
\label{eqn:lambdaopastSVec}
\end{flalign}
The following properties can be easily derived from (\ref{eqn:lambdaopastSVec}), see also \cite{BourbakiA1}. 
We therefore can omit the proof.
\begin{lem}\label{lem:lambdaopastSVec}
\begin{itemize}
\item[(i)] For any identity $\SPt^\op$-morphism $\lambda^\op = \id_{\pt_n} : \pt_n\to\pt_n$
the map $\lambda^\op_\ast$ is the identity. For any two $\SPt^\op$-morphisms $\lambda^\op : \pt_n\to\pt_m$ and
 $\lambda^{\prime\op} : \pt_m\to\pt_l$ 
we have that $(\lambda^{\prime\op}\circ^\op \lambda^\op)_\ast =\lambda^{\prime\op}_\ast\circ \lambda^{\op}_\ast $.
\item[(ii)] $\lambda^\op_\ast$ preserves identities and compositions, i.e.\
\begin{flalign}
\lambda^\op_\ast(\id_{\Lambda_n\otimes V}) = \id_{\Lambda_m\otimes V}~,\qquad \lambda^\op_\ast(L^\prime\circ L) = \lambda^\op_\ast(L^\prime)\circ \lambda^\op_\ast(L) ~,
\end{flalign}
for all objects $V$ in $\SVec$ and 
all $\Lambda_n\text{-}\Mod$-morphisms $L : \Lambda_n\otimes V\to\Lambda_n\otimes V^\prime$ 
and $L^\prime : \Lambda_n\otimes V^{\prime} \to\Lambda_{n}\otimes V^{\prime\prime}$.
\item[(iii)] $\lambda^\op_\ast$ preserves isomorphisms, i.e.\ 
$L : \Lambda_n\otimes V\to\Lambda_n\otimes V^\prime$ is a $\Lambda_n\text{-}\Mod$-isomorphism if and only if 
$\lambda^\op_\ast(L) : \Lambda_m\otimes V\to\Lambda_m\otimes V^\prime$ is a
$\Lambda_m\text{-}\Mod$-isomorphism.
\end{itemize}
\end{lem}

With these preparations we can now define the $\SSet$-category
$\eSVec$.
\begin{defi}\label{defi:eSVec}
The $\SSet$-category $\eSVec$ is given by the following data:
\begin{itemize}
\item The objects are  all objects $V$ in $\SVec$.

\item For any two objects $V$ and $V^\prime$ in $\eSVec$,
the object of morphisms from $V$ to $V^\prime$ is given by the following
functor $\eSVec(V,V^\prime) : \SPt^\op \to\Set$: For any object 
$\pt_n$ in $\SPt^\op$ we define $\eSVec(V,V^\prime) (\pt_n)$ to be the set
of all $\Lambda_n\text{-}\Mod$-morphisms $L : \Lambda_n\otimes V\to \Lambda_n\otimes V^\prime$.
For any $\SPt^\op$-morphism $\lambda^\op : \pt_n\to\pt_m$ we define the map of sets
\begin{flalign}
\eSVec(V,V^\prime)(\lambda^\op) :=\lambda^\op_\ast :  \eSVec(V,V^\prime) (\pt_n) \longrightarrow  \eSVec(V,V^\prime) (\pt_m)~,
\end{flalign}
where $\lambda^\op_\ast$  is given in (\ref{eqn:lambdaopastSVec}).

 \item For any three objects $V$, $V^{\prime}$ and $V^{\prime\prime}$ 
 in $\eSVec$, we define the composition morphism $\bullet : \eSVec(V^\prime,V^{\prime\prime})
 \times \eSVec(V,V^\prime) \to \eSVec(V,V^{\prime\prime})$ 
 to be the natural transformation with components 
 \begin{flalign}
\nn \bullet_{\pt_n} :=\circ : \eSVec(V^\prime,V^{\prime\prime})(\pt_n)
 \times \eSVec(V,V^\prime)(\pt_n) &\longrightarrow \eSVec(V,V^{\prime\prime})(\pt_n)~,
 \end{flalign}
where $\circ $ is the composition of $\Lambda_n\text{-}\Mod$-morphisms.

\item For any object $V$ in $\eSVec$, we define the identity on $V$ morphism 
$\oone : \mathfrak{I} \to \eSVec(V,V)$ to be the natural transformation with components
\begin{flalign}
\oone_{\pt_n} : \pt = \mathfrak{I}(\pt_n) \longrightarrow \eSVec(V,V)(\pt_n)~,~~\star \longmapsto \id_{\Lambda_n\otimes V}~,
\end{flalign}
where $\id_{\Lambda_n\otimes V}$ is the identity $\Lambda_n\text{-}\Mod$-morphism.
\end{itemize}
\end{defi}
\sk

\begin{rem}
Let us briefly compare our category $\eSVec$ with the functorial formulation of super(linear) algebra 
used in \cite[Section 3]{Sachse}, \cite[Section 2]{Molotkov} and \cite{Schwarz}. In the latter approach, 
one considers module objects in the monoidal category $\SSet$, i.e.\ functors $M : \SPt^\op \to \Set$ such that all $M(\pt_n)$ 
carry module structures and all $M(\lambda^\op) : M(\pt_n) \to M(\pt_m)$ are homomorphisms of modules. 
Together with appropriate morphisms in $\SSet$, they form the category $\mathrm{Mod}(\SSet)$. It is easy 
to see \cite[Section 3.4]{Sachse} that any object $V$  in $\SVec$ gives rise to an 
object $\overline{V}$ in $\mathsf{Mod}(\SSet)$ (i.e.\ a functor) defined by 
\begin{flalign}
\overline{V}(\pt_n) := (V \otimes \Lambda_n)_{0} ~~,\qquad
\overline{V}(\lambda^\op : \pt_n \to \pt_m) := (\id_V \otimes \lambda^\ast)\vert_{0}~.
\end{flalign}
Moreover, it can be shown \cite[Corollary 3.4.2]{Sachse} that the assignment 
$\overline{\,\cdot\, }: \SVec \to \mathsf{Mod}(\SSet)$ defines a fully faithful functor; 
its image $\overline{\SVec}$ consists of \emph{representable} modules. Comparing with 
Definition \ref{defi:eSVec}, we see that $\eSVec$ and $\overline{\SVec}$ have isomorphic classes 
of objects but our enriched category $\eSVec$ contains more morphisms. In fact, since $\overline{\, \cdot \, }$ is full, 
we have $\Hom_{\overline{\SVec}}(\overline{V},\overline{W}) \cong \Hom_{\SVec}(V,W) \cong \eSVec(V,W)(\pt_0)$. Thus, all 
additional information contained in $\eSVec(V,W)(\pt_n)$ for $n > 0$ is not seen in $\overline{\SVec}$ and our 
definition provides a proper enrichment of the latter category. It may be possible to
give a natural meaning to our enrichment constructions inside the functor categories
$\mathsf{Mod}(\SSet)$ (or $\mathsf{SMod}(\SSet)$), but this discussion is beyond the scope of the present publication.      
\end{rem}
\sk

We now can define a $\SSet$-functor $\eOO : \eSLoc^\op \to \eSVec$ as follows: To any
object $\bbM$ in $\eSLoc^\op$ we assign the object $\eOO(\bbM) := \OO(\bbM)$
in $\eSVec$. To any two objects $\bbM$ and $\bbM^\prime$ in $\eSLoc^\op$
we assign the $\SSet$-morphism $\eOO_{\bbM,\bbM^\prime} : \eSLoc^\op(\bbM,\bbM^\prime) \to
\eSVec(\OO(\bbM),\OO(\bbM^\prime))$ given by the natural transformation (of functors
from $\SPt^\op$ to $\Set$) with components
\begin{flalign}
(\eOO_{\bbM,\bbM^\prime})_{\pt_n} : \big(\chi : \pt_n\times M^\prime\to \pt_n \times M\big) \longmapsto \big(\chi^\ast : \Lambda_n\otimes \OO(\bbM)\to \Lambda_n\otimes \OO(\bbM^\prime)\big)~.
\end{flalign}
Naturality of these components is easily checked. We obtain
\begin{propo}
The assignment $\eOO : \eSLoc^\op\to \eSVec$ given above is a $\SSet$-functor.
\end{propo}
\begin{proof}
We have to prove that this assignment is compatible with the composition and identity, see Definition
\ref{defi:enrichedfunctor}. Let $\pt_n$ be any object in $\SPt^\op$ and $\bbM$, $\bbM^\prime$ and $\bbM^{\prime\prime}$ 
any three objects in $\eSLoc^\op$. We obtain that, for all $\chi \in \eSLoc^\op(\bbM,\bbM^\prime)(\pt_n)$
 and $\chi^\prime\in\eSLoc^\op(\bbM^\prime,\bbM^{\prime\prime})(\pt_n)$,
\begin{flalign}
\nn (\eOO_{\bbM,\bbM^{\prime\prime}})_{\pt_n}\big( \chi^{\prime}\bullet_{\pt_n}^\op \chi\big)
&=(\eOO_{\bbM,\bbM^{\prime\prime}})_{\pt_n}\big( \chi\circ \chi^\prime\big) =
(\chi\circ \chi^\prime)^\ast = \chi^{\prime\ast}\circ\chi^{\ast}\\
\nn & =  (\eOO_{\bbM^\prime,\bbM^{\prime\prime}})_{\pt_n}(\chi^\prime)\circ 
(\eOO_{\bbM,\bbM^{\prime}})_{\pt_n}(\chi)\\
&= (\eOO_{\bbM^\prime,\bbM^{\prime\prime}})_{\pt_n}(\chi^\prime)\bullet_{\pt_n}  (\eOO_{\bbM,\bbM^{\prime}})_{\pt_n}(\chi)~,
\end{flalign}
which proves compatibility with the composition. Compatibility with the identity is shown by
\begin{flalign}
\nn (\eOO_{\bbM,\bbM})_{\pt_n}(\oone_{\pt_n}(\star)) &=(\eOO_{\bbM,\bbM})_{\pt_n}(\id_{\pt_n\times M})
=(\id_{\pt_n\times M})^\ast \\
&= \id_{\Lambda_n\otimes \OO(\bbM)}= \oone_{\pt_n}(\star)~,
\end{flalign}
for all objects $\bbM$ in $\eSLoc^\op$.
\end{proof}

\begin{rem} \label{rem:furtherExt}
As in the ordinary case, we can enlarge the morphisms in the $\SSet$-category
$\eSVec$ by replacing in Definition \ref{defi:eSVec} all appearances of 
the sets $\Hom_{\Lambda_n\text{-}\Mod}(\Lambda_n\otimes V,\Lambda_n\otimes V^\prime)$ 
of $\Lambda_n\text{-}\Mod$-morphisms by the sets underlying 
the internal hom-objects $\underline{\Hom}_{\Lambda_n}(\Lambda_n\otimes V,\Lambda_n\otimes V^\prime)$ in the 
category $\Lambda_n\text{-}\Mod$, i.e.\ all right $\Lambda_n$-linear maps 
$L: \Lambda_n\otimes V \to \Lambda_n\otimes V^\prime$. We denote the resulting
$\SSet$-category by $\underline{\eSVec}$ and note that there
is an obvious $\SSet$-functor $\eSVec \to \underline{\eSVec}$.
Consequently, we can regard the $\SSet$-functor $\eOO : \eSLoc^\op \to \eSVec$ 
also as a $\SSet$-functor (denoted with a slight abuse of notation by the same symbol)
$\eOO : \eSLoc^\op \to \underline{\eSVec}$. As explained in Footnote \ref{footnote:underline},
this generalization will be needed in Section \ref{sec:examples} to describe 
$1\vert 1$-dimensional examples (i.e.\ superparticles),  which are somewhat peculiar. 
\end{rem}

\subsection{Structure of the \texorpdfstring{$\SSet$}{SSet}-natural transformations \texorpdfstring{$\eOO\Rightarrow \eOO$}{eO => eO}}
In super-field theories, cf.\ Definition \ref{defi:sft}, we have described
the dynamics by a suitable natural transformation $P : \OO\Rightarrow \OO$ of functors
from $\SLoc^\op$ to $\underline{\SVec}$. In the enriched setting, we shall use
suitable $\SSet$-natural transformations $P : \eOO \Rightarrow \eOO $
of $\SSet$-functors from $\eSLoc^\op$ to $\underline{\eSVec}$.
Recalling Definition \ref{defi:enrichednatural}, such enriched natural transformations
are given by assigning to every object $\bbM$ in $\eSLoc^\op$
a $\SSet$-morphism $P_{\bbM} : \mathfrak{I} \to \underline{\eSVec}(\OO(\bbM),\OO(\bbM))$,
such that the diagram given in this definition commutes.
The $\SSet$-morphisms $P_{\bbM} : \mathfrak{I} \to \underline{\eSVec}(\OO(\bbM),\OO(\bbM))$
are given by a natural transformations of functors from $\SPt^\op$ to $\Set$, whose 
components are maps of sets that we denote by
\begin{flalign}
\nn  (P_{\bbM})_{\pt_n} : \pt = \mathfrak{I}(\pt_n)  &\longrightarrow \underline{\eSVec}(\OO(\bbM),\OO(\bbM))(\pt_n)~,\\
\star &\longmapsto (P_{\bbM})_{\pt_n}(\star)=:P_{\bbM/\pt_n}~.
\end{flalign}
Consequently, any $\SSet$-morphism $P_{\bbM} : \mathfrak{I} \to \underline{\eSVec}(\OO(\bbM),\OO(\bbM))$
is specified by a collection of right $\Lambda_n$-linear maps $P_{\bbM/\pt_n} : \Lambda_n\otimes 
\OO(\bbM)\to \Lambda_n\otimes \OO(\bbM) $, for all objects $\pt_n$ in $\SPt^\op$,
such that
 \begin{flalign}\label{eqn:enrichednaturaldiffop}
\lambda^\op_\ast\big(P_{\bbM/\pt_n} \big)  = P_{\bbM/\pt_m}~,
\end{flalign}
for all $\SPt^\op$-morphisms $\lambda^\op : \pt_n\to\pt_m$.
Notice that, for all objects $\pt_m$ in $\SPt$,  there is the terminal $\SPt$-morphism $\lambda : \pt_m\to\pt$
given by $\und{\lambda} = \id_{\pt}$ and the unit $\lambda^\ast= \eta_m : \bbR\to\Lambda_m$
in $\Lambda_m$.
As a consequence, the equation (\ref{eqn:enrichednaturaldiffop}) applied to the 
terminal $\SPt^\op$-morphisms  $\lambda^\op : \pt\to\pt_m$
implies that any $P_{\bbM/\pt_m}$ can be expressed in terms of the single right
 $\Lambda_0 {=} \bbR$-linear map $P_{\bbM/\pt} : \OO(\bbM)\to\OO(\bbM)$ via
\begin{flalign}
\nn P_{\bbM/\pt_m} = \id_{\Lambda_m}\otimes P_{\bbM/\pt} : \Lambda_m\otimes 
\OO(\bbM) &\longrightarrow \Lambda_m\otimes \OO(\bbM)~,\\
\zeta\otimes F &\longmapsto (-1)^{\vert \zeta\vert \,\vert  P_{\bbM/\pt}\vert }~\zeta\otimes  P_{\bbM/\pt} (F)~.\label{eqn:enrichednaturaldiffopsimple}
\end{flalign}
The commutative diagram in Definition \ref{defi:enrichednatural} requires that
\begin{flalign}
P_{\bbM^\prime/\pt_n} \bullet_{\pt_n}  (\eOO_{\bbM,\bbM^\prime})_{\pt_n}(\chi) = 
(\eOO_{\bbM,\bbM^\prime})_{\pt_n}(\chi)\bullet_{\pt_n} P_{\bbM/\pt_n}~,
\end{flalign}
for all objects $\bbM$ and $\bbM^\prime$ in $\eSLoc^\op$, all objects $\pt_n$ in $\SPt^\op$
and all $\chi\in\eSLoc^\op(\bbM,\bbM^\prime)(\pt_n)$.
Using (\ref{eqn:enrichednaturaldiffopsimple}), these conditions are equivalent to
the commutative diagrams
\begin{flalign}
\xymatrix{
\ar[d]_-{\chi^\ast} \Lambda_n\otimes \OO(\bbM^\prime) \ar[rrr]^-{\id_{\Lambda_n}\otimes P_{\bbM^\prime/\pt}} &&& \Lambda_n\otimes\OO(\bbM^\prime)\ar[d]^-{\chi^\ast}\\
\Lambda_n\otimes \OO(\bbM) \ar[rrr]_-{\id_{\Lambda_n}\otimes P_{\bbM/\pt}} &&& \Lambda_n\otimes \OO(\bbM)
}
\end{flalign}
 of right $\Lambda_n$-linear maps, for all objects $\bbM$ and $\bbM^\prime$ in $\eSLoc^\op$, all objects $\pt_n$ in $\SPt^\op$
and all $\chi\in\eSLoc^\op(\bbM^\prime,\bbM)(\pt_n)$. (For later convenience we have exchanged here $\bbM$ and $\bbM^\prime$.)
We therefore have shown 
\begin{propo}\label{propo:ordinaryimpliesenriched}
There is a bijective correspondence between 
\begin{enumerate}
\item $\SSet$-natural transformations 
$P : \eOO\Rightarrow \eOO$ of $\SSet$-functors
from $\eSLoc^\op$ to $\underline{\eSVec}$,
\item ordinary natural transformations $P : \OO\Rightarrow \OO$ of functors from $\SLoc^\op$ to
$\underline{\SVec}$, such that the diagram
\begin{flalign}\label{eqn:enricheddiffopsimplediagrams}
\xymatrix{
\ar[d]_-{\chi^\ast} \Lambda_n\otimes \OO(\bbM^\prime) \ar[rrrr]^-{P_{\bbM^\prime/\pt_n} \,:=\, \id_{\Lambda_n}\otimes P_{\bbM^\prime}} &&&& \Lambda_n\otimes \OO(\bbM^\prime)\ar[d]^-{\chi^\ast}\\
\Lambda_n\otimes \OO(\bbM) \ar[rrrr]_-{P_{\bbM/\pt_n} \,:=\,\id_{\Lambda_n}\otimes P_{\bbM}} &&&& \Lambda_n\otimes \OO(\bbM)
}
\end{flalign}
of right $\Lambda_n$-linear maps commutes, for all objects 
$\bbM$ and $\bbM^\prime$ in $\eSLoc^\op$, all objects $\pt_n$ in $\SPt^\op$
and all $\chi\in\eSLoc^\op(\bbM^\prime,\bbM)(\pt_n)$.
\end{enumerate}
\end{propo}

\subsection{The definition}
With this preparation we can now give a simple definition of enriched super-field theories.
In particular, using Proposition \ref{propo:ordinaryimpliesenriched} we can define an enriched
super-field theory to be a super-field theory (cf.\ Definition \ref{defi:sft})
together with extra conditions which ensure that the ordinary natural transformation $P : \OO\Rightarrow
\OO$ extends to a $\SSet$-natural transformation $P : \eOO\Rightarrow \eOO$.  
\begin{defi}\label{defi:sftenriched}
An enriched super-field theory is a super-field theory according to
Definition \ref{defi:sft}, such that the diagram (\ref{eqn:enricheddiffopsimplediagrams}) 
of right $\Lambda_n$-linear maps commutes, for all objects $\bbM$ and 
$\bbM^\prime$ in $\eSLoc^\op$, all objects $\pt_n$ in $\SPt^\op$
and all $\chi\in\eSLoc^\op(\bbM^\prime,\bbM)(\pt_n)$.
\end{defi}

%%%%%%%%%%%%%%%%%%%%%%%%%%%%%%%%%%%%%%%%%%%%%%%%
%%%%%%%%%%%%%%%%%%%%%%%%%%%%%%%%%%%%%%%%%%%%%%%%

\section{\label{sec:enriched}Construction of enriched super-quantum field theories}
We show that given any enriched super-field theory according to Definition \ref{defi:sftenriched}
one can construct a $\SSet$-functor $\eAA : \eSLoc\to\eSAlg$, i.e.\ an enriched super-QFT.
As in Section \ref{sec:ordinary} we decompose our construction into two steps:
First, we construct a $\SSet$-functor $\eLL : \eSLoc\to\eXX$ which describes the enriched classical theory.
Second, we construct a $\SSet$-functor $\eQQ : \eXX\to\eSAlg$ describing the enriched quantization.
We shall also study properties of the enriched functors and establish a connection
between the enriched super-field theory and the enriched super-QFT by constructing
an enriched locally covariant quantum field. In contrast to ordinary super-QFTs, our 
enriched approach captures also supersymmetry transformations. It should be emphasized 
that the $\SSet$-categories $\eSLoc$, $\eXX$ and $\eSAlg$ 
are defined using $\bbZ_2$-parity preserving morphisms. 
The appearance of supersymmetry transformations in the enriched setting is
due to the higher superpoints of the morphism supersets in  $\eSLoc$, $\eXX$ and $\eSAlg$.

\subsection{The \texorpdfstring{$\SSet$}{SSet}-functor \texorpdfstring{$\eOO_{\cc} : \eSLoc \to \eSVec$}{eOc : eSLoc -> eSVec}}
As a preparatory step we construct a $\SSet$-functor $\eOO_{\cc} : \eSLoc \to \eSVec$
that assigns to objects $\bbM$ in $\eSLoc$ the super-vector spaces of compactly supported sections 
$\eOO_{\cc}(\bbM) :=\OO_{\cc}(\bbM)$. 
We assign to any two objects $\bbM$ and $\bbM^\prime$ in $\eSLoc$ 
the $\SSet$-morphism ${\eOO_{\cc}}_{\bbM,\bbM^\prime} : \eSLoc(\bbM,\bbM^\prime) 
\to \eSVec(\OO_{\cc}(\bbM),\OO_{\cc}(\bbM^\prime))$ given by the natural transformation
(of functors from $\SPt^\op$ to $\Set$) with components
\begin{flalign}
({\eOO_{\cc}}_{\bbM,\bbM^\prime})_{\pt_n} : \big(\chi : \pt_n\times M\to\pt_n\times M^\prime\big)\longmapsto
\big(\chi_\ast : \Lambda_n \otimes \OO_{\cc}(\bbM)\to \Lambda_n\otimes \OO_{\cc}(\bbM^\prime)\big)~,
\end{flalign}
where $\chi_\ast$ denotes the push-forward of compactly supported sections of the structure sheaf $\OO_{\pt_n\times M}$
 (cf.\ (\ref{eqn:pushdefi})), which exists as a consequence of the conditions 1.)\ and 2.)\ in Definition \ref{defi:eSLoc}.
Because of functoriality of the push-forwards, i.e.\ $(\id_{\pt_n\times M})_\ast = \id_{\Lambda_n\otimes \OO_{\cc}(\bbM)}$
and $(\chi^\prime \circ\chi)_\ast = \chi^\prime_\ast\circ\chi_\ast$ as  in Lemma \ref{lem:pushforward} (iii),
we have shown
\begin{propo}
The assignment $\eOO_{\cc} : \eSLoc \to \eSVec$ given above is a $\SSet$-functor.
\end{propo}
\sk

Let $\bbM = (M,\Omega,E)$ be any object in $\eSLoc$ and $\pt_n$ any object in $\SPt^\op$.
Making use of the $\pt_n$-relative differential geometry on $M/\pt_n$,
together with the $\pt_n$-relative supervielbein $\1\otimes E\in \Omega^{1}(M/\pt_n,\mathfrak{t})$, 
we can define a $\pt_n$-relative version of the pairing
(\ref{eqn:pairing}) by
\begin{flalign}
\nn \ip{\,\cdot\,}{\,\cdot\, }_{\bbM/\pt_n} : \big(\Lambda_n\otimes\OO_{\cc}(\bbM)\big)\otimes_{\Lambda_n}
\big(\Lambda_n\otimes\OO_{\cc}(\bbM)\big) &\longrightarrow \Lambda_n~,\\
H_1\otimes_{\Lambda_n} H_2 &\longmapsto \int_{M/\pt_n} \Ber(\1\otimes E)~H_1\,H_2~,\label{eqn:relativepairing}
\end{flalign}
where $\int_{M/\pt_n}$ is the $\pt_n$-relative Berezin integral, see e.g.\ \cite[\S 3.10]{DMSuper}.
Explicitly, the $\pt_n$-relative pairing reads as
\begin{flalign}
\nn \ip{\zeta_1\otimes F_1}{\zeta_2\otimes F_2}_{\bbM/\pt_n} &=
(-1)^{(\vert \zeta_1\vert + \vert\zeta_2\vert) \,\vert\Ber(E)\vert}\,(-1)^{\vert\zeta_2\vert\,\vert F_1\vert}\,\zeta_1\,\zeta_2\,\int_{M} \Ber(E)~F_1\,F_2\\
 &= (-1)^{(\vert \zeta_1\vert + \vert\zeta_2\vert) \,\vert P_{\bbM}\vert}\,(-1)^{\vert\zeta_2\vert\,\vert F_1\vert}\,\zeta_1\,\zeta_2\,\ip{F_1}{F_2}_{\bbM}~,\label{eqn:relativepairing2}
\end{flalign}
for all homogeneous $\zeta_1,\zeta_2\in\Lambda_n$ and $F_1,F_2\in\OO_{\cc}(\bbM)$, 
i.e.\ it is given by $\Lambda_n$-superlinear extension of the pairing on $\OO_c(\bbM)$. 
Here, we have moreover 
used that  by assumption $\vert \Ber(E)\vert = \vert P_{\bbM}\vert$.
Notice that (\ref{eqn:relativepairing}) can be extended to all $H_1,H_2\in\Lambda_n\otimes \OO(\bbM)$ with compactly 
overlapping support. 
\begin{lem}\label{lem:relativepushforward}
Let $\bbM$ and $\bbM^\prime$ be any two objects in $\eSLoc$ and $\pt_n$ 
any object in $\SPt^\op$. Then for any $\chi\in\eSLoc(\bbM,\bbM^\prime)(\pt_n)$ the following properties hold true:
\begin{itemize}
\item[(i)] $\chi^\ast\circ \chi_\ast = \id_{\Lambda_n\otimes \OO_{\cc}(\bbM)}$ and 
$\chi_\ast\circ \chi^\ast(H) = H$, for all $H\in \Lambda_n\otimes \OO_{\cc}(\bbM^\prime)$
 such that $\supp(H) \subseteq\und{\chi}(\und{M})$.
 \item[(ii)] $ \ip{H_1}{\chi_\ast(H_2)}_{\bbM^\prime/\pt_n}=\ip{\chi^\ast(H_1)}{H_2}_{\bbM/\pt_n} $,
 for all $H_1\in \Lambda_n\otimes \OO(\bbM^\prime)$ and $H_2\in\Lambda_n\otimes \OO_{\cc}(\bbM)$.
\end{itemize}
\end{lem}
\begin{proof}
The proof of item (i) is as in Lemma \ref{lem:pushforward} (i).
In the proof of item (ii) one follows the same steps as in Lemma \ref{lem:pushforward} (ii),
but uses instead of the usual transformation formula (\ref{eqn:trafoformel}) its relative version for the relative Berezin integral, see
\cite[\S 3.10]{DMSuper}.
\end{proof}

\subsection{Enriched properties of the super-Green's operators}
Let us fix any enriched super-field theory according to Definition
\ref{defi:sftenriched}. As a consequence of $P_{\bbM}$ being formally super-self adjoint (cf.\ (\ref{eqn:superself}))
one easily checks by using (\ref{eqn:relativepairing2}) that the $\pt_n$-relative super-differential operator
$P_{\bbM/\pt_n} := \id_{\Lambda_n}\otimes P_{\bbM} : \Lambda_n\otimes \OO(\bbM)\to \Lambda_n\otimes \OO(\bbM)$
satisfies
\begin{flalign}
\ip{H_1}{P_{\bbM/\pt_n}(H_2)}_{\bbM/\pt_n} = (-1)^{\vert H_1\vert \,\vert P_{\bbM}\vert}\,
\ip{ P_{\bbM/\pt_n}(H_1)}{H_2}_{\bbM/\pt_n}~,
\end{flalign}
for any object $\pt_n$ in $\SPt^\op$ and all homogeneous $H_1,H_2\in\Lambda_n\otimes \OO(\bbM)$
with compactly overlapping support. The right $\Lambda_n$-linear maps
$G_{\bbM/\pt_n}^\pm := 
\id_{\Lambda_n}\otimes G_{\bbM}^\pm : \Lambda_n\otimes \OO_{\cc}(\bbM) \to \Lambda_n\otimes \OO(\bbM)$
are easily seen to be retarded/advanced super-Green's operators for $P_{\bbM/\pt_n}$.
We also define  $G_{\bbM/\pt_n} := G_{\bbM/\pt_n}^+ - G_{\bbM/\pt_n}^-  : \Lambda_n\otimes \OO_{\cc}(\bbM)
 \to \Lambda_n\otimes \OO_{\sc}(\bbM)$.
Our statements in Subsection \ref{subsec:SGOP} easily generalize to the $\pt_n$-relative setting 
since all the proofs are algebraic and only use the properties of super-Green's operators.
We summarize without repeating the proofs the main properties which are used in the next subsections.
\begin{lem}\label{lem:enrichedgreen}
\begin{itemize}
\item[(i)] Let $\bbM$ be any object in $\eSLoc$ and $\pt_n$ any object in $\SPt^\op$. Then
\begin{flalign}
\ip{H_1}{G_{\bbM/\pt_n}^\pm(H_2)}_{\bbM/\pt_n} = (-1)^{(\vert H_1\vert +\vert P_{\bbM}\vert)\,\vert P_{\bbM}\vert}\, \ip{G_{\bbM/\pt_n}^\mp(H_1)}{H_2}_{\bbM/\pt_n} ~,
\end{flalign}
for all homogeneous $H_1,H_2\in\Lambda_n\otimes \OO_{\cc}(\bbM)$.

\item[(ii)] Let $\bbM$ be any object in $\eSLoc$ and $\pt_n$ any object in $\SPt^\op$. Then the 
sequence of right $\Lambda_n$-linear maps
\begin{flalign}\label{eqn:enrichedsequence}
\xymatrixcolsep{3pc}\xymatrix{
0 \to  \Lambda_n\otimes \OO_{\cc}(\bbM) \ar[r]^-{ P_{\bbM/\pt_n}} &
\Lambda_n\otimes \OO_{\cc}(\bbM) \ar[r]^-{ G_{\bbM/\pt_n}} &\Lambda_n\otimes \OO_{\sc}(\bbM)\ar[r]^-{ P_{\bbM/\pt_n}} & \Lambda_n\otimes \OO_{\sc}(\bbM)
}
\end{flalign}
is a complex which is exact everywhere.

\item[(iii)] Let $\pt_n$ be any object in $\SPt^\op$ and $\chi\in\eSLoc(\bbM,\bbM^\prime)(\pt_n)$. Then
\begin{flalign}
G^\pm_{\bbM/\pt_n} = \chi^\ast\circ G^\pm_{\bbM^\prime/\pt_n}\circ \chi_\ast~.
\end{flalign}
\end{itemize}
\end{lem}

\subsection{The \texorpdfstring{$\SSet$}{SSet}-functor \texorpdfstring{$\eLL : \eSLoc \to\eXX$}{eL : eSLoc -> eX}}
Given any object $\bbV = (V,\tau)$ in
$\XX$ and any object $\pt_n$ in $\SPt^\op$, we can consider the object
$\Lambda_n\otimes V$ in $\Lambda_n\text{-}\Mod$ and define a 
$\Lambda_n\text{-}\Mod$-morphism
\begin{flalign}
\nn \tau_{\pt_n} : (\Lambda_n\otimes V)\otimes_{\Lambda_n}(\Lambda_n\otimes V)&\longrightarrow \Lambda_n~,\\
(\zeta_1\otimes v_1)\otimes_{\Lambda_n}(\zeta_2\otimes v_2) &\longmapsto (-1)^{\vert \zeta_2\vert\,\vert v_1\vert}\,\zeta_1\zeta_2\,\tau(v_1,v_2)~.
\end{flalign}
Let us now enrich the category $\XX$ given in Definition \ref{defi:XX}.
\begin{defi}\label{defi:eXX}
The $\SSet$-category $\eXX$ is given by the following data:
\begin{itemize}
\item The objects are all objects $\bbV = (V,\tau)$ in $\XX$.

\item For any two objects $\bbV$ and $\bbV^\prime$ in $\eXX$, the object of morphisms
from $\bbV$ to $\bbV^\prime$ is the following functor $\eXX(\bbV,\bbV^\prime) : \SPt^\op\to \Set$:
For any object $\pt_n$ in $\SPt^\op$ we define $\eXX(\bbV,\bbV^\prime)(\pt_n)$ to 
be the set of all $\Lambda_n\text{-}\Mod$-morphisms
$L : \Lambda_n\otimes V\to \Lambda_n\otimes V^\prime$ satisfying
$\tau_{\pt_n}^\prime\circ (L\otimes_{\Lambda_n} L) = \tau_{\pt_n}$.
For any $\SPt^\op$-morphism 
$\lambda^\op: \pt_n\to\pt_m$ we define the map of sets
\begin{flalign}
\eXX(\bbV,\bbV^\prime)(\lambda^\op):= \lambda^\op_\ast : \eXX(\bbV,\bbV^\prime)(\pt_n)\longrightarrow \eXX(\bbV,\bbV^\prime)(\pt_m)~,
\end{flalign}
where $\lambda^\op_\ast$ is given in (\ref{eqn:lambdaopastSVec}). 

\item The composition and identity morphisms are defined as in the $\SSet$-category 
$\eSVec$, see Definition \ref{defi:eSVec}.
\end{itemize}
\end{defi}
\sk

We can now define the $\SSet$-functor $\eLL : \eSLoc\to \eXX$ as follows:
To any object $\bbM$ in $\eSLoc$ we assign the object $\eLL(\bbM) 
:=\LL(\bbM)$ in $\eXX$ that is given in (\ref{eqn:SSspace}).
 To any two objects $\bbM$ and $\bbM^\prime$
in $\eSLoc$ we assign the $\SSet$-morphism $\eLL_{\bbM,\bbM^\prime} : \eSLoc(\bbM,\bbM^\prime)\to
\eXX(\LL(\bbM),\LL(\bbM^\prime))$ given by the natural transformation (of functors
from $\SPt^\op$ to $\Set$) with components 
\begin{flalign}\label{eqn:chiastenr}
(\eLL_{\bbM,\bbM^\prime})_{\pt_n} (\chi)  : \Lambda_n\otimes \LL(\bbM)\longrightarrow \Lambda_n\otimes \LL(\bbM^\prime)~,~~
[H] \longmapsto \big[\chi_\ast(H)\big]~,
\end{flalign}
for all $\chi \in \eSLoc(\bbM,\bbM^\prime)(\pt_n)$. Notice that (\ref{eqn:chiastenr}) is well-defined
since
\begin{flalign}
\nn \chi_\ast \big(\zeta\otimes P_{\bbM}(F) \big) &= (-1)^{\vert P_{\bbM}\vert\,\vert\zeta\vert}\,
\chi_\ast \circ  P_{\bbM/\pt_n} \circ \chi^\ast\circ \chi_\ast(\zeta\otimes F) \\
\nn &=(-1)^{\vert P_{\bbM}\vert\,\vert\zeta\vert}\,
\chi_\ast \circ \chi^\ast\circ P_{\bbM^\prime/\pt_n} \circ \chi_\ast(\zeta\otimes F) \\
&= (-1)^{\vert P_{\bbM}\vert\,\vert\zeta\vert}\, P_{\bbM^\prime/\pt_n} \circ \chi_\ast(\zeta\otimes F)
\in\Lambda_n\otimes P_{\bbM^\prime}(\OO_{\cc}(\bbM^\prime))~,
\end{flalign}
where in the first and third equality we have used Lemma \ref{lem:relativepushforward} (i)
and in the second equality the commutative diagram (\ref{eqn:enricheddiffopsimplediagrams}).
Moreover, (\ref{eqn:chiastenr}) defines an element in $\eXX(\LL(\bbM),\LL(\bbM^\prime))(\pt_n)$ because of the equality
\begin{flalign}\label{eqn:tauenriched}
\tau_{\bbM/\pt_n}\left([H_1] ,[H_2]\right)= \ip{G_{\bbM/\pt_n}(H_1)}{H_2}_{\bbM/\pt_n}~,
\end{flalign}
for all $[H_1],[H_2]\in\Lambda_n\otimes \LL(\bbM)$ and all objects $\bbM$ in $\eSLoc$, 
from which it follows that $\tau_{\bbM^\prime/\pt_n} \circ 
((\eLL_{\bbM,\bbM^\prime})_{\pt_n} (\chi)\otimes_{\Lambda_n}(\eLL_{\bbM,\bbM^\prime})_{\pt_n} (\chi)) = \tau_{\bbM/\pt_n}$,
 for all $\chi\in \eSLoc(\bbM,\bbM^\prime)(\pt_n)$; indeed,
\begin{flalign}
\nn \tau_{\bbM^\prime/\pt_n}\left([\chi_\ast(H_1)] ,[\chi_\ast(H_2)]\right) &= 
 \ip{G_{\bbM^\prime/\pt_n}\circ \chi_\ast(H_1)}{\chi_\ast(H_2)}_{\bbM^\prime/\pt_n}\\
\nn &= \ip{\chi^\ast\circ G_{\bbM^\prime/\pt_n}\circ \chi_\ast(H_1)}{H_2}_{\bbM/\pt_n}\\
&=\ip{G_{\bbM/\pt_n}(H_1)}{H_2}_{\bbM/\pt_n} = \tau_{\bbM/\pt_n}\left([H_1],[H_2]\right)~,
\end{flalign}
for all $[H_1],[H_2]\in \Lambda_n\otimes \LL(\bbM)$,
where in the second equality we have used Lemma \ref{lem:relativepushforward} (ii) and in the third equality
Lemma \ref{lem:enrichedgreen} (iii). In summary, we have shown 
\begin{propo}
The assignment $\eLL : \eSLoc \to \eXX$ given above is a $\SSet$-functor.
\end{propo}

We finish this subsection by observing that the $\SSet$-functor $\eLL :\eSLoc \to \eXX$ satisfies an
enriched version of the properties in Theorem \ref{theo:LCCFT} for ordinary super-field theories, 
which can be proven in exactly the same way:
\begin{theo}\label{theo:enLCCFT}
For any enriched super-field theory according to Definition \ref{defi:sftenriched} the associated $\SSet$-functor
$\eLL :\eSLoc\to\eXX$ satisfies the following properties, for all objects $\pt_n$ in $\SPt^\op$:
\begin{itemize}
\item Enriched locality: For any $\chi \in \eSLoc(\bbM,\bbM^\prime)(\pt_n)$, we 
have that $(\eLL_{\bbM,\bbM^\prime})_{\pt_n}(\chi)\in\eXX(\LL(\bbM),\LL(\bbM^\prime))(\pt_n)$ is monic.
\item Enriched super-causality: Given $\chi_1 \in \eSLoc(\bbM_1,\bbM)(\pt_n)$ and $\chi_2 \in \eSLoc(\bbM_2,\bbM)(\pt_n)$,
such that the images of the reduced $\tLor$-morphisms $\und{\bbM_1}\stackrel{\und{\chi_1}}{\longrightarrow} 
\und{\bbM} \stackrel{\und{\chi_2}}{\longleftarrow} \und{\bbM_2}$ are causally disjoint, then
\begin{flalign}
\tau_{\bbM/\pt_n}\left((\eLL_{\bbM_1,\bbM})_{\pt_n}(\chi_1)\big(\Lambda_n\otimes \LL(\bbM_1)\big) , 
(\eLL_{\bbM_2,\bbM})_{\pt_n}(\chi_2)\big(\Lambda_n\otimes \LL(\bbM_2)\big)\right) =\{0\}~.
\end{flalign}
\item Enriched time-slice axiom: Given any $\chi \in \eSLoc(\bbM,\bbM^\prime)(\pt_n)$ such that
$\und{\chi} : \und{\bbM}\to\und{\bbM^\prime}$ is Cauchy, then
 $(\eLL_{\bbM,\bbM^\prime})_{\pt_n}(\chi) \in\eXX(\LL(\bbM),\LL(\bbM^\prime))(\pt_n)$ is an isomorphism.
\end{itemize} 
\end{theo}

\subsection{The \texorpdfstring{$\SSet$}{SSet}-quantization functor \texorpdfstring{$\eQQ : \eXX \to \eSAlg $}{eQ : eX -> eS*Alg}}
We define an enriched version of the category $\SAlg$ using ``extension of scalars'' for algebras 
(cf.\ \cite[Chapter  III.1.5]{BourbakiA1} for the general concept) and adapt the results obtained in Subsection
\ref{ssec:eOO} to the category of super-$\ast$-algebras. For this let us denote by $\Lambda_n^\bbC:=\Lambda_n\otimes \bbC$
 the complexification of the Grassmann algebra, for all $n\in \bbN^0$,
and notice that $\Lambda_n^\bbC$ is an object in $\SAlg$ when equipped with the superinvolution
$\ast := \id_{\Lambda_n} \otimes \overline{\,\cdot\,} : \Lambda_n^\bbC\to \Lambda_n^\bbC$.
We shall denote the product in $\Lambda_n^\bbC$ by $\mu_n^\bbC$ and the unit by $\eta_n^\bbC$.
Let $A$ and  $A^\prime$ be any two objects in $\SAlg$ and $\pt_n$ any object in $\SPt^\op$.
A $\SAlg$-morphism $\kappa : \Lambda_n^\bbC \otimes_{\bbC} A \to\Lambda_n^\bbC\otimes_{\bbC} A^\prime$
is called a {\em $\Lambda^\bbC_n$-relative $\SAlg$-morphism} (in short $\Lambda_n^{\bbC}\text{-}\SAlg$-morphism) 
provided that $\kappa(\zeta\otimes_{\bbC} \1) = \zeta \otimes_{\bbC}\1$ for all $\zeta\in\Lambda_n^{\bbC}$,
 i.e.\ $\kappa$ is $\Lambda_n$-superlinear.
Notice that the identity $\id_{\Lambda_n^\bbC\otimes_{\bbC} A}$ is a $\Lambda_{n}^\bbC\text{-}\SAlg$-morphism
and that any two $\Lambda_{n}^\bbC\text{-}\SAlg$-morphisms 
$\kappa : \Lambda_n^\bbC \otimes_{\bbC} A \to\Lambda_n^\bbC\otimes_{\bbC} A^\prime$ and
$\kappa^\prime : \Lambda_n^\bbC \otimes_{\bbC} A^\prime \to\Lambda_n^\bbC\otimes_{\bbC} A^{\prime\prime}$
can be composed, i.e.\ $\kappa^\prime \circ \kappa : \Lambda_n^\bbC \otimes_{\bbC} A \to
\Lambda_n^\bbC\otimes_{\bbC} A^{\prime\prime}$ is a $\Lambda_{n}^\bbC\text{-}\SAlg$-morphism.
In analogy to Lemma \ref{lem:SVecptcharacterization}, the $\Lambda_n^\bbC\text{-}\SAlg$-morphisms 
$\kappa : \Lambda_n^{\bbC}\otimes_{\bbC} A\to \Lambda_n^{\bbC}\otimes_{\bbC} A^\prime$
can be easily characterized.
\begin{lem}\label{lem:relalgmorph}
Let $A$ and $A^\prime$ be any two objects in $\SAlg$ and $\pt_n$ any object in $\SPt^\op$.  Then the map
\begin{flalign}
\nn \gamma_{\pt_n} : \Hom_{\Lambda_n^\bbC\text{-}\SAlg}(\Lambda^\bbC_n\otimes_{\bbC}A,\Lambda_{n}^{\bbC}\otimes_{\bbC}A^\prime) &\longrightarrow \Hom_{\SAlg}(A,\Lambda_n^{\bbC}\otimes_{\bbC} A^\prime)~,\\
\big(\kappa : \Lambda_n^{\bbC}\otimes_{\bbC} A\to \Lambda^n_{\bbC}\otimes_{\bbC} A^\prime\big)&\longmapsto
\big(\kappa \circ (\eta_n^\bbC\otimes_{\bbC}\id_{A}) :  A\to \Lambda_n^{\bbC}\otimes_{\bbC} A^\prime\big)
\end{flalign}
is a bijection of sets.
\end{lem}

Given any two objects $A$ and $A^\prime$ in $\SAlg$ and any $\SPt^\op$-morphism
$\lambda^\op : \pt_n\to\pt_m$ (i.e.\ a $\SPt$-morphism $\lambda : \pt_m\to\pt_n$)
we can define a map of sets
\begin{flalign}
\Hom_{\SAlg}(A,\Lambda_n^{\bbC}\otimes_{\bbC} A^\prime)\longrightarrow \Hom_{\SAlg}(A,\Lambda_m^{\bbC}\otimes_{\bbC} A^\prime)~,~~\varphi \longmapsto (\lambda^\ast \otimes_{\bbC}\id_{A^\prime})\circ \varphi~,
\end{flalign}
where the extension of $\lambda^\ast : \Lambda_n\to \Lambda_m$ to the complexifications is implicitly understood.
Using also Lemma \ref{lem:relalgmorph} we obtain a map of sets
\begin{flalign}
\nn \lambda^\op_\ast : \Hom_{\Lambda_n^\bbC\text{-}\SAlg}(\Lambda^\bbC_n \otimes_{\bbC}A,\Lambda_{n}^{\bbC}\otimes_{\bbC}A^\prime)&\longrightarrow \Hom_{\Lambda_m^\bbC\text{-}\SAlg}(\Lambda^\bbC_m\otimes_{\bbC}A,\Lambda_{m}^{\bbC}\otimes_{\bbC}A^\prime)~,\\
\kappa &\longmapsto \gamma_{\pt_m}^{-1}\big((\lambda^\ast\otimes_{\bbC}\id_{A^\prime})\circ \gamma_{\pt_n}(\kappa)\big)~.\label{eqn:changebasisalg}
\end{flalign}
The following properties can be easily derived from (\ref{eqn:changebasisalg}). We therefore can omit the proof.
\begin{lem}\label{lem:lambdaopastSAlg}
\begin{itemize}
\item[(i)] For any identity $\SPt^\op$-morphism $\lambda^\op = \id_{\pt_n} : \pt_n\to\pt_n$
the map $\lambda^\op_\ast$ is the identity. For any two $\SPt^\op$-morphisms $\lambda^\op : \pt_n\to\pt_m$ and
 $\lambda^{\prime\op} : \pt_m\to\pt_l$ 
we have that $(\lambda^{\prime\op}\circ^\op \lambda^\op)_\ast =\lambda^{\prime\op}_\ast\circ \lambda^{\op}_\ast $.
\item[(ii)] $\lambda^\op_\ast$ preserves identities and compositions, i.e.\
\begin{flalign}
\lambda^\op_\ast(\id_{\Lambda_n^\bbC\otimes_\bbC A}) = \id_{\Lambda^\bbC_m\otimes_{\bbC} A}~,\qquad \lambda^\op_\ast(\kappa^\prime\circ \kappa) = \lambda^\op_\ast(\kappa^\prime)\circ \lambda^\op_\ast(\kappa) ~,
\end{flalign}
for all objects $A$ in $\SAlg$ and 
all $\Lambda_n^\bbC\text{-}\SAlg$-morphisms $\kappa : 
\Lambda_n^\bbC\otimes_{\bbC} A\to\Lambda_n^\bbC \otimes_{\bbC} A^\prime$ 
and $\kappa^\prime : \Lambda_n^\bbC\otimes_{\bbC} A^{\prime} \to\Lambda_{n}^\bbC\otimes_{\bbC} A^{\prime\prime}$.
\item[(iii)] $\lambda^\op_\ast$ preserves isomorphisms, i.e.\ 
$\kappa : \Lambda_n^\bbC\otimes_{\bbC} A\to\Lambda_n^\bbC\otimes_{\bbC} A^\prime$ is 
a $\Lambda_n^\bbC\text{-}\SAlg$-isomorphism if and only if 
$\lambda^\op_\ast(\kappa) : \Lambda_m^\bbC\otimes_{\bbC} A\to\Lambda_m^\bbC \otimes_{\bbC} A^\prime$ is a
$\Lambda_m^\bbC\text{-}\SAlg$-isomorphism.
\end{itemize}
\end{lem}
\sk

With these preparations we can now define the $\SSet$-category $\eSAlg$.
\begin{defi}
The $\SSet$-category $\eSAlg$ is given by the following data:
\begin{itemize}
\item The objects are all objects $A$ in $\SAlg$.

\item For any two objects $A$ and $A^\prime$ in $\eSAlg$, 
the object of morphisms from $A$ to $A^\prime$ is given by the following functor $\eSAlg(A,A^\prime) : \SPt^\op \to \Set$:
 For any object $\pt_n$ in $\SPt^\op$ we define $\eSAlg(A,A^\prime)(\pt_n)$ to be the set of
all $\Lambda_n^\bbC\text{-}\SAlg$-morphisms $\kappa : \Lambda_n^\bbC\otimes_{\bbC}A\to 
\Lambda_n^\bbC\otimes_{\bbC} A^\prime$.
For any $\SPt^\op$-morphism $\lambda^\op :\pt_n\to\pt_m$ we define the
map of sets
\begin{flalign}
\eSAlg(A,A^\prime)(\lambda^\op) := \lambda^\op_\ast : \eSAlg(A,A^\prime)(\pt_n)\longrightarrow \eSAlg(A,A^\prime)(\pt_m)~,
\end{flalign}
where $\lambda^\op_\ast$ is given in (\ref{eqn:changebasisalg}).

\item For any three objects $A$, $A^\prime$ and $A^{\prime\prime}$ in $\eSAlg$, we define the composition
morphism $\bullet : \eSAlg(A^\prime,A^{\prime\prime})\times \eSAlg(A,A^\prime) \to \eSAlg(A,A^{\prime\prime})$
to be the natural transformation with components 
\begin{flalign}
\bullet_{\pt_n}:=\circ : \eSAlg(A^\prime,A^{\prime\prime})(\pt_n)\times \eSAlg(A,A^\prime)(\pt_n)\longrightarrow\eSAlg(A,A^{\prime\prime})(\pt_n)~,
\end{flalign}
where $\circ$ is the composition of $\Lambda_n^\bbC\text{-}\SAlg$-morphisms.

\item For any object $A$ in $\eSAlg$, we define the identity on $A$ morphism
$\oone : \mathfrak{I} \to  \eSAlg(A,A)$ to be the natural transformation with components
\begin{flalign}
\oone_{\pt_n}: \pt= \mathfrak{I}(\pt_n) \longrightarrow \eSAlg(A,A)(\pt_n)~,~~\star\longmapsto \id_{\Lambda_n^\bbC\otimes_{\bbC} A}~,
\end{flalign}
where $\id_{\Lambda_n^\bbC\otimes_{\bbC}A}$ is the identity $\Lambda_n^\bbC\text{-}\SAlg$-morphism.
\end{itemize}
\end{defi}
\sk

The quantization $\SSet$-functor $\eQQ : \eXX \to \eSAlg$ is constructed as follows:
To any object $\bbV$ in $\eXX$ we assign the object $\eQQ(\bbV) := \QQ(\bbV)$ in $\eSAlg$
that has been constructed in (\ref{eqn:Qalg}).
To any two objects $\bbV$ and $\bbV^\prime$ in $\eXX$
we assign the $\SSet$-morphism $\eQQ_{\bbV,\bbV^\prime} : \eXX(\bbV,\bbV^\prime) \to \eSAlg(\QQ(\bbV),\QQ(\bbV^\prime))$
given by the natural transformation (of functors from $\SPt^\op$ to $\Set$) with components
 \begin{flalign}
 (\eQQ_{\bbV,\bbV^\prime})_{\pt_n} : \big(L: \Lambda_n\otimes V\to \Lambda_n\otimes V^\prime\big) 
\mapsto \big((\eQQ_{\bbV,\bbV^\prime})_{\pt_n} (L) :  \Lambda^\bbC_n\otimes_{\bbC} \QQ(\bbV) \to\Lambda^\bbC_n\otimes_{\bbC} \QQ(\bbV^\prime)\big)\,,
 \end{flalign}
 where  $(\eQQ_{\bbV,\bbV^\prime})_{\pt_n} (L)$ is the $\Lambda_n^\bbC\text{-}\SAlg$-morphism which is specified
by defining on the generators $(\eQQ_{\bbV,\bbV^\prime})_{\pt_n} (L)(\zeta\otimes_{\bbC} v) := L(\zeta\otimes v)$,
for all $v\in V$ and $\zeta\in \Lambda_n$. 
It remains to show that $(\eQQ_{\bbV,\bbV^\prime})_{\pt_n} (L)$ is well-defined, i.e.\ that it
preserves the two-sided super-$\ast$-ideals (\ref{eqn:comrel}).  Written in terms of the tensor product
superalgebra $\Lambda_n^\bbC\otimes_{\bbC} \mathcal{T}_\bbC(\bbV)$ the 
super-canonical (anti)commutation relation super-$\ast$-ideal is generated by the elements
\begin{flalign}
w_1\,w_2+ (-1)^{\dim(S)+1}\,(-1)^{\vert w_1\vert\,\vert w_2\vert}\,w_2\,w_1 - \beta\, \tau_{\pt_n}(w_1,w_2)\otimes_{\bbC}\1~,
\end{flalign}
for all homogeneous $w_1,w_2\in\Lambda_n\otimes V$. Using now 
that by definition of $\eXX$, $\tau_{\pt_n}^\prime\circ (L\otimes_{\Lambda_n} L) = \tau_{\pt_n}$, 
we obtain that $(\eQQ_{\bbV,\bbV^\prime})_{\pt_n} (L)$ 
is a well-defined $\Lambda_n^\bbC\text{-}\SAlg$-morphism. 
By direct inspection one further observes that $\eQQ_{\bbV,\bbV^\prime} $ is compatible
with composition and identity.
We therefore have shown
\begin{propo}
The assignment $\eQQ : \eXX \to \eSAlg$ given above is a $\SSet$-functor.
\end{propo}

\subsection{The enriched locally covariant quantum field theory \texorpdfstring{$\eAA : \eSLoc \to \eSAlg $}{eA : eSLoc -> eS*Alg}}
Recalling Remark \ref{rem:enrichedcomposition}, we can compose the two $\SSet$-functors
$\eLL :\eSLoc\to \eXX$ and $\eQQ : \eXX \to \eSAlg$
in order to define the $\SSet$-functor
\begin{flalign}
\eAA := \eQQ\circ\eLL : \eSLoc \longrightarrow \eSAlg~.
\end{flalign}
By using the same arguments as in the proof of Theorem \ref{theo:LCQFT},
the results of Theorem \ref{theo:enLCCFT} imply
\begin{theo}\label{theo:enLCQFT}
For any enriched super-field theory according to Definition \ref{defi:sftenriched} the associated $\SSet$-functor
$\eAA :\eSLoc\to\eSAlg$ satisfies the following properties, for all objects $\pt_n$ in $\SPt^\op$:
\begin{itemize}
\item Enriched locality: For any $\chi \in \eSLoc(\bbM,\bbM^\prime)(\pt_n)$, we have that 
$(\eAA_{\bbM,\bbM^\prime})_{\pt_n}(\chi)\in\eSAlg(\AA(\bbM),\AA(\bbM^\prime))(\pt_n)$ is monic.

\item Enriched super-causality: Given $\chi_1 \in \eSLoc(\bbM_1,\bbM)(\pt_n)$ and $\chi_2 \in \eSLoc(\bbM_2,\bbM)(\pt_n)$,
such that the images of the reduced $\tLor$-morphisms 
$\und{\bbM_1}\stackrel{\und{\chi_1}}{\longrightarrow} \und{\bbM} \stackrel{\und{\chi_2}}{\longleftarrow} \und{\bbM_2}$
 are causally disjoint, then
\begin{flalign}
A_1\,A_2 + (-1)^{\dim(S)+1} \,(-1)^{\vert A_1\vert \,\vert A_2\vert}\,A_2\,A_1 =0~,
\end{flalign}
for all homogeneous $A_1\in (\eAA_{\bbM_1,\bbM})_{\pt_n}(\chi_1)\big(\Lambda_n^\bbC\otimes_{\bbC} \AA(\bbM_1)\big)
\subseteq \Lambda_n^\bbC\otimes_{\bbC} \AA(\bbM)$
and $A_2\in (\eAA_{\bbM_2,\bbM})_{\pt_n}(\chi_2)\big(\Lambda_n^\bbC\otimes_{\bbC} \AA(\bbM_2)\big)\subseteq
\Lambda_n^\bbC\otimes_{\bbC} \AA(\bbM)$.

\item Enriched time-slice axiom: Given any $\chi \in \eSLoc(\bbM,\bbM^\prime)(\pt_n)$ such that
$\und{\chi} : \und{\bbM}\to\und{\bbM^\prime}$ is Cauchy, then 
$(\eAA_{\bbM,\bbM^\prime})_{\pt_n}(\chi) \in\eSAlg(\AA(\bbM),\AA(\bbM^\prime))(\pt_n)$ is an isomorphism.
\end{itemize} 
\end{theo}
\sk

We conclude this section by showing that the $\SSet$-functor $\eAA$ has
an enriched locally covariant quantum field given by a $\SSet$-natural transformation (cf.\ Definition \ref{defi:enrichednatural})
$\Phi : \eOO_{\cc} \Rightarrow \eAA$ of $\SSet$-functors from $\eSLoc$ to $\eSVec$, where $\eAA$ is implicitly
understood to be composed with the forgetful $\SSet$-functor $\eSAlg\to\eSVec$.
We assign to every object $\bbM$ in $\eSLoc$ the $\SSet$-morphism $\Phi_{\bbM} : 
\mathfrak{I}\to \eSVec(\OO_{\cc}(\bbM),\AA(\bbM))$ given by the natural transformation (of functors
from $\SPt^\op$ to $\Set$) with components
\begin{subequations}\label{eqn:enlcqf}
\begin{flalign}
\nn(\Phi_{\bbM})_{\pt_n} : \pt = \mathfrak{I}(\pt_n) &\longrightarrow \eSVec(\OO_{\cc}(\bbM),\AA(\bbM))(\pt_n)~,\\
\star &\longmapsto (\Phi_{\bbM})_{\pt_n} (\star) =:\Phi_{\bbM/\pt_n} 
\end{flalign}
given by the $\Lambda_n\text{-}\Mod$-morphisms
\begin{flalign}
\Phi_{\bbM/\pt_n}  : \Lambda_n\otimes \OO_{\cc}(\bbM)\longrightarrow \Lambda_{n}^\bbC\otimes_{\bbC} \AA(\bbM)~,
~~\zeta \otimes F \longmapsto \zeta\otimes_{\bbC} [F]~.
\end{flalign}
\end{subequations}
Notice that the diagram in Definition \ref{defi:enrichednatural} commutes, i.e.\
\begin{flalign}
\Phi_{\bbM^\prime/\pt_n} \bullet_{\pt_n} ({\eOO_{\cc}}_{\bbM,\bbM^\prime})_{\pt_n}(\chi) =
(\eAA_{\bbM,\bbM^\prime})_{\pt_n}(\chi)\bullet_{\pt_n}  \Phi_{\bbM/\pt_n}~,
\end{flalign}
for all $\chi\in\eSLoc(\bbM,\bbM^\prime)(\pt_n)$ and all objects  $\pt_n$ in $\SPt^\op$.
In complete analogy to the ordinary case discussed in Subsection \ref{subsec:lcqft}
we obtain that 
\begin{flalign}
\Phi_{\bbM/\pt_n}\big(P_{\bbM/\pt_n}(H)\big) = 0~,
\end{flalign}
for all $H\in\Lambda_n\otimes \OO_{\cc}(\bbM)$,
 all objects $\bbM$ in $\eSLoc$ and all objects $\pt_n$ in $\SPt^\op$,
as well as the super-canonical (anti)commutation relations
\begin{flalign}
\Phi_{\bbM/\pt_n}(H_1)\,\Phi_{\bbM/\pt_n}(H_2) + (-1)^{\vert H_1\vert\,\vert H_2\vert}\,\Phi_{\bbM/\pt_n}(H_2)\,\Phi_{\bbM/\pt_n}(H_1) = \beta\,\tau_{\bbM/\pt_n}(H_1,H_2)~,
\end{flalign}
for all homogeneous $H_1,H_2\in \Lambda_n\otimes \OO_{\cc}(\bbM)$, all objects $\bbM$ in $\eSLoc$ and
all objects $\pt_n$ in $\SPt^\op$.
\sk

In contrast to the ordinary case discussed in Remark \ref{rem:componentfields},
the even and odd component quantum fields {\em do not} 
form natural transformations in our enriched setting. 
The reason for this is that the push-forward
$({\eOO_{\cc}}_{\bbM,\bbM^\prime})_{\pt_n}(\chi) =\chi_\ast : \Lambda_{n}\otimes \OO_{\cc}(\bbM) 
\to \Lambda_n\otimes \OO_{\cc}(\bbM^\prime)$
along a generic $\chi\in \eSLoc(\bbM,\bbM^\prime)(\pt_n)$
{\em does not} restrict to a $\Lambda_n\text{-}\Mod$-morphism
$\Lambda_{n}\otimes \OO^{\mathrm{even}/\mathrm{odd}}_{\cc}(\bbM) \to
\Lambda_{n}\otimes \OO^{\mathrm{even}/\mathrm{odd}}_{\cc}(\bbM^\prime) $.
In fact, the push-forward of a homogeneous element 
$\1\otimes F\in\Lambda_n\otimes \OO^{\mathrm{even}/\mathrm{odd}}_{\cc}(\bbM)$
schematically reads as
\begin{flalign} \label{eqn:SUSYtrafo}
\chi_\ast(\1\otimes F) = \sum_{I}\zeta^{I}\otimes F_I~,
\end{flalign}
where $\{\zeta^I\in\Lambda_n\}$ is any adapted basis for $\Lambda_n$ and $F_I\in\OO_{\cc}(\bbM^\prime)$
is of $\bbZ_2$-parity $\vert F_I\vert = \vert F\vert - \vert \zeta_I\vert ~\text{mod} ~2$.
Hence, if there is a non-vanishing coefficient $F_I$ for some odd $\zeta^I$, then
we can not restrict $\chi_\ast$ to a $\Lambda_n\text{-}\Mod$-morphism
$\Lambda_{n}\otimes \OO^{\mathrm{even}/\mathrm{odd}}_{\cc}(\bbM) \to
\Lambda_{n}\otimes \OO^{\mathrm{even}/\mathrm{odd}}_{\cc}(\bbM^\prime)$.
Those $\chi$ with this property can be identified with supersymmetry transformations, see
Section \ref{sec:examples} for some illustrating examples.
In summary, our formalism implies that  (as expected) supersymmetry transformations 
prevent the even and odd component quantum fields to be enriched natural transformations.

\begin{rem}
We would like to emphasize that the presence of the supersymmetry transformations is due to the 
enrichment of the involved categories over $\SSet$. The morphism 
$\chi\in \eSLoc(\bbM,\bbM^\prime)(\pt_n) \subseteq \Hom_{\SMan/\pt_n}(M/\pt_n,M'/\pt_n)$ appearing
in \eqref{eqn:SUSYtrafo} respects $\bbZ_2$-parity by definition. However, due to the presence of odd sections
of the structure sheaf $\Lambda_n$ of $\pt_n$, the action of $\chi_\ast$ need not preserve the splittings  
$\OO_\cc(M^{(\prime)}) = \OO_\cc^{\mathrm{even}}(M^{(\prime)}) \oplus \OO_\cc^{\mathrm{odd}}(M^{(\prime)})$ 
of the second tensor factor. Moreover, these odd parameters appearing in the structure sheaf
$\Lambda_n$ of $\pt_n$ are exactly the odd quantities used in the physics literature to parametrize 
supersymmetry transformations. 
\end{rem}

\begin{rem}
We note that our enriched super-QFT $\eAA :\eSLoc\to\eSAlg$ 
cannot be restricted to its even part\footnote{
We would like to thank the anonymous referee for suggesting this question to us.
}  $\eAA_0$ 
(cf.\ Remark \ref{rem:evensubtheory} for such a construction for the ordinary super-QFT
$\AA :\SLoc\to\SAlg$): Recall that the $\SSet$-morphism 
$\eAA_{\bbM,\bbM^\prime} : \eSLoc(\bbM,\bbM^\prime) \to \eSAlg(\eAA(\bbM),\eAA(\bbM^{\prime}))$
is a natural transformation (of functors from $\SPt^{\op}$ to $\Set$) whose components 
$(\eAA_{\bbM,\bbM^\prime})_{\pt_n}$ assign to elements $\chi\in\eSLoc(\bbM,\bbM^\prime)(\pt_n)$
certain super-$\ast$-algebra morphisms $(\eAA_{\bbM,\bbM^\prime})_{\pt_n}(\chi) : 
\Lambda^{\bbC}_n\otimes_{\bbC} \eAA(\bbM) \to \Lambda_n^{\bbC}\otimes_{\bbC} \eAA(\bbM^\prime)$.
If now $\chi$ corresponds to a supersymmetry transformation, i.e.\ it is like in (\ref{eqn:SUSYtrafo})
with a non-vanishing coefficient $F_I$ for some odd $\zeta^I$, then $(\eAA_{\bbM,\bbM^\prime})_{\pt_n}(\chi)$
does {\em not} map 
$\Lambda^{\bbC}_n\otimes_{\bbC} \eAA_0(\bbM)$ to $\Lambda^{\bbC}_n\otimes_{\bbC} \eAA_0(\bbM^\prime)$;
indeed, a supersymmetry transformation ``maps an element of the even component 
$\eAA_0(\bbM)$ to some element of the odd component $\eAA_1(\bbM^\prime)$''.
The fact that the even (i.e.\ true) observables in a supergeometric quantum 
field theory do not carry a representation of the supersymmetry transformations 
is however not problematic from the physical point of view. Supersymmetry and supergeometry are
meant to serve as a selection criterion for quantum field theories by demanding that 
only a limited collection of observables, such as for example the action functional and quantities derived from it,
transform covariantly under supersymmetry transformations.
\end{rem}

%%%%%%%%%%%%%%%%%%%%%%%%%%%%%%%%%%%%%%%%%%%%%%%%
%%%%%%%%%%%%%%%%%%%%%%%%%%%%%%%%%%%%%%%%%%%%%%%%

\section{\label{sec:examples}Examples}

\subsection{\texorpdfstring{$1\vert 1$}{1,1}-dimensions}
As our first example we shall study a super-field theory in $1\vert 1$-dimensions, i.e.\ a superparticle.
For defining this theory we have to provide all the data listed in Definition \ref{defi:sft}.

\paragraph{Representation theoretic data:} We take $W = \bbR$ together with 
the standard $1$-dimensional (Lorentz) metric 
\begin{flalign}
g :W\otimes W \longrightarrow \bbR~,~~w_1\otimes w_2 \longmapsto w_1w_2~.
\end{flalign}
The corresponding spin group is $\mathrm{Spin}(W,g) \simeq \{+1,-1\}$ and we take $S=\bbR$ 
the $1$-dimensional spin representation
\begin{flalign}
\rho^{S} : \mathrm{Spin}(W,g) \times S \longrightarrow S\,,~(z,s)\longmapsto z\,s~.
\end{flalign}
Notice that $\rho^{W} : \mathrm{Spin}(W,g) \times W\to W\,,~(z,w)\mapsto w$ is the trivial representation.
As $\Gamma$ we shall take the following  $\mathrm{Spin}(W,g) $-equivariant symmetric pairing
\begin{flalign}
\Gamma : S \otimes S \longrightarrow W~,~~s_1\otimes s_2\longmapsto s_1 s_2 ~.
\end{flalign}
As positive cone we take $C:=\bbR^+ \subset W$ and we notice that $\Gamma(s,s)\in \overline{C}$, for all $s\in S$,
and $\Gamma(s,s) =0$ only for $s=0$. For $\epsilon$ we take 
\begin{flalign}
\epsilon : S \otimes S \longrightarrow \bbR~,~~s_1\otimes s_2\longmapsto s_1 s_2
\end{flalign}
 and we notice that it is symmetric and $\mathrm{Spin}(W,g) $-invariant.
We define orientations on $W$ and $S$ by using the normalized standard bases
$p=1\in W=\bbR$ and $q=1\in S =\bbR$.
In $1\vert 1$-dimensions, the super-Poincar{\'e} super-Lie algebra coincides with the 
supertranslation super-Lie algebra (since $\mathfrak{spin}$ is trivial), and we obtain for 
the super-Lie bracket relations in the normalized adapted basis $\{p,q\}$ for $\mathfrak{sp} = \mathfrak{t} = W\oplus S$
\begin{flalign}
[p,p] =0~,\quad [p,q]=0~,\quad [q,q] = -2\,p~.
\end{flalign}

\paragraph{The objects in $\ghSCart$:} 
Let us characterize explicitly the objects in the category
 $\ghSCart$ for the case of  $1\vert 1$-dimensions.
To simplify our studies we shall further demand that the underlying topological spaces
are connected, which is also physically reasonable as they describe a time interval. 
Given any such object $\bbM = (M,\Omega,E)$ we first notice
 that $\Omega = 0$ since $\mathfrak{spin}$ is trivial in $1\vert 1$-dimensions.
 Moreover, the reduced $1$-dimensional manifold $\und{M}$ is diffeomorphic
 to the real line $\bbR$ as the reduced Lorentz manifold $\und{\bbM}$ is assumed to be globally hyperbolic.
 The structure sheaf $\OO_{M}$ is isomorphic to $C^\infty_{\bbR}\otimes \bigwedge^\bullet\bbR$
 and the supervielbein can be expanded as
 \begin{flalign} 
 E = (\gamma \, \dd t + \alpha\,\theta\,\dd\theta)\otimes p + (\delta\,\dd\theta + \beta\,\theta\,\dd t)\otimes q~,
 \end{flalign}
 where $\alpha,\beta,\gamma,\delta\in C^\infty(\bbR)$ are coefficient functions and 
 $t,\theta\in \OO(\bbM)$ are any global even/odd coordinate functions.  
 As $E$ is by assumption non-degenerate, the functions $\gamma$ and $\delta$ 
 have to be invertible and we may choose new coordinates (denoted with abuse of notation by the same symbols)
 $t\in (t_0,t_1)\subseteq \bbR$ and $\theta$, such that 
\begin{flalign}\label{eqn:11E}
 E = (\dd t + \alpha\,\theta\,\dd\theta)\otimes p + (\dd\theta + \beta\,\theta\,\dd t)\otimes q~,
\end{flalign}
where now $\alpha,\beta\in C^\infty(t_0,t_1)$.
Coordinate functions in which $E$ takes the form (\ref{eqn:11E}) are called {\em geometric coordinates}. 
The supercurvature $R_{\bbM}=0$ vanishes in $1\vert 1$-dimensions and the supertorsion is given by
\begin{flalign}\label{eqn:torsion11}
T_{\bbM} = \dd E = (\partial_t\alpha\, \theta\,\dd t\wedge \dd\theta + \alpha\,\dd\theta\wedge\dd\theta) \otimes p + 
(\beta\,\dd\theta\wedge\dd t)\otimes q~,
\end{flalign}
where $\partial_t \alpha$ denotes the time derivative of $\alpha$. 
The pairing (\ref{eqn:pairing}) reads as
\begin{flalign}\label{eqn:pairing11}
\ip{F_1}{F_2}_{\bbM} = \int_{t_0}^{t_1} \dd t\, \big(f_1\,h_2 + h_1\,f_2\big)~,
\end{flalign}
where we have used the expansion $F = f + \theta \,h\in\OO(\bbM)$
with $f,h\in C^\infty(t_0,t_1)$.

\paragraph{Super-differential operators:}
Using the dual superderivations 
\begin{flalign}
X = \partial_t - \beta\,\theta\,\partial_{\theta} ~,\quad D = \partial_{\theta} + \alpha\,\theta\,\partial_{t}
\end{flalign}
corresponding to $E$, we define an odd super-differential operator
\begin{flalign}\label{eqn:diffop11}
P_{\bbM} := X\circ D : \OO(\bbM)\longrightarrow \OO(\bbM)\,,~ F=f+\theta \,h\longmapsto \partial_t h + \theta \, \partial_t(\alpha\partial_t f)  - \theta\, \alpha\,\beta\,\partial_t f~.
\end{flalign}
Notice that due to the last term in (\ref{eqn:diffop11}), the super-differential operator $P_{\bbM}$
is i.g.\ not formally super-self adjoint with respect to the pairing (\ref{eqn:pairing11}).
If however $\beta=0$, then $P_{\bbM}$ is formally super-self adjoint.
Comparing with  (\ref{eqn:torsion11}), the constraint $\beta=0$ can be regarded 
as a supertorsion constraint which demands that the odd part of the supertorsion vanishes.
Such constraints also appear in supergravity, see e.g.\ \cite[Eqns.\ (10) and (11)]{Wess:1977fn}.

\paragraph{The category $\SLoc$:}
We define a full category $\SLoc$ of $\ghSCart$ by the conditions that 1.) the underlying
topological spaces are connected and 2.) {\em all} supergravity supertorsion constraints given in 
\cite[Eqns.\ (10) and (11)]{Wess:1977fn} hold true, which implies that $\beta =0$ and that
$\alpha$ is a constant which we fix to $\alpha=1$. 
Then the admissible super-Cartan supermanifolds $\bbM = (M,\Omega=0,E)$ 
are such that $\und{M} = (t_0,t_1)\subseteq\bbR$ is an open interval (or $\bbR$ itself) 
 and 
\begin{flalign}
E = (\dd t + \theta\,\dd\theta)\otimes p + \dd\theta \otimes q~
\end{flalign}
in geometric coordinates.
The morphisms in $\SLoc$ can be explicitly characterized:
Let $\bbM$ and $\bbM^\prime$ be any two objects in $\SLoc$. A $\SMan$-morphism
$\chi : M\to M^\prime$ is specified by its action on the (geometric) coordinates 
$t^\prime,\theta^\prime\in\OO(\bbM^\prime)$, which we can parametrize by the ansatz
\begin{flalign}
\chi^\ast(t^\prime) = a(t)~,\quad \chi^{\ast}(\theta^\prime)  = b(t)\,\theta~,
\end{flalign}
where $a,b\in C^\infty(t_0,t_1)$. In order to qualify as a $\SLoc$-morphism, $\chi$ has to preserve
the supervielbeins
\begin{flalign}
\chi^\ast(E^\prime) = \big(\dd a(t) + b(t)^2\,\theta\, \dd\theta\big)\otimes p + \dd \big(b(t)\,\theta\big) \otimes q = E=(\dd t + \theta\,\dd\theta)\otimes p + \dd\theta \otimes q~,
\end{flalign}
which implies that $a(t) = t+c$, with $c\in\bbR$, and $b(t)\equiv 1$.
Furthermore, for the reduced morphism
$\und{\chi} : \und{M} = (t_0,t_1)\to \und{M^\prime} = (t_0^\prime,t_1^\prime)$
to exist, the constant $c\in\bbR$ has to be such that $(t_0+c,t_1+c)\subseteq (t_0^\prime,t_1^\prime)$,
i.e.\ $t_0^\prime -t_0 \leq c\leq t_1^\prime-t_1$.
For a generic $F^\prime=f^\prime(t^\prime)+\theta^\prime\,h^\prime(t^\prime)\in\OO(\bbM^\prime)$, we have that
\begin{flalign}
\chi^\ast(F^\prime) = f^\prime(t+c) +\theta\,h^\prime(t+c) \in\OO(\bbM)~,
\end{flalign}
hence $\chi : \bbM\to\bbM^\prime$ describes a translation by $c$. 
 It then follows automatically that such $\chi : \bbM\to \bbM^\prime$ are morphisms 
 in $\ghSCart$ (and hence in $\SLoc$),  i.e.\ they satisfy the additional conditions imposed 
 in Definitions \ref{defi:SCart0} and \ref{defi:SCart}.

\paragraph{The natural transformation $P: \OO \Rightarrow \OO$:}
So far we have established  the first two points in the Definition \ref{defi:sft}
of a super-field theory. It remains to show that (\ref{eqn:diffop11}) are the components
of a natural transformation of formally super-self adjoint and super-Green's hyperbolic super-differential
operators. For any object $\bbM$ in $\SLoc$ the super-differential operator (\ref{eqn:diffop11}) 
takes the form
\begin{flalign}\label{eqn:diffop11simple}
P_{\bbM} (F) = \partial_t h +   \theta\,\partial_t^2 f~,
\end{flalign}
for all $F= f + \theta\,h\in \OO(\bbM)$, from which it is clear that it is formally 
super-self adjoint and super-Green's hyperbolic with super-Green's operators given by 
\begin{flalign}
G_{\bbM}^\pm (F) = G^\pm_{\partial_t^2}(h) + \theta\,G_{\partial_t}^{\pm}(f)~,
\end{flalign}
for all $F=f + \theta\,h\in \OO_{\cc}(\bbM)$, where $G^\pm_{\partial_t^2}$ and $G_{\partial_t}^{\pm}$
denote the retarded/advanced Green's operators for the component differential operators
$\partial_t^2$ and $\partial_t$, respectively.
The super-differential operators (\ref{eqn:diffop11simple}) are the components of a natural transformation
since they are translation invariant,
hence we have constructed an example of a super-field theory according to Definition
\ref{defi:sft}. Applying Theorem \ref{theo:LCQFT} we further
obtain a super-QFT, which in the $1\vert 1$-dimensional case describes a quantized superparticle.

\paragraph{Enriched super-field theory:}
We shall now show that the super-field theory defined above is also an enriched super-field theory
according to Definition \ref{defi:sftenriched}.
Let us take any two objects  $\bbM$ and $\bbM^\prime$ in $\SLoc^\op$ and
any object $\pt_n$ in $\SPt^\op$. 
We characterize explicitly the set $\eSLoc^\op(\bbM^\prime,\bbM)(\pt_n)$. Any 
$\SMan/\pt_n$-morphism $\chi : M/\pt_n\to M^\prime/\pt_n$ is specified by its action on the (geometric)
coordinates $\1\otimes t^\prime,\1\otimes \theta^\prime\in\Lambda_n\otimes \OO(M^\prime)$, 
which we can parametrize by the ansatz
\begin{subequations}
\begin{flalign}
\chi^\ast (\1\otimes t^\prime) &= \sum_{I \,\text{even}} \zeta^I\otimes a_I(t) + \sum_{I\,\text{odd}}\zeta^I\otimes  b_I(t)\,\theta~,\\
\chi^\ast(\1\otimes \theta^\prime) &= \sum_{I\,\text{even}} \zeta^I\otimes c_I(t)\,\theta + \sum_{I\,\text{odd}}\zeta^I\otimes  d_I(t)~,
\end{flalign}
\end{subequations}
where $\zeta^I$ is any adapted basis for the super-vector space underlying the Grassmann algebra
 $\Lambda_n$ and $a_I,b_I,c_I,d_I\in C^\infty(t_0,t_1)$. 
 For $\chi\in \eSLoc^\op(\bbM^\prime,\bbM)(\pt_n)$ it has to preserve the $\pt_n$-relative 
 supervielbeins, i.e.\ $\chi^\ast(\1\otimes E^\prime) = \1\otimes E$. 
 The odd part of this condition reads as
\begin{flalign}
\nn \chi^\ast(\1\otimes \dd \theta^\prime) &= (\id_{\Lambda_n} \otimes \dd) \big(\chi^\ast(\1\otimes \theta^\prime)\big)\\
 &=\sum_{I\,\text{even}} \zeta^I\otimes \dd\big(c_I(t)\,\theta\big) 
+ \sum_{I\,\text{odd}}\zeta^I\otimes  \dd\big(d_I(t)\big)= \1\otimes\dd\theta~,
\end{flalign}
and it implies that 
$\chi^\ast(\1\otimes \theta^\prime) = \1\otimes\theta + \zeta\otimes \1$,
for some odd element $\zeta\in\Lambda_n$. Using this result, the even part of the above condition reads as
\begin{flalign}
\nn \chi^\ast\big(\1\otimes(\dd t^\prime +\theta^\prime \,\dd\theta^\prime)\big)  &= 
(\id_{\Lambda_n} \otimes \dd) \big(\chi^\ast(\1\otimes t^\prime)\big)
+ \chi^\ast(\1\otimes \theta^\prime)\,(\id_{\Lambda_n} \otimes \dd) \big(\chi^\ast(\1\otimes \theta^\prime)\big)\\
\nn &=\sum_{I \,\text{even}} \zeta^I\otimes\dd\big( a_I(t) \big)+ \sum_{I\,\text{odd}}\zeta^I\otimes \dd\big( b_I(t)\,\theta\big) + \1\otimes \theta\,\dd\theta + \zeta \otimes \dd\theta\\
&= \1\otimes \big(\dd t +\theta \,\dd\theta\big)~,
\end{flalign}
which implies that $\chi^\ast(\1\otimes t^\prime) = \1\otimes (t +c)  - \zeta\otimes \theta$.
As in the ordinary case, the constant $c\in\bbR$ has to satisfy
$t_0^\prime-t_0\leq c\leq t_1^\prime-t_1 $ for the reduced morphism 
$\und{\chi} : \und{M} = (t_0,t_1)\to\und{M^\prime}= (t^\prime_0,t^\prime_1)$ to exist.
Hence, we have shown that
\begin{flalign}\label{eqn:eSLocmorph11}
\eSLoc^\op(\bbM^\prime,\bbM)(\pt_n) \simeq \{c\in \bbR : t^\prime_0-t_0\leq c\leq t^\prime_1-t_1 \}\times (\Lambda_n)_1~,
\end{flalign} 
where  $(\Lambda_n)_1$ is the odd part of the Grassmann algebra $\Lambda_n = (\Lambda_{n})_0\oplus(\Lambda_n)_1$.
In particular, the $\pt_n$-relative automorphisms $\eSLoc^\op(\bbM,\bbM)(\pt_n)$ are in bijective correspondence with
$\bbR\times (\Lambda_n)_1$ if $\und{M} =\bbR$ and with $(\Lambda_n)_1$  if $\und{M}\subset \bbR$ 
is a proper interval. These $\pt_n$-relative automorphisms describe ordinary and supertranslations.
For a generic $F^\prime = f(t^\prime) + \theta^\prime\,h(t^\prime)\in\OO(\bbM^\prime)$, we have that
\begin{flalign}\label{eqn:11supertrafos}
\chi^\ast(\1\otimes F^\prime) =\1\otimes \big(f^\prime(t+c) + \theta\,h^\prime(t+c)\big)
+\zeta\otimes\big(h^\prime(t+c) - \theta\,\partial_{t} f^\prime(t+c)\big)~,
\end{flalign}
which reproduces the usual supersymmetry transformations.
The diagram (\ref{eqn:enricheddiffopsimplediagrams}) commutes, since using
(\ref{eqn:11supertrafos}) one can easily compute that
\begin{flalign}
\chi^\ast\big(\1\otimes P_{\bbM^\prime}(F^\prime)\big) = (\id_{\Lambda_n}\otimes P_{\bbM})\big(\chi^\ast(\1\otimes F^\prime)\big)~,
\end{flalign}
for all $F^\prime\in\OO(\bbM^\prime)$. We remind the reader that 
$(\id_{\Lambda_n}\otimes P_{\bbM})(\zeta\otimes F) 
= (-1)^{\vert\zeta\vert\,\vert P_{\bbM}\vert} \,\zeta\otimes P_{\bbM}(F)$,
for all homogeneous $\zeta\in\Lambda_n$ and $F\in\OO(\bbM)$.
 As a consequence, we have constructed an enriched super-field theory
according to Definition \ref{defi:sftenriched}. 
Applying Theorem \ref{theo:enLCQFT} we further
obtain an enriched super-QFT, which in the $1\vert 1$-dimensional case describes a quantized 
superparticle together with its supersymmetry transformations.

\paragraph{Supersymmetry transformations in the enriched super-QFT:}
In order to illustrate the structure of supersymmetry transformations 
let us fix any object $\bbM$ in $\SLoc$. The $\SSet$-functor $\eAA : \eSLoc \to \eSAlg$
studied in Theorem \ref{theo:enLCQFT} provides us with a superalgebra of observables
$\AA(\bbM)$ and a $\SSet$-morphism 
$\eAA_{\bbM,\bbM} : \eSLoc(\bbM,\bbM)\to \eSAlg(\AA(\bbM),\AA(\bbM))$ which describes the enriched
automorphism group of $\AA(\bbM)$. Using (\ref{eqn:eSLocmorph11}) and
focusing only on the odd part given by $(\Lambda_n)_1$  (i.e.\ proper supersymmetry transformations),
we obtain the $\Lambda_n^\bbC\text{-}\SAlg$-automorphisms
$(\eAA_{\bbM,\bbM})_{\pt_n}(\chi) : \Lambda_n^\bbC\otimes_{\bbC} \AA(\bbM)
\to\Lambda_n^\bbC\otimes_{\bbC}\AA(\bbM)$, for any object $\pt_n$ in $\SPt^\op$,
which can be parametrized by the odd elements $\zeta\in (\Lambda_n)_1$. 
Explicitly, on the generators $\Phi_{\bbM/\pt_n}(\1\otimes F) = \1\otimes_{\bbC} [F]\in 
\Lambda_n^\bbC\otimes_{\bbC} \AA(\bbM)$, with $F\in \OO_{\cc}(\bbM)$, these 
$\Lambda_n^\bbC\text{-}\SAlg$-automorphisms act as
\begin{flalign}
\nn (\eAA_{\bbM,\bbM})_{\pt_n}(\chi) \big(\Phi_{\bbM/\pt_n}(\1\otimes F) \big) 
&= \Phi_{\bbM/\pt_n} (\1\otimes F) - \Phi_{\bbM/\pt_n}(\zeta\otimes Q(F))\\
&=\Phi_{\bbM/\pt_n} (\1\otimes F) - (\zeta\otimes_{\bbC}\1)\,\Phi_{\bbM/\pt_n}(\1 \otimes Q(F))~,\label{eqn:superexplicit}
\end{flalign}
where $Q := \partial_\theta - \theta\,\partial_t$ is the generator 
of supersymmetry transformations and $\zeta\in(\Lambda_n)_1$. 
On the superalgebra of observables $\AA(\bbM)$ itself, 
these $\Lambda_n^\bbC\text{-}\SAlg$-automorphisms can be 
understood as an odd superderivation $\widehat{Q} : \AA(\bbM)\to\AA(\bbM)$:
On the generators $\Phi_{\bbM}(F) = [F]\in\AA(\bbM)$, with $F\in \OO_{\cc}(\bbM)$, the superderivation 
$\widehat{Q}$ reads as
\begin{subequations}
\begin{flalign}
\widehat{Q}\big(\Phi_{\bbM}(F)\big) = -\Phi_{\bbM}(Q(F))
\end{flalign}
and it satisfies the superderivation property
\begin{flalign}
\widehat{Q}\big(\Phi_{\bbM}(F_1)\,\Phi_{\bbM}(F_2)\big)= \widehat{Q}\big(\Phi_{\bbM}(F_1)\big)\, \Phi_{\bbM}(F_2) 
+ (-1)^{\vert F_1\vert} \,\Phi_{\bbM}(F_1) \, \widehat{Q}\big(\Phi_{\bbM}(F_2)\big)~,
\end{flalign}
\end{subequations}
for all homogeneous $F_1,F_2\in\OO_{\cc}(\bbM)$, as a consequence of $(\eAA_{\bbM,\bbM})_{\pt_n}(\chi)$ being
a $\Lambda_{n}^\bbC\text{-}\SAlg$-morphism.
We may decompose $\Phi_{\bbM}$ into its component quantum fields 
(careful: neither $\psi_{\bbM}$ nor $\phi_{\bbM}$ are natural in the enriched setting) via
\begin{flalign}
\Phi_{\bbM}(F) = \psi_{\bbM}(f) + \phi_{\bbM}(h)~,
\end{flalign}
for all $F=f+\theta\,h\in\OO_{\cc}(\bbM)$, from which we recover 
the usual supersymmetry transformations
\begin{flalign}
\widehat{Q}\big(\psi_{\bbM}(f)\big) = \phi_{\bbM}(\partial_t f)~~,\qquad \widehat{Q}\big(\phi_{\bbM}(h)\big)= - \psi_{\bbM}(h)~, 
\end{flalign}
for all $f,h\in C^\infty(t_0,t_1)$.

\subsection{\texorpdfstring{$3\vert 2$}{3,2}-dimensions}
Our second example is the free Wess-Zumino model in $3\vert 2$-dimensions. 
There are two reasons why we prefer to analyse this toy-model instead of the physically 
more relevant Wess-Zumino model in $4\vert 4$-dimensions: First, in 
$3\vert 2$-dimensions there are only 2 odd dimensions (instead of 4 in $4\vert 4$-dimensions), 
which considerably simplifies the component analysis of the super-Cartan structures and superfields 
discussed below. Second, the Wess-Zumino model in $4\vert 4$-dimensions requires
a chirality constraint, the implementation of which would lead to additional technical/computational 
complications that we would like to postpone to future work.
\sk

We provide all the data listed in Definition \ref{defi:sft}.

\paragraph{Representation theoretic data:} We take $W=\bbR^3$ together with the 
standard  Lorentz metric $g=\mathrm{diag}(1,-1,-1)$.
The corresponding spin group is $\mathrm{Spin}(W,g) \simeq \mathrm{SL}(2,\bbR)$
and we take $\rho^{S} : \mathrm{Spin}(W,g) \times S\to S$ the defining
representation of $\mathrm{SL}(2,\bbR)$ on $S=\bbR^2$.
Notice that $\mathrm{SL}(2,\bbR)$ is the two-fold cover of the identity component 
$\mathrm{SO}_0(1,2)$ of the special pseudo-orthogonal group 
and that $\rho^{W} :\mathrm{Spin}(W,g) \times W\to W$ is given by composing the covering map with
the defining representation of $\mathrm{SO}_0(1,2)$.
\sk

The standard basis $\{p_\alpha\}$ of $W=\bbR^3$, with $\alpha =0,1,2$, is an orthonormal basis
of $(W,g)$. We denote the metric coefficients in this basis by $g_{\alpha\beta} := g(p_\alpha,p_\beta)$ and
the coefficients of the inverse metric in the dual basis $\{ p^\alpha\}$ by $g^{\alpha\beta}$.
Elements $w\in W$ will be indicated by $w = w^\alpha\, p_\alpha$, with summation over repeated indices understood.
 For the coefficients of the dual vector $g(w,\,\cdot\,)=w^\alpha\, g_{\alpha\beta}\, p^\beta \in W^\ast$ we shall 
also write $w_\beta = w^\alpha \, g_{\alpha\beta}$. Notice that  $w^\alpha = w_\beta\,g^{\beta\alpha}$.
The choice of basis $\{p_\alpha\}$ fixes an orientation on $W$.
\sk

The representation $\rho^{S}$ induces up to $\mathrm{SL}(2,\bbR)$-equivalence 
a unique representation of the Clifford algebra $\mathrm{Cl}(W,g)$ in terms of purely 
imaginary matrices $\gamma_\alpha$, with $\alpha = 0,1,2$, on the complexification $S_\bbC$ of 
$S$. These matrices thus satisfy the Clifford relations
\begin{flalign}
\gamma_\alpha \, \gamma_\beta + \gamma_\beta\, \gamma_\alpha = 2 \, g_{\alpha\beta}\, \mathrm{id}_{S_\bbC}~.
\end{flalign}
Furthermore, there exists a charge conjugation matrix $\mathsf{C}$ which satisfies
\begin{flalign}
\mathsf{C}^\mathrm{T}=-\mathsf{C}~,\qquad \gamma^\mathrm{T}_\alpha = - \mathsf{C} \, \gamma_\alpha\,  \mathsf{C}^{-1}~,
\end{flalign}
where $^{\mathrm{T}}$ denotes matrix transposition.
A possible representation is given in terms of the Pauli matrices by 
$\gamma_0 = \sigma_2$, $\gamma_1 = i\, \sigma_1$, $\gamma_2=i\, \sigma_3$ and
$\mathsf{C}=\gamma_0$. We define the antisymmetrized product 
\begin{flalign}
\gamma_{\alpha \beta} := \frac{1}{2} \big(\gamma_\alpha \gamma_\beta - \gamma_\beta\gamma_\alpha\big)
\end{flalign}
and note the identities
\begin{subequations}
\begin{gather}
\left(\gamma_\alpha\, \mathsf{C}^{-1}\right)^\mathrm{T}= \gamma_\alpha \,\mathsf{C}^{-1}~,\qquad 
\left(\gamma_{\alpha\beta} \,\mathsf{C}^{-1}\right)^\mathrm{T}= \gamma_{\alpha\beta}\, \mathsf{C}^{-1}~,
\qquad \mathrm{Tr}(\gamma_\alpha)=\mathrm{Tr}(\gamma_{\alpha\beta})=0~, \\
\gamma_{\alpha\beta} = i \,\epsilon_{\alpha\beta\delta} \,\gamma^\delta~, \qquad 
L = \frac{1}{2} \mathrm{Tr}(L)\,\mathrm{id}_{S_\bbC} + \frac{1}{2} \mathrm{Tr}(L\, \gamma_\alpha)\,\gamma^\alpha~,
\end{gather}
\end{subequations}
where $\epsilon_{\alpha\beta\delta}$ is totally antisymmetric with $\epsilon_{012}=1$ and 
$L$ is an arbitrary endomorphism of $S_{\bbC}$.
\sk

As the pairing $\Gamma$ we shall take  
\begin{flalign}
\Gamma : S \otimes S \longrightarrow W~,~~s_1\otimes s_2\longmapsto 
(s_1, \gamma^\alpha\, \mathsf{C}^{-1}\, s_2 ) ~p_\alpha ~,
\end{flalign}
where $(\,\cdot\,, \,\cdot\, )$ is the standard inner product on $S=\bbR^2$.
(Notice that $\gamma^\alpha\, \mathsf{C}^{-1}$ is a real matrix.) 
The pairing $\Gamma$ is symmetric, $\mathrm{Spin}(W,g)$-equivariant 
and positive with respect to the forward light cone
$C = \{w\in W : g(w,w)>0~\text{and}~ w^0>0 \}$.
For $\epsilon$ we take 
\begin{flalign}
\epsilon : S \otimes S \longrightarrow \bbR~,~~s_1\otimes s_2\longmapsto ( s_1, i\, \mathsf{C}^{-1} s_2)~,
\end{flalign}
which is $\mathrm{Spin}(W,g)$-invariant, antisymmetric and non-degenerate.
We choose a symplectic basis $\{q_a\}$ of $S=\bbR^2$, with $a = 1,2$,
and use the index notation $\epsilon_{ab} := \epsilon(q_a,q_b)$ with $\epsilon_{12} = -\epsilon_{21}=1$.
We further set $\epsilon^{ab} := \epsilon_{ab}$  for the coefficients of the
 symplectic structure on the dual vector space $S^\ast$ in the dual basis $\{q^a\}$ and notice that
$\epsilon^{ab}\,\epsilon_{bc} = -\delta^a_c$. Elements $s\in S$ will be indicated by
$s=s^a\,q_a$, with summation over repeated indices understood.
For the coefficients of the dual spinor $\epsilon(s,\,\cdot\,)= s^a\,\epsilon_{ab}\,q^b \in S^\ast$ 
we shall also write $s_b = s^a\,\epsilon_{ab}$.  Notice that  $s^a =  s_b\,\epsilon^{ab}$.
The choice of basis $\{q_a\}$ fixes an orientation on $S$. 
The two pairings $\Gamma$ and $\epsilon$ read in our bases as
\begin{flalign}
\Gamma(s_1,s_2)=-s^a_1 \, s^b_2 \,i\,\gamma^\alpha_{ab}\, p_\alpha ~,\qquad \epsilon(s_1,s_2)=s^a_1 \, s^b_2 \,\epsilon_{ab}~.
\end{flalign}
In order to state explicitly the super-Lie bracket relations in the super-Poincar{\'e} super-Lie algebra 
$\mathfrak{sp}=(\mathfrak{spin} \oplus W) \oplus S$ corresponding to this choice of data, 
we recall that the Lie algebra $\mathfrak{spin}$ may be spanned by generators $L_{\alpha\beta}=-L_{\beta\alpha}$,  with
$\alpha, \beta= 0,1,2$, and that the Lie algebra actions induced by $\rho^W$ and $\rho^S$ read as
\begin{subequations}
\begin{flalign}
\rho_\ast ^W:\mathfrak{spin}\otimes W&\longrightarrow W ~,~~ L_{\alpha\beta}\otimes p_\gamma\longmapsto 
g_{\alpha\gamma}\, p_\beta - g_{\beta\gamma}\,p_\alpha~,\\
\rho_\ast^S : \mathfrak{spin}\otimes S&\longrightarrow  S~,~~ L_{\alpha\beta}\otimes q_a \longmapsto 
\frac{1}{2} {\gamma_{\alpha\beta}}_a^{\phantom{a}b}\,q_b~.
\end{flalign}
\end{subequations}
Thus, the super-Lie bracket relations in $\mathfrak{sp}$ read in the adapted basis $\{p_\alpha,q_a\}$ as
\begin{subequations}
\begin{gather}
\left[L_{\alpha\beta},L_{\gamma\delta}\right]=g_{\beta\gamma}\, L_{\alpha\delta}+g_{\alpha\delta}\, L_{\beta\gamma}-g_{\alpha\gamma}\, L_{\beta\delta}-g_{\beta\delta}\, L_{\alpha\gamma}~,\\
\left[L_{\alpha\beta},p_\gamma\right]=g_{\alpha\gamma} \, p_\beta-g_{\beta\gamma}\,  p_\alpha~,\qquad \left[L_{\alpha\beta},q_a\right]=\frac{1}{2} {\gamma_{\alpha\beta}}_a^{\phantom{a}b}\, q_b~,\\
[p_\alpha,p_\beta]=[p_\alpha,q_a]=0~,\qquad [q_a,q_b]=2\,i\, \gamma^\alpha_{ab}\,  p_\alpha~.
\end{gather}
\end{subequations}

\paragraph{The objects in $\ghSCart$:} We characterize explicitly the objects $\bbM = (M,\Omega,E)$
in the category $\ghSCart$ for the case of  $3\vert 2$-dimensions under the following simplifying assumptions: 
As in the $1\vert 1$-dimensional case, we assume that the underlying topological spaces $\und{M}$ are connected
and furthermore that the structure sheaf $\OO_M$ is {\it globally} isomorphic
to $C^\infty_{\und{M}}\otimes \bigwedge^\bullet \bbR^2$.
Due to the latter assumption, there exist global odd coordinate functions $\theta^a$, with $a=1,2$, and we set
(with abuse of notation) 
\begin{flalign}
\theta^2:= -\epsilon_{ab}\, \theta^a\, \theta^b = -2\, \theta^1\, \theta^2~.
\end{flalign}
Notice that
\begin{flalign}
\theta^a\, \theta^b = -\frac{1}{2} \epsilon^{ab}\, \theta^2~.
\end{flalign}
The most general even supervielbein $E=e^\alpha\otimes   p_\alpha + \xi^a  \otimes q_a\in
\Omega^1(M,\mathfrak{t})$  (summation over
repeated indices understood) is given by
\begin{subequations}
\begin{flalign}
e^\alpha &= {\und{e}}^\beta\left(\delta^{\phantom{\alpha}\alpha}_\beta+j^{\phantom{\alpha}\alpha}_\beta \, \frac{\theta^2}{2}\right)-
\dd\theta^b\, h^\alpha_{bc}\, \theta^c~,\\
\xi^a &= {\und{e}}^\beta\, l^{\phantom{\alpha a}a}_{\beta\,b}\,\theta^b + \dd \theta^b \, \left(\delta^{\phantom{a}a}_b+k^{\phantom{a}a}_b\, \frac{\theta^2}{2}\right)~,
\end{flalign}
\end{subequations}
where ${\und{E}}={\und{e}}^\alpha \otimes p_\alpha$ is the vielbein on the reduced 
Lorentz manifold $\und{M}$ and $j^{\phantom{\alpha}\alpha}_\beta, 
h^\alpha_{bc}, l^{\phantom{\alpha a}a}_{\beta\,b}, k^{\phantom{a}a}_b\in C^\infty(\und{M})$ 
are coefficient functions. Notice that we have chosen the odd coordinates $\theta^b$
in such a way that the coefficient of  $\dd\theta^b$ in $\xi^a$ has a very simple form. This is always possible 
due to the assumption of non-degeneracy of $E$. Similar to the $1\vert 1$-dimensional case, we call such odd 
coordinates {\em geometric coordinates}. We introduce the dual Lorentz vielbein 
${\und{e}}_\alpha$ by the duality condition $\langle{\und{e}}_\alpha, {\und{e}}^\beta\rangle= 
\delta^{\phantom{\alpha}\beta}_\alpha$ and the dual superderivations $\partial_a$  by 
$\langle \partial_a,\dd\theta^b\rangle = \delta^{\phantom{a}b}_a$. Notice 
that $\langle\partial_a,{\und{e}}^\alpha \rangle =\langle  {\und{e}}_\alpha, \dd\theta^a\rangle = 0 $. 
Using these dual superderivations, we can write the inverse supervielbein as 
\begin{subequations}
\begin{flalign}
e_\alpha &= \left(\delta^{\phantom{\alpha}\beta}_\alpha+J^{\phantom{\alpha}\beta}_\alpha\,  \frac{\theta^2}{2}\right){\und{e}}_\beta- l^{\phantom{\alpha a}b}_{\alpha\,c}\, \theta^c \, \partial_b~,\quad\text{with }~~ J^{\phantom{\alpha}\beta}_\alpha:=-j^{\phantom{\alpha}\beta}_\alpha+l^{\phantom{\alpha a}b}_{\alpha\,c}\, h^{\beta\,c}_{b}~,\\
\xi_a &= h^\beta_{ab}\, \theta^b\, {\und{e}}_\beta + \left(\delta^{\phantom{a}b}_a+K^{\phantom{a}b}_a\, \frac{\theta^2}{2}\right)\partial_b~,\quad \text{with} ~~  K^{\phantom{a}b}_a:=-k^{\phantom{a}b}_a-h^{\beta\,c}_{a}\, l^{\phantom{\alpha a}b}_{\beta\,c}~.
\end{flalign}
\end{subequations}
The duality relations $\langle e_\alpha, e^\beta\rangle= 
\delta^{\phantom{\alpha}\beta}_\alpha$, $\langle \xi_a,\xi^b\rangle =
 \delta^{\phantom{a}b}_a$ and $\langle\xi_a,e^\alpha \rangle =\langle  e_\alpha, \xi^a\rangle = 0 $
 hold true. For the super-spin connection,
 the most general even $\Omega\in\Omega^1(M,\mathfrak{spin})$  can be expanded as
\begin{flalign}
\Omega = {\und{e}}^\alpha\left(\omega_\alpha+\lambda_\alpha \, \frac{\theta^2}{2}\right)+\dd\theta^a \, \phi_{ab}\, \theta^b~,
\end{flalign}
where $\omega_\alpha, \lambda_\alpha, \phi_{ab}\in C^\infty(\und{M},\mathfrak{spin})$ are coefficient functions
with values in the Lie algebra $\mathfrak{spin}$.

\paragraph{The category $\SLoc$:} We define a full subcategory $\SLoc$ of $\ghSCart$ 
by the conditions that 1.)\ the underlying topological spaces are connected, 2.)\ the structure sheaves
are globally isomorphic to $C^\infty_{\und{M}}\otimes\bigwedge^\bullet\bbR^2$ and 3.)\ 
the supergravity supertorsion constraints given in \cite[Eqns.\ (10) and (11)]{Wess:1977fn} hold true. 
In order to discuss the latter constraints, we consider an arbitrary (local) coordinate system $x^\mu$ on $\und{M}$, 
with $\mu = 0,1,2$, and use the notation $X^M$ for the combined super-coordinate system $\{x^\mu,\theta^m\}$, 
where $\theta^m$ are global geometric odd coordinates. (With abuse of notation we will denote the indices
on $X^M$ by the same symbol as the supermanifold.) 
We set $\vert M\vert:=\vert X^M\vert$ for the $\bbZ_2$-parity of $X^M$.
Analogously, we use the notation $P_A$ for the combined generators $\{p_\alpha,q_a\}$ of the 
supertranslation super-Lie algebra $\mathfrak{t} = W\oplus S$ and set $\vert A\vert : =\vert P_A\vert$. Consequently, 
we can expand  the supervielbein as $E=E^A \otimes P^{\phantom{B}}_A = \dd X^M \, 
E^A_M \otimes P^{\phantom{B}}_A$ and its inverse as $E^{\phantom{B}}_A = E^M_A \, \partial^{\phantom{B}}_M$. Using this notation, 
we may expand and compute the supertorsion (\ref{eqn:supertorsion}) 
\begin{subequations}
\begin{flalign}
T_{\bbM}=T^A \otimes  P^{\phantom{B}}_A = \frac{1}{2} E^B \wedge E^C\,  T_{BC}^{\phantom{BC}A} \otimes P^{\phantom{B}}_A~,
\end{flalign}
where
\begin{flalign}
T_{BC}^{\phantom{BC}A} = (-1)^{\vert M\vert\,\vert C\vert } E^M_B \, E^N_C \left(\partial^{\phantom{A}}_N \, E^A_M - (-1)^{\vert N\vert \, \vert M\vert}\partial^{\phantom{A}}_M\, E^A_N\right)+\Omega_{BC}^{\phantom{BC}A}-(-1)^{\vert B\vert\,\vert C\vert}\Omega_{CB}^{\phantom{BC}A}~,\label{eqn:torsionCoeff}
\end{flalign}
and
\begin{flalign}
\Omega_{CB}^{\phantom{BC}A}:=\begin{cases}
\langle E^{\phantom{B}}_C,\Omega^{\phantom{A}A}_B \rangle & ~,~~\text{if $(A,B)=(\alpha,\beta)$ or $(A,B)=(a,b)$}~,\\
0 & ~,~~\text{otherwise}~.
\end{cases}
\end{flalign}
\end{subequations}
Here, $\Omega_\alpha^{\phantom{\alpha}\beta}$ and $\Omega_a^{\phantom{a}b}$ are defined, 
for arbitrary $w=w^\alpha \,p_\alpha\in W$ and $s=s^a \, q_a\in S$, as
\begin{flalign}
\rho^W_\ast(\Omega\otimes w) =: w^\alpha\,\Omega_\alpha^{\phantom{\alpha}\beta} \otimes p_\beta~,\qquad 
\rho^S_\ast(\Omega\otimes s)=:s^a\,\Omega_a^{\phantom{a}b} \otimes q_b~.
\end{flalign}
Note that the signs in \eqref{eqn:torsionCoeff} are only correct if $E$ is even.
\sk

The supertorsion constraints introduced in \cite{Wess:1977fn} read as
\begin{flalign}
T_{bc}^{\phantom{bc}\alpha}=2\, i\, \gamma^{\alpha}_{bc}~,\qquad T_{\beta c}^{\phantom{bc}a}=T_{bc}^{\phantom{bc}a}=T_{\beta\gamma}^{\phantom{bc}\alpha}=T_{\beta c}^{\phantom{bc}\alpha}=0~. 
\end{flalign}
It is remarkable that these constraints, together with the requirement that $E$ and $\Omega$ are
even, determine the supervielbein and super-spin connection uniquely in terms of the reduced Lorentz
vielbein $\und{E}$. In particular, demanding that $E$ is even rules out a non-vanishing 
Rarita-Schwinger field (gravitino) in the expansion of $E$. The unique solution to the supertorsion constraints is
\begin{subequations}
\begin{gather}
j_{\alpha}^{\phantom{\alpha}\beta}=-i\, {\gamma^\beta}_{a}^{\phantom{a}b}\, (\omega_\alpha)_{b}^{\phantom{a}a} =\frac{1}{2} \epsilon^{\beta\gamma\delta}\, (\omega_\alpha)_{\gamma\delta}~,\qquad J_{\alpha}^{\phantom{\alpha}\beta}=0~,\\
h^\alpha_{ab}=i\, \gamma^{\alpha}_{ab}~,\qquad {l_\alpha}_{a}^{\phantom{a}b}=-(\omega_\alpha)_{a}^{\phantom{a}b}~,\qquad k_{a}^{\phantom{a}b}=0~,\qquad K_{a}^{\phantom{a}b}=i\, {\gamma^\beta}_{a}^{\phantom{a}c}\, (\omega_\beta)_{c}^{\phantom{a}b}~,\\
(\omega_\beta)_{\gamma}^{\phantom{\alpha}\alpha}-(\omega_\gamma)_{\beta}^{\phantom{\alpha}\alpha}={\und{e}}
_\beta^\mu \, {\und{e}}^\nu_\gamma\, \left(\partial^{\phantom{\alpha}}_\mu {\und{e}}^{\alpha}_\nu-\partial^{\phantom{\alpha}}_\nu {\und{e}}^{\alpha}_\mu\right)~,\qquad \phi_{ab}=0~,\qquad (\lambda_\alpha)_{\beta}^{\phantom{\alpha}\gamma}=
-\mathrm{Ric}_{\alpha\delta}\,\epsilon^{\phantom{\delta}\gamma\delta}_{\beta}~,\label{eqn:solConstraints}
\end{gather}
\end{subequations}
where the first identity in \eqref{eqn:solConstraints} implies that $\omega = {\und{e}}^\alpha\,\omega_\alpha$ 
is the connection one-form of the Levi-Civita connection and $\mathrm{Ric}$ denotes the corresponding 
Ricci curvature tensor, cf.\ \cite[Section 3.4]{Wald}.
\sk

The Berezinian density $\Ber(E)$ has a particularly simple form for objects 
$\bbM = (M,\Omega,E)$ in $\SLoc$. In fact, for a general
\begin{flalign}
F=\varphi+\psi_a \, \theta^a + \eta\, \frac{\theta^2}{2}\in\OO_{\cc}(\bbM)\label{eqn:expansion32}
\end{flalign}
one has
\begin{flalign}
\int_M \Ber(E)\, F = \int_{\und{M}}\vol_{\und{\bbM}}\left(\eta+\varphi\left(j_{\alpha}^{\phantom{\alpha}\alpha}-l^{\phantom{\alpha a}b}_{\alpha\,c}h^{\alpha\,c}_{b}-k^{\phantom{a}a}_a\right)\right)= \int_{\und{M}}\vol_{\und{\bbM}}\, \eta~,
\end{flalign}
where $\vol_{\und{\bbM}}$ is the canonical volume form on the reduced oriented Lorentz manifold. 
Consequently, the pairing \eqref{eqn:pairing} reads as
\begin{flalign}\label{eqn:pairing32}
\ip{F_1}{F_2}_{\bbM} = \int_{\und{M}}\vol_{\und{\bbM}}\,\big(\varphi_1 \, \eta_2 + \varphi_2\,  \eta_1 + {\psi_1}^a \, {\psi_2}_a\big)~,
\end{flalign}
where the expansion \eqref{eqn:expansion32} has been used.

\paragraph{The natural transformation $P: \OO \Rightarrow \OO$:} Given any super-Cartan supermanifold
$\bbM$, we can consider the super-differential operator 
$\mathscr{D}^a_\Omega \circ \mathscr{D}_a : \OO(\bbM)\to \OO(\bbM)$ defined by
\begin{subequations}
\begin{flalign}
\mathscr{D}_a (F) & := \xi_a(F)~,\\
{\mathscr{D}_\Omega}_b \circ \mathscr{D}_a (F) & :=\xi_b\left(\mathscr{D}_a (F)\right) + \left\langle\xi_b,\Omega_a^{\phantom{a}c}\right\rangle\mathscr{D}_c (F)~,\\
\mathscr{D}^a_\Omega \circ \mathscr{D}_a (F) & := \epsilon^{ab}{\mathscr{D}_\Omega}_b \circ \mathscr{D}_a (F)\,,
\end{flalign}
\end{subequations}
for all $F\in \OO(\bbM)$.
For any object $\bbM $ in $\SLoc$, one may compute using the expansion \eqref{eqn:expansion32} 
\begin{flalign}
\nn\mathscr{D}^a_\Omega \circ \mathscr{D}_a(F)  =& -2 \eta\\
\nn& + \left(2\, {h^\alpha}_a^{\phantom{a}b}\, {\und{e}}_\alpha(\psi_b)+K_a^{\phantom{a}b}\, \psi_b+{h^\alpha}_a^{\phantom{a}b}\,(\omega_\alpha)_b^{\phantom{a}c}\, \psi_c + (\phi_a^{\phantom{a}b})_b^{\phantom{a}c}\, \psi_c\right)\theta^a\\
\nn & - \left(2 \, {h^\alpha}_a^{\phantom{a}b}\, {\und{e}}_\alpha\left({h^\beta}_b^{\phantom{a}a}\, {\und{e}}_\beta(\varphi)
\right)+\left(K_d^{\phantom{d}c}+{h^\beta}_d^{\phantom{a}a}\, (\omega_\beta)_a^{\phantom{a}c}+(\phi_d^{\phantom{a}a})_a^{\phantom{a}c}\right){h^\alpha}_c^{\phantom{a}d}\, {\und{e}}_\alpha(\varphi)\right) \frac{\theta^2}{2}\\
\nn& +\left({h^\beta}_a^{\phantom{a}b}\, (\omega_\beta)_b^{\phantom{a}a}+(\phi_a^{\phantom{a}b})_b^{\phantom{a}a}\right)\eta\, \frac{\theta^2}{2} \\
 =& -2\, \eta + 2\, (i\displaystyle{\not}\nabla \psi)_a\, \theta^a +  2\, \Box \varphi\, \frac{\theta^2}{2}~,
 \end{flalign}
 where $i\displaystyle{\not}\nabla$ is the geometric Dirac operator and $\Box$ is the geometric 
 d'Alembert operator. Here we used the essential identity
 \begin{flalign}
\sigma:= {\gamma^\beta}_a^{\phantom{a}b}\, (\omega_\beta)_b^{\phantom{a}a}=\frac{i}{2}\,(\omega_\alpha)_{\beta\gamma}\, \epsilon^{\alpha\beta\gamma}=0~,
 \end{flalign}
 which holds because for each point $x\in\und{M}$ one can find a vielbein ${\und{E}}$ such 
 that $(\omega_\alpha)_{\beta\gamma}(x)=0$ and thus $\sigma(x)=0$. However, $\sigma$ 
 is a scalar invariant and thus $\sigma\equiv 0$ independent of the chosen vielbein.
\sk
 
For any object $\bbM$ in $\SLoc$ and any constant $m\geq 0$, we define the super-differential operator
\begin{flalign}
P_{\bbM} := \frac{1}{2} \mathscr{D}^a_\Omega\circ  \mathscr{D}_a + m \,\id_{\OO(\bbM)}: \OO(\bbM)\longrightarrow \OO(\bbM)~.\label{eqn:wave32}
\end{flalign}
Solutions of $P_{\bbM} (F) = 0$ satisfy
\begin{flalign}
-\eta+m\, \varphi=0~,\qquad i\displaystyle{\not}\nabla \psi_a + m\, \psi_a = 0~,\qquad  \Box \varphi + m \,\eta = 0~,
\end{flalign}
and thus in particular $\psi_a$ is a solution of the Dirac equation and $\varphi$ is a solution 
of the Klein-Gordon equation $\Box\varphi + m^2\varphi = 0$. 
The super-differential operator $P_{\bbM}$ is 
formally super-self adjoint with respect to the pairing \eqref{eqn:pairing32}. In general, one may show 
for arbitrary homogeneous $F_1\in\OO(\bbM)$ and arbitrary $F_2\in\OO_\cc(\bbM)$ that (see e.g.\ \cite[Section 5.2.8.]{Buchbinder})
 \begin{flalign}
\int_M \Ber(E) \,F_1  \,\mathscr{D}^a_\Omega\circ \mathscr{D}_a (F_2) = (-1)^{1+\vert F_1\vert}\int_M \Ber(E)\, \mathscr{D}^a (F_1)  \, \mathscr{D}_a (F_2)  
 \end{flalign}
if $ T_{a\beta}^{\phantom{bc}\beta}-T_{ab}^{\phantom{bc}b}=0$ 
and thus not all supertorsion constraints are necessary for the formal super-self adjointness 
of $P_{\bbM}$. Notice that $P_{\bbM}$ is also super-Green's hyperbolic: If we write $P_{\bbM}$ in matrix form as
 \begin{flalign}
 P_{\bbM} = \begin{pmatrix}
  m & 0 & -1\\
 0 & i\displaystyle{\not}\nabla + m & 0\\
 \Box & 0 & m
 \end{pmatrix}~,
 \end{flalign}
then the retarded/advanced super-Green's operators $G^\pm_{\bbM}$ for $P_{\bbM}$ can be written as
 \begin{flalign}
 G^\pm_{\bbM} = \begin{pmatrix}
  m \, G_{\Box+m^2}^\pm & 0 & G_{\Box+m^2}^\pm\\
 0 & G_{i{\not}\nabla + m }^\pm& 0\\
 -\Box \circ G_{\Box+m^2}^\pm & 0 & m \,G_{\Box+m^2}^\pm
 \end{pmatrix}~,
 \end{flalign}
 where $G_{\Box+m^2}^\pm$ and $G_{i{\not}\nabla + m }^\pm$ are the retarded/advanced Green's operators 
 for the component differential operators $\Box+m^2$ and $i\displaystyle{\not}\nabla + m$, respectively.
  Finally, the super-differential  operators \eqref{eqn:wave32} are the components of a natural transformation
since they are constructed geometrically in terms of the supervielbein $E$ and a constant $m\geq 0$. 
Hence, we have constructed an example of a 
super-field theory according to Definition \ref{defi:sft}. Applying Theorem \ref{theo:LCQFT} 
we further obtain a super-QFT, which in the present case describes the quantized 
free Wess-Zumino model in $3\vert 2$-dimensions.

\paragraph{Enriched super-field theory:} We shall now discuss the super-field theory
defined above in the enriched setting. We consider two objects  
$\bbM$ and $\bbM^\prime$ in $\eSLoc^\op$ and
any object $\pt_n$ in $\SPt^\op$. Before discussing the set $\eSLoc^\op(\bbM^\prime,\bbM)(\pt_n)$ 
in more detail, we remark that, since any $\chi\in\eSLoc^\op(\bbM^\prime,\bbM)(\pt_n)$ preserves by definition 
the $\pt_n$-relative supervielbeins, and $P_{\bbM}$ is constructed geometrically, 
the super-field theory discussed in this example automatically satisfies the axioms of an enriched super-field theory
given in Definition \ref{defi:sftenriched}.
\sk
 
Instead of fully characterizing the set $\eSLoc^\op(\bbM^\prime,\bbM)(\pt_n)$, we aim 
for analyzing a presumably large subset which contains supersymmetry transformations 
by considering a well-motivated ansatz. For a generic $F^\prime \in\OO_{\cc}(\bbM^\prime)$ expanded 
as in \eqref{eqn:expansion32}, we consider $\SMan/\pt_n$-morphisms 
$\chi : M/\pt_n\to M^\prime/\pt_n$ of the form 
\begin{subequations}
\begin{flalign}
& \chi^\ast(\1\otimes F^\prime) = \1 \otimes F + \zeta \otimes Q(F)~,\\
&F:=\und{\chi}^\ast(\varphi^\prime)+\und{\chi}^\ast(\psi^\prime_a)\,\theta^a
+\und{\chi}^\ast(\eta^\prime)\, \frac{\theta^2}{2}~,
\end{flalign}
\end{subequations}
where $\zeta\in\Lambda_n$ is odd, $\und{\chi} : \und{M}\to\und{M^\prime}$ is a smooth map 
which preserves the reduced Lorentz vielbein $\und{\chi}^\ast(\und{E^\prime})=\und{E}$ and $Q$ is an odd superderivation. 
A necessary condition for such $\chi$ to be an element of $\eSLoc^\op(\bbM^\prime,\bbM)(\pt_n)$ 
is that $Q$ supercommutes with the dual supervielbein, i.e.\
\begin{flalign}\label{eqn:restrictionQ}
\left[Q,\xi_a\right]=0~,\qquad \left[Q,e_\alpha\right]=0~.
\end{flalign}
We can  expand $Q$ as
\begin{gather}
Q = A^\alpha_b\,  \theta^b \, \und{e}_\alpha + \left(B^a + C^a \, \frac{\theta^2}{2}\right)\partial_a 
\end{gather}
and evaluate the constraints \eqref{eqn:restrictionQ} in the given order. After a straightforward 
computation we find that these constraints are satisfied if and only if
\begin{gather}
B^c \, \left( \omega_{\alpha}\right)_{c}^{\phantom{b}a}=0~, \qquad \und{e}_\alpha(B^c)=0~, \qquad A^\alpha_b = - i\, B^c\, \gamma^\alpha_{bc}~,\qquad C^b = 0~.
\end{gather}
Thus, we find that a non-zero $Q$ is only possible if the Levi-Civita connection $\omega$ is 
vanishing, and hence also the super-spin connection $\Omega$ is vanishing on $\bbM$. In this case 
$\chi$ defined as above is an element of $\eSLoc^\op(\bbM^\prime,\bbM)(\pt_n)$ if and only 
if the super-spin connection $\Omega^\prime$ on $\bbM^\prime$ is also vanishing.
\sk

This rather restrictive condition for the existence of interesting enriched morphisms 
$\chi\in\eSLoc^\op(\bbM^\prime,\bbM)(\pt_n)$ originates from our requirement that the supervielbein 
$E$ is even and that $\chi$ must preserve the $\pt_n$-relative supervielbein. In the treatment of 
supergravity one usually deals with supervielbeins which are not purely even and considers, in the terminology 
of this paper, enriched morphisms which have to preserve the $\pt_n$-relative supervielbein and 
$\pt_n$-relative super-spin connection only up to local Lorentz transformations. This class of enriched 
morphisms contains the so-called supergravity transformations \cite{Wess:1992cp}.

\paragraph{Supersymmetry transformations in the enriched super-QFT:} We close the discussion 
of this example in analogy to the $1\vert 1$-dimensional case by illustrating the structure of 
supersymmetry transformations. As discussed above, these transformations appear only if we consider 
an object $\bbM=(M,\Omega,E)$ in $\eSLoc$ with $\Omega=0$, such as for example
the $3\vert 2$-dimensional super-Minkowski space. Given any such object, we can 
use the $\SSet$-functor $\eAA : \eSLoc \to \eSAlg$ constructed in Theorem \ref{theo:enLCQFT} to obtain a 
superalgebra of observables $\AA(\bbM)$ and a $\SSet$-morphism 
$\eAA_{\bbM,\bbM} : \eSLoc(\bbM,\bbM)\to \eSAlg(\AA(\bbM),\AA(\bbM))$, which describes the enriched
automorphism group of $\AA(\bbM)$. Proper supersymmetry transformations are described by the odd superderivations
\begin{flalign}
Q_B := B^a \, \partial_a - i \, B^a \, \gamma^\alpha_{ab}\, \theta^b\, \und{e}_\alpha~,
\end{flalign}
with $B^a$ a constant spinor, and they are parametrized by odd elements $\zeta\in (\Lambda_n)_1$. As in the 
$1\vert 1$-dimensional case, such supersymmetry transformations may be understood as odd superderivations 
$\widehat{Q}_B : \AA(\bbM)\to\AA(\bbM)$ which act on the generators $\Phi_{\bbM}(F) = [F]\in\AA(\bbM)$, with 
$F\in\OO_{\cc}(\bbM)$, as
\begin{flalign}
\widehat{Q_B}\big(\Phi_{\bbM}(F)\big) = -\Phi_{\bbM}(Q_B(F))~.
\end{flalign} 
We may decompose $\Phi_{\bbM}(F)$ into its component quantum fields by using the expansion
\begin{flalign}
 F = f + \rho_a \,\theta^a + h\,\frac{\theta^2}{2} ~.
 \end{flalign}
 Explicitly, we set 
\begin{flalign}
\Phi_{\bbM}(F) = \phi_{\bbM}(f) + \psi^a_{\bbM}(\rho_a) + \eta_{\bbM}(h)~,
\end{flalign}
and we recover the usual supersymmetry transformations
\begin{subequations}
\begin{flalign}
\widehat{Q_B}\big(\phi_{\bbM}(f)\big) &= \psi^a_{\bbM}(f\, B_a)~,\\
\widehat{Q}\big(\psi^a_{\bbM}(\rho_a)\big)& = \phi_{\bbM}(B^a \, i\displaystyle{\not}\nabla \rho_a)+\eta_{\bbM}(B^a\,  \rho_a)~,\\
\widehat{Q_B}\big(\eta_{\bbM}(h)\big) &= -\psi^a_{\bbM}(i\displaystyle{\not}\nabla (h \,B_a))~.
\end{flalign}
\end{subequations}

%%%%%%%%%%%%%%%%%%%%%%%%%%%%%%%%%%%%%%%%%%%%%%%%
%%%%%%%%%%%%%%%%%%%%%%%%%%%%%%%%%%%%%%%%%%%%%%%%

\section*{Acknowledgements}
We would like to thank the referees for valuable comments and suggestions.
We thank Chris Fewster for useful comments on this work.
A.S.\ also thanks Sven Meinhardt for many general discussions on supergeometry.
The work of T.-P.H.\ and A.S.\ is supported by a Research Fellowship of Deutsche Forschungsgemeinschaft (DFG) 
and F.H.\ is supported by the SFB 647:\emph{Raum-Zeit-Materie} funded by the DFG. 
Furthermore, T.-P.H.\ and A.S.\ would like to thank Mathematisches Forschungsinstitut Oberwolfach (MFO) 
for the support under the program Research in Pairs and the great hospitality at the institute.

\appendix

\section{\label{app:enriched}Basics of enriched category theory}
We review some elementary definitions of enriched category theory which will be used in our work. 
For detailed introductions to this subject see e.g.\ \cite{Enriched} and \cite{Enriched2}.
\sk

Let $\mathsf{V}$ be a monoidal category. For our purposes we can 
assume that the associator in $\mathsf{V}$ is trivial.
We denote the monoidal bifunctor by $\otimes : \mathsf{V}\times \mathsf{V}\to\mathsf{V}$ 
and the unit object in $\mathsf{V}$ by $I$.
\begin{defi}\label{defi:enrichedcategory}
A $\mathsf{V}$-category (or a category enriched over $\mathsf{V}$)
$\mathsf{C}$ consists of
\begin{itemize}
\item a class $\mathrm{Ob}(\mathsf{C})$ of objects;
\item for all objects $A,B\in \mathrm{Ob}(\mathsf{C})$, an object $\mathsf{C}(A,B)$ in $\mathsf{V}$ called
the ``object of morphisms from $A$ to $B$'';
\item for all objects $A,B,C\in \mathrm{Ob}(\mathsf{C})$, a morphism 
$\bullet_{A,B,C} : \mathsf{C}(B,C)\otimes \mathsf{C}(A,B)\to \mathsf{C}(A,C)$
in $\mathsf{V}$ called the ``composition'';
\item for every object $A\in \mathrm{Ob}(\mathsf{C})$, a morphism
$\oone_{A} : I \to \mathsf{C}(A,A)$ in $\mathsf{V}$ called the ``identity on $A$''.
\end{itemize}
This data must satisfy the associativity and unit axioms, which are expressed by commutativity of the diagrams
\begin{subequations}
\begin{flalign}
\xymatrix{
\ar[d]_-{\bullet_{B,C,D}\,\otimes\,\id_{\mathsf{C}(A,B)}}\mathsf{C}(C,D)\otimes\mathsf{C}(B,C)\otimes\mathsf{C}(A,B)\ar[rrr]^-{\id_{\mathsf{C}(C,D)} \,\otimes\, \bullet_{A,B,C} } &&&  \mathsf{C}(C,D)\otimes\mathsf{C}(A,C)\ar[d]^-{\bullet_{A,C,D}}\\
\mathsf{C}(B,D)\otimes\mathsf{C}(A,B) \ar[rrr]_-{\bullet_{A,B,D}} &&&\mathsf{C}(A,D)
}
\end{flalign}
\begin{flalign}
\xymatrix{
\ar[d]_-{\oone_{B} \,\otimes\, \id_{\mathsf{C}(A,B)}} I\otimes \mathsf{C}(A,B) \ar[rrd]^-{\simeq}&& && \ar[lld]_-{\simeq}\mathsf{C}(A,B)\otimes I \ar[d]^-{\id_{\mathsf{C}(A,B)} \,\otimes\,\oone_{A}} \\
\mathsf{C}(B,B)\otimes \mathsf{C}(A,B)\ar[rr]_-{\bullet_{A,B,B}} && \mathsf{C}(A,B) && \ar[ll]^-{\bullet_{A,A,B}}\mathsf{C}(A,B) \otimes \mathsf{C}(A,A)
}
\end{flalign}
\end{subequations}
in the category $\mathsf{V}$, for all objects $A,B,C,D \in \mathrm{Ob}(\mathsf{C})$.
\end{defi}

\begin{rem}
For ease of notation, we shall always drop the labels on the composition and identity, i.e.\ we simply write
 $\bullet : \mathsf{C}(B,C)\otimes \mathsf{C}(A,B)\to \mathsf{C}(A,C)$ and $\oone :  I \to \mathsf{C}(A,A)$.
\end{rem}

\begin{defi}\label{defi:enrichedfunctor}
Let $\mathsf{C}$ and $\mathsf{D}$ be two $\mathsf{V}$-categories. A $\mathsf{V}$-functor (or enriched functor)
$\mathfrak{F} : \mathsf{C}\to\mathsf{D}$ is given by the following assignment:
\begin{itemize}
\item for every object $A\in\mathrm{Ob}(\mathsf{C})$, an object $\mathfrak{F}(A) \in \mathrm{Ob}(\mathsf{D})$;
\item for all objects $A,B\in \mathrm{Ob}(\mathsf{C})$, a morphism $\mathfrak{F}_{A,B} : \mathsf{C}(A,B) \to \mathsf{D}(\mathfrak{F}(A),\mathfrak{F}(B))$ in $\mathsf{V}$.
\end{itemize}
This assignment must be compatible with the composition and identity, which is expressed by commutativity
of the diagrams
\begin{subequations}
\begin{flalign}
\xymatrix{
\ar[d]_-{\mathfrak{F}_{B,C}\,\otimes\,\mathfrak{F}_{A,B}}\mathsf{C}(B,C) \otimes \mathsf{C}(A,B) \ar[rr]^-{\bullet} && \mathsf{C}(A,C)\ar[d]^-{\mathfrak{F}_{A,C}}\\
\mathsf{D}(\mathfrak{F}(B),\mathfrak{F}(C)) \otimes \mathsf{D}(\mathfrak{F}(A),\mathfrak{F}(B))\ar[rr]_-{\bullet} &&\mathsf{D}(\mathfrak{F}(A),\mathfrak{F}(C))
}
\end{flalign}
\begin{flalign}
\xymatrix{
&& \mathsf{C}(A,A) \ar[dd]^-{\mathfrak{F}_{A,A}}\\
I \ar[rru]^-{\oone} \ar[rrd]_-{\oone}&&\\
  && \mathsf{D}(\mathfrak{F}(A),\mathfrak{F}(A))
}
\end{flalign}
\end{subequations}
in the category $\mathsf{V}$, for all objects $A,B,C \in \mathrm{Ob}(\mathsf{C})$.
\end{defi}

\begin{rem}\label{rem:enrichedcomposition}
Notice that $\mathsf{V}$-functors can be composed:
Consider three  $\mathsf{V}$-categories $\mathsf{C},\mathsf{D},\mathsf{E}$ and
two $\mathsf{V}$-functors  $\mathfrak{F} : \mathsf{C}\to\mathsf{D}$ and $\mathfrak{G} :\mathsf{D}\to\mathsf{E}$.
We define the $\mathsf{V}$-functor $\mathfrak{G}\circ \mathfrak{F} : \mathsf{C}\to\mathsf{E}$ by the following assignment:
To every object $A\in\mathrm{Ob}(\mathsf{C})$ we assign $(\mathfrak{G}\circ\mathfrak{F})(A) 
:= \mathfrak{G}(\mathfrak{F}(A)) \in\mathrm{Ob}(\mathsf{E})$.
To all objects $A,B\in\mathrm{Ob}(\mathsf{C})$ we assign the
$\mathsf{V}$-morphism 
\begin{flalign}
(\mathfrak{G}\circ\mathfrak{F})_{A,B} := \mathfrak{G}_{\mathfrak{F}(A),\mathfrak{F}(B)}\circ \mathfrak{F}_{A,B}  : \mathsf{C}(A,B) \longrightarrow \mathsf{E}(\mathfrak{G}(\mathfrak{F}(A)),\mathfrak{G}(\mathfrak{F}(B)))~,
\end{flalign}
where $\circ$ denotes the composition of morphisms in $\mathsf{V}$.
It is easy to check that $\mathfrak{G}\circ\mathfrak{F} : \mathsf{C}\to\mathsf{E}$ is a $\mathsf{V}$-functor.
\end{rem}

\begin{defi}\label{defi:enrichednatural}
Let $\mathsf{C}$ and $\mathsf{D}$ be two $\mathsf{V}$-categories and
$\mathfrak{F} ,\mathfrak{G}: \mathsf{C}\to\mathsf{D}$ two parallel $\mathsf{V}$-functors.
A $\mathsf{V}$-natural transformation (or enriched natural transformation)
$\eta : \mathfrak{F} \Rightarrow \mathfrak{G}$ is given by assigning to every object $A\in\mathrm{Ob}(\mathsf{C})$
a morphism $\eta_A : I \to \mathsf{D}(\mathfrak{F}(A),\mathfrak{G}(A))$ in $\mathsf{V}$, such that
the diagram
\begin{flalign}
\xymatrix{
&I\otimes \mathsf{C}(A,B) \ar[rr]^-{\eta_{B}\,\otimes\,\mathfrak{F}_{A,B}} && \mathsf{D}(\mathfrak{F}(B),\mathfrak{G}(B))\otimes \mathsf{D}(\mathfrak{F}(A),\mathfrak{F}(B)) \ar[dr]^-{\bullet} &\\
\ar[ur]^-{\simeq}\ar[dr]_-{\simeq} \mathsf{C}(A,B) & &&  & \mathsf{D}(\mathfrak{F}(A),\mathfrak{G}(B))\\
&\mathsf{C}(A,B)\otimes I \ar[rr]_-{\mathfrak{G}_{A,B}\,\otimes\,\eta_A}&&  \mathsf{D}(\mathfrak{G}(A),\mathfrak{G}(B))\otimes \mathsf{D}(\mathfrak{F}(A),\mathfrak{G}(A)) \ar[ur]_-{\bullet} &
}
\end{flalign}
in the category $\mathsf{V}$ commutes, for all objects $A,B\in\mathrm{Ob}(\mathsf{C})$.
\end{defi}

\begin{rem}
Notice that for $\mathsf{V}$ being the monoidal category $\Set$ of sets, with monoidal bifunctor
given by the Cartesian product and unit object given by any singleton set $\pt := \{\star\}$,
all definitions above reduce to the definitions in ordinary category theory.
Hence, category theory enriched over $\Set$ is the same as ordinary category theory.
\end{rem}

%%%%%%%%%%%%%%%%%%%%%%%%%%%%%%%%%%%%%%%%%%%%%%%%
%%%%%%%%%%%%%%%%%%%%%%%%%%%%%%%%%%%%%%%%%%%%%%%%

\end{document}